\documentclass[twoside,11pt]{article}

\usepackage[preprint]{jmlr2e}
\usepackage[T1]{fontenc}    % use 8-bit T1 fonts
\usepackage{lmodern}
\usepackage{hyperref}       % hyperlinks
\usepackage{url}            % simple URL typesetting
\usepackage{booktabs}       % professional-quality tables
\usepackage{amsfonts}       % blackboard math symbols
\usepackage{nicefrac}       % compact symbols for 1/2, etc.
\usepackage{tabularx} % extra features for tabular environment
\usepackage{multirow}
\usepackage{wrapfig}
\usepackage{mathtools}
\usepackage{commath}
\usepackage{stmaryrd}
\usepackage{bbm}
\usepackage[ruled,lined]{algorithm2e}
\usepackage{cases}
\usepackage[toc,page,titletoc,title]{appendix}

\usepackage{filecontents,url}
\usepackage{tikz}
\usetikzlibrary{fit}
\usetikzlibrary{arrows,shapes}  
\usepackage{color}
\usepackage{enumitem}
\newcommand{\subscript}[2]{$#1 _ #2$}

\newcommand{\B}[1]{\mathbf{#1}} 
\newcommand{\pr}[1]{\left(#1\right)} 
\newcommand{\br}[1]{\left[#1\right]} 
\newcommand{\bbr}[1]{\left\{#1\right\}} 
\newcommand{\nr}[1]{\left\|#1\right\|} 
\newcommand{\joint}[1]{\Pi_{\rho}#1} 
\newcommand{\marginal}[1]{\pi_{\rho}#1}

\newcommand{\R}{\mathbb{R}}
\newcommand{\Z}{\mathbb{Z}} 
\newcommand{\btheta}{\boldsymbol{\theta}}

\newcommand{\bz}{\B{z}}
\newcommand{\bZ}{\B{Z}}

\newcommand{\bx}{\B{x}}
\newcommand{\by}{\B{y}}
\newcommand{\F}{\mathcal{F}}
\newcommand{\tA}{\tilde{A}}
\newcommand{\tbz}{\tilde{\bz}}
\newcommand{\grad}{\nabla}
\newcommand{\mrd}{\mathrm{d}}
\renewcommand{\l}{\left}
\renewcommand{\r}{\right}
\newcommand{\E}{\mathbb{E}}
\newcommand{\PP}{\mathbb{P}}
\newcommand{\econst}{\mathrm{e}}

\newcommand{\bthetastar}{\boldsymbol{\theta^\star}}

\newcommand{\inner}[2]{\ensuremath{%
\left\langle#1,#2\right\rangle%
}}

\newcommand{\eqsp}{\vspace{0ex}}
\newcommand{\dd}{\mathrm{d}}

\newcommand{\asszero}{$(A_0)$}
\newcommand{\assone}{$(A_1)$}
\newcommand{\asstwo}{$(A_2)$}
\newcommand{\assfour}{$(A_3)$}
\newcommand{\assfive}{$(A_4)$}
\newcommand{\asssix}{$(A_5)$}
\newcommand{\assseven}{$(A_6)$}
\newtheorem{assumption}{Assumption}
\newcommand\Laplace{\Delta}
\DeclarePairedDelimiterX{\infdivx}[2]{(}{)}{%
  #1\;\delimsize|\delimsize|\;#2%
}
\newcommand{\kld}[2]{\ensuremath{D_{KL}\infdivx{#1}{#2}}\xspace}

\allowdisplaybreaks

\makeatletter
\newcommand{\thickhline}{%
    \noalign {\ifnum 0=`}\fi \hrule height 1pt
    \futurelet \reserved@a \@xhline
}

\usepackage{lastpage}
\jmlrheading{}{}{}{}{}{}{Maxime Vono, Daniel Paulin and Arnaud Doucet}
\ShortHeadings{Efficient MCMC Sampling with Dimension-Free Convergence Rate...}{Vono, Paulin and Doucet}

\begin{document}

\title{Efficient MCMC Sampling with Dimension-Free \\ Convergence Rate using ADMM-type Splitting}

\author{\name Maxime Vono\hspace{0.3mm}\thanks{Both authors contributed equally.} \email maxime.vono@huawei.com \\
       \addr Lagrange Mathematics and Computing Research Center, Huawei\\
       75007 Paris, France
       \AND
       \name Daniel Paulin\hspace{0.3mm}\footnotemark[1] \email paulindani@gmail.com  \\
       \addr School of Mathematics\\
       University of Edinburgh, United Kingdom 
       \AND
       \name Arnaud Doucet \email doucet@stats.ox.ac.uk  \\
       \addr Department of Statistics\\
       University of Oxford, United Kingdom
       }

\editor{}

\maketitle

\begin{abstract}%
    Performing exact Bayesian inference for complex models is computationally intractable. Markov chain Monte Carlo (MCMC) algorithms can provide reliable approximations of the posterior distribution but are expensive for large data sets and high-dimensional models. 
    A standard approach to mitigate this complexity consists in using subsampling techniques or distributing the data across a cluster. 
    However, these approaches are typically unreliable in high-dimensional scenarios. 
    We focus here on a recent alternative class of MCMC schemes exploiting a splitting strategy akin to the one used by the celebrated alternating direction method of multipliers (ADMM) optimization algorithm. 
    These methods appear to provide empirically state-of-the-art performance but their theoretical behavior in high dimension is currently unknown.
    In this paper, we propose a detailed theoretical study of one of these algorithms known as the split Gibbs sampler.
    Under regularity conditions, we establish explicit convergence rates for this scheme using Ricci curvature and coupling ideas. 
    We support our theory with numerical illustrations.
\end{abstract}


%\tableofcontents{}

\section{Introduction}
\label{sec:introdution}
 We are interested in performing Bayesian inference for large data sets and potentially high-dimensional models. For complex models, the posterior distribution is intractable and needs to be approximated. 
 %\textcolor{red}{MV: to be removed. Stochastic Variational Bayes approaches are popular in these scenarios as they are computationally rather cheap \citep{hoffman2013stochastic}. However, there is a lack of theoretical guarantees available for such approaches as the minimization problem one has to solve is typically not convex. 
 %Additionally, variational approximations tend to underestimate uncertainty.
 %This is essentially due to the choice of the Kullback-Leibler divergence as objective function which is known to favor variational distributions more concentrated around regions where the posterior distribution places high mass \citep{MacKay2002,turner_sahani_2011}.}
 To this end, many Markov chain Monte Carlo (MCMC) schemes  have been proposed over the past five years; see for instance \cite{bardenet2017markov} for a recent overview.

These methods can be loosely speaking divided into two groups: subsampling-based techniques and divide-and-conquer approaches. Subsampling-based approaches are MCMC techniques that only require accessing a subsample of the observations at each iteration: these include the popular stochastic gradient Langevin dynamics (SGLD) \citep{Welling2011,dubey2016variance,brosse2018promises,chatterji2018theory,baker2019}, subsampling versions of the Metropolis--Hastings algorithm \citep{Bardenet2014,Korattikara2014,bardenet2017markov,quiroz2014speeding,Cornish2019} and methods based on piecewise-deterministic MCMC schemes  \citep{bouchard2015bouncy,bierkens2016zigzag}. However, all the subsampling methods accessing $\mathcal{O}(1)$ data points at each iteration only provide reliable posterior approximations if they rely on some control variate ideas which require estimating the mode of the posterior and this posterior to be concentrated \citep{Welling2011,dubey2016variance,bardenet2017markov,brosse2018promises,chatterji2018theory,baker2019,Cornish2019}. Practically, as pointed out in \citet{bardenet2017markov,Cornish2019}, this means that such methods are of limited practical interest as they only work well in scenarios where the Bernstein-von Mises approximation of the target is excellent. Divide-and-conquer methods are techniques which consider the common scenario where the data are distributed across a cluster. These schemes run independent MCMC chains to estimate ``local'' posteriors on each node of the cluster and then recombine these ``local'' posteriors to obtain an approximation of the full posterior \citep{Wang2013,Neiswanger2014,Minsker2014,Wang2015,Scott2016,Scott2017,hasenclever2017distributed}. However, these methods often use parametric or kernel density approximations of the local posteriors so as to combine them. This can be unreliable in high-dimensional scenarios; see \cite{bardenet2017markov} and \cite{Rendell2018} for a detailed discussion.

An alternative approach to perform MCMC, amenable to a distributed implementation, has been recently introduced independently in \cite{Vono2019} and \cite{Rendell2018}; see also \cite{dai2012sampling,Chowdhury2018shepherding} and \cite{barbos2017clone} for earlier related ideas. 
It is inspired by the well-known variable splitting technique used in optimization, for instance  by quadratic penalty approaches or  the alternating direction method of multipliers (ADMM), see  \citet{Boyd2011}. 
In the sampling context, this corresponds to defining an artificial hierarchical Bayesian model where the parameter of interest is becoming a ``master'' parameter which is artificially replicated as many times as one ``splits'' the target distribution. In this context, we can then  develop MCMC schemes which alternate sampling the node parameters given the master parameter then the master parameter given the node parameters.  Experimentally, these methods appear promising but it is yet unclear how such schemes behave in high-dimensional scenarios. 
This paper aims at studying theoretically one of these samplers called split Gibbs sampler (SGS).

\textit{Contributions.} Our contributions are as follows.
\begin{itemize}
    \item We present non-asymptotic bounds on the total variation (TV) and 1-Wasserstein distances between the original posterior distribution and the distribution targeted by this class of MCMC schemes. This allows us to quantify the ``bias'' introduced by these methods and significantly sharpens and complements previous results in \cite{Vono2019_AXDA}. 
    \item Using Ricci curvature and coupling techniques, we establish explicit dimension-free convergence rates for SGS.
    Combining our bounds on the bias and convergence rates, we provide mixing time bounds with explicit dependencies with respect to (w.r.t.) the dimension of the problem, its associated condition number and the prescribed precision.
    In both 1-Wasserstein and TV distances, we show that our complexity results are competitive with those recently derived for MCMC schemes based on Langevin or Hamiltonian dynamics. 
    \item We illustrate these theoretical results on several applications, demonstrating the benefits of SGS over state-of-the-art MCMC approaches.
\end{itemize}

\textit{Notations and conventions.} We denote by $\mathcal{B}(\mathbb{R}^d)$ the Borel $\sigma$-field of $\mathbb{R}^d$.
The total variation norm between two probability measures $\mu$ and $\nu$ on $(\mathbb{R}^d,\mathcal{B}(\mathbb{R}^d))$ is defined by 
\begin{equation*}
    \nr{\mu - \nu}_{\mathrm{TV}} = \sup_{f \in \mathbb{M}(\mathbb{R}^d), \nr{f}_{\infty} \leq 1} \left|\int_{\btheta \in \mathbb{R}^d}f(\btheta)\mathrm{d}\mu(\btheta) - \int_{\btheta \in \mathbb{R}^d}f(\btheta)\mathrm{d}\nu(\btheta) \right|,
\end{equation*}
 where $\mathbb{M}(\mathbb{R}^d)$ denotes the set of all Borel measurable functions $f$ on $\mathbb{R}^d$ and $\nr{f}_{\infty} = \sup_{\btheta \in \mathbb{R}^d}|f(\btheta)|$.
Let $\mu$, $\nu$ be two probability measures on $(\R^d,\mathcal{B}(\R^d))$. 
Define the Kullback-Leibler (KL) divergence of $\mu$ from $\nu$ by
\begin{equation*}
D_{\mathrm{KL}} (\mu||\nu) = 
\begin{cases}
  \int_{\R^d} \frac{\mathrm{d} \mu}{\dd \nu}(\btheta)\log\pr{\frac{\mathrm{d} \mu}{\mathrm{d} \nu}(\btheta)}\,\mathrm{d} \nu(\btheta)\eqsp, & \text{if $\mu \ll \nu$}\\
  +\infty & \text{otherwise.}
\end{cases}  
\end{equation*} 
For $1\le p<\infty$, and a metric $w:\mathbb{R}^d\times \mathbb{R}^d\to \R$, the Wasserstein distance of order $p$ between two probability measures $\mu$ and $\nu$ on $(\mathbb{R}^d,\mathcal{B}(\mathbb{R}^d))$ is defined by
\begin{equation*}
    W_p^w(\mu,\nu) = \pr{\inf_{\pi \in \mathcal{U}(\mu,\nu)} \int_{\btheta,\btheta' \in \mathbb{R}^d}w(\btheta,\btheta')^p\mathrm{d}\pi(\btheta,\btheta')}^{1/p}\ ,
\end{equation*}
where $\mathcal{U}(\mu,\nu)$ is the set of all probability measures which admit $\mu$ and $\nu$ as marginals. For $p=\infty$, the Wasserstein distance of order $\infty$ is defined as
\begin{equation*}
    W_{\infty}^w(\mu,\nu) = \inf_{\pi \in \mathcal{U}(\mu,\nu), (X,Y)\sim \pi} \mathrm{ess}\sup w(X,Y) \ .
\end{equation*}
In the case when $w$ is the Euclidean metric, we will denote these by $W_p(\mu,\nu)$. 
For the sake of simplicity, with little abuse, we shall use the same notations for
a probability distribution and its associated probability density function.
For a Markov chain with transition kernel $\B{P}$ on $\mathbb{R}^d$ and invariant distribution $\pi$, we define the $\epsilon$-mixing time associated to a statistical distance $D$, precision $\epsilon > 0$ and initial distribution $\nu$, by
\begin{equation*}
t_{\mathrm{mix}}(\epsilon;\nu) = \min \bbr{t \geq 0 \ \big\vert \ D(\nu\B{P}^t,\pi) \leq \epsilon}\ ,
\end{equation*}
which stands for the minimum number of steps of the Markov chain such that its distribution is at most at an $\epsilon$ $D$-distance from the invariant distribution $\pi$. 
The Euclidean norm on $\mathbb{R}^d$ is denoted by $\nr{\cdot}$.
For $n \geq 1$, we refer to the set of integers between $1$ and $n$ with the notation $[n]$.
The $d$-multidimensional Gaussian probability distribution with mean $\boldsymbol{\mu}$ and covariance matrix $\B{\Sigma}$ is denoted by $\mathcal{N}(\cdot;\boldsymbol{\mu},\B{\Sigma})$.
The parabolic cylinder special function is defined, for all $d>0$ and $z \in \mathbb{R}$, by $D_{-d}(z) = \exp(-z^2/4)\Gamma(d)^{-1}\int_0^{+\infty}e^{-xz - x^2/2}x^{d-1}\mathrm{d}x$, where $\Gamma(\cdot)$ denotes the Gamma function.
For $0\leq i < j$, we use the notation $\B{u}_{i:j}$ to refer to the vector $[\B{u}_i^\top,\B{u}_{i+1}^\top,\hdots,\B{u}_{j}^\top]^\top$ built by stacking $j-i+1$ vectors ($\B{u}_k; k \in \{i,i+1,\hdots,j\}$).
For $f: \mathbb{R}^d \rightarrow \mathbb{R}$ and any  $\btheta \in \mathbb{R}^d$, we use the notations $f(\btheta)_{-} = -\min(f(\btheta),0)$ and $f(\btheta)_{+} = \max(f(\btheta),0)$.

\section{Background and Problem Formulation}
\label{sec:split_Gibbs_sampler}

This section sets up the simulation problem considered in this paper and briefly reviews the approximate Bayesian approach proposed by \cite{Vono2019} and \cite{Rendell2018}.

\subsection{Bayesian Model}
\label{subsec:problem_statement}
    
We consider the situation where one is interested in carrying out Bayesian inference about a parameter $\btheta \in \mathbb{R}^d$ based on observed data $\mathrm{D} = \{\B{x}_j,y_j\}_{j=1}^n$, where for any $j \in [n]$, $\B{x}_j$ are covariates (also called features) associated to observation $y_j$.
    We assume that the number of observations $n$ is large so that the data set $\mathrm{D}$ is partitioned into $S \in [n]$ subsets $\{\mathrm{D}_s\}_{s=1}^S$, called \emph{shards}, such that $\sqcup_{s=1}^S \mathrm{D}_s = \mathrm{D}$.
    Under this framework, the posterior distribution of interest is assumed to admit a density w.r.t. the Lebesgue measure of the form 
    \begin{align}
        \label{eq:target_density_0}
        \pi(\btheta \mid \mathrm{D}) \propto p(\btheta) \prod_{s=1}^S  \pi_s(\mathrm{D}_s \mid \btheta) \ ,
    \end{align}
    where $\{\pi_s(\mathrm{D}_s \mid \btheta)\}_{s=1}^S$ are likelihood functions associated to $\{\mathrm{D}_s\}_{s=1}^S$ and $p(\btheta)$ is the prior density for $\btheta$.
    Contrary to the majority of divide-and-conquer MCMC approaches, we do not assume that the prior factorizes across shards, that is $p(\btheta) \propto \prod_{s=1}^S p_s(\btheta)$.
    In the sequel, it will be convenient to characterize the posterior distribution via potential functions. 
    To this end, we assume that the posterior density defined in \eqref{eq:target_density_0} can be re-written as $\pi(\btheta \mid \mathrm{D}) \propto \exp(-U(\btheta))$ where  
    \begin{align}
        U(\btheta) = \sum_{i=1}^b  U_i(\B{A}_i\btheta)\ ,
    \label{eq:target_density}
    \end{align}
    for $b \in \mathbb{N} \setminus \{0\}$, some matrices $\B{A}_{i} \in \mathbb{R}^{d_i \times d}$ and potential functions $U_{i}:\mathbb{R}^{d_i}\rightarrow \mathbb{R}$ with $i\in[b]$. 
    This definition of the posterior encompasses two main scenarios that are ubiquitous in Bayesian machine learning and illustrated in Examples \ref{example:linar_regression} and \ref{example:logistic}.   
    More precisely, if $p(\btheta) \propto \prod_{s=1}^S p_s(\btheta)$, then $b=S$ and for any $i \in [b]$, $U_i(\B{A}_i\btheta) = - \log p_i(\btheta) - \log \pi_i(\mathrm{D}_i \mid \btheta)$.
    Conversely, if $p(\btheta)$ does not factorize across shards, one can set $b=S+1$ by assigning,  for $s \in [S]$, one potential $U_s$ to each likelihood contribution $\pi_s(\mathrm{D}_s \mid \btheta)$, and one potential $U_{S+1}$ to the prior $p(\btheta)$. 
    In both cases, for any $i \in [b]$, the potential $U_i$ is assumed to be dependent on the subset $\mathrm{D}_i$ of the observations; potentially $\mathrm{D}_i = \{\emptyset\}$ if $U_i$ refers to the prior.
    To simplify notation, this dependence is notationally omitted and we will denote by $\pi(\btheta)$ the posterior distribution in the rest of the paper. 
    We give hereafter two illustrative standard statistical machine learning examples that fit into the considered Bayesian framework.
\begin{example}
\label{example:linar_regression} 
Bayesian ridge linear regression. We consider the model defined by
    \begin{align*}
        &y_j \sim \mathcal{N}(\B{x}_j^\top\btheta,\sigma^2)\ , \quad \forall j \in [n]\ , \\
        &\btheta \sim \mathcal{N}\pr{\B{0}_d,\tau \B{I}_d}\ ,
    \end{align*}
    where $\btheta \in \mathbb{R}^d$ are the unknown regression parameters and $\tau > 0$ is a fixed regularization parameter.
    In this case, the posterior density writes
    $$
    \pi(\btheta) \propto \exp\pr{-\frac{1}{2\tau}\nr{\btheta}^2} \prod_{j=1}^n\exp\pr{-\frac{1}{2\sigma^2}(y_j - \B{x}_j^\top\btheta)^2}\ .
    $$
    By dividing the data set $\mathrm{D} = \{\B{x}_j,y_j\}_{j =1}^n$ into $S$ shards $\{\mathrm{D}_{s}\}_{s=1}^b$,
    the posterior density can be re-written as in \eqref{eq:target_density_0}, that is
    $$
    \pi(\btheta) \propto  \exp\pr{-\frac{1}{2\tau}\nr{\btheta}^2}\prod_{s=1}^S\exp\pr{- \frac{1}{2\sigma^2}\sum_{\{\B{x}_j,y_j\} \in \mathrm{D}_s}(y_j - \B{x}_j^\top\btheta)^2} \ .
    $$
    Under this factorization, one can characterize $\pi(\btheta)$ via \eqref{eq:target_density} by setting $b=S+1$ with the choices $\B{A}_{S+1} = \B{I}_d$, $U_{S+1} = \nr{\btheta}^2/(2\tau)$, and for any $s \in [S]$, $\B{A}_s = \B{I}_d$, $U_s(\btheta) = \sum_{\{\B{x}_j,y_j\} \in \mathrm{D}_s}(y_j - \B{x}_j^\top\btheta)^2/(2\sigma^2)$.
    In this case, note that $D_{S+1} = \{\emptyset\}$. 
    Robust linear regression also falls into this framework. In this case, for any $j \in [n]$, $y_j$ is distributed according to Student's t-distribution.
\end{example}

\begin{example}
    \label{example:logistic} 
    Bayesian logistic regression with Zellner prior. Consider the model defined by
    \begin{align}
        &y_j \sim \mathrm{Bernoulli}\pr{\sigma\pr{\B{x}_j^\top\btheta}}\ , \ \forall j \in [n]\ , \label{eq:logistic_likelihood}\\
        &\btheta \sim \mathcal{N}\pr{\B{0}_d,\B{\Sigma}}\ ,\label{eq:logistic_prior}
    \end{align}
    where $\btheta \in \mathbb{R}^d$ are the unknown regression parameters, $\sigma(u) = 1/(1+\mathrm{e}^{-u})$ is the logistic link, and $\B{\Sigma}^{-1} = \alpha \sum_{j=1}^n\B{x}_j\B{x}_j^\top$, with $\alpha = 3d/(\pi^2 n)$ which corresponds to a Zellner prior \citep{SabanesBove2011,Hanson2014}.
    In this scenario, the posterior density writes $\pi \propto \mathrm{e}^{-U}$ with
    \begin{equation*}
        U(\btheta) = \sum_{j=1}^n y_j\B{x}_j^\top\btheta + \log\br{1 + \exp\pr{-\B{x}_j^\top\btheta}} + \frac{\alpha}{2}\nr{\B{x}_j^\top\btheta}^2\ .
    \end{equation*}
    This posterior density can be re-written as in \eqref{eq:target_density} by setting $b=n$ and for  $i \in [b]$, $d_i=1$, $\B{A}_i = \B{x}_i^\top$ and $U_i(u) = y_i u + \log(1+\mathrm{e}^{-u}) + \alpha u^2/2$. 
    In this case, $\mathrm{D}_i = \{\B{x}_i,y_i\}$ for any $i \in [b]$. 
    Similarly to the logistic regression, other Bayesian generalized linear models such as multinomial logistic regression and Poisson regression also fall into this framework, see \citet{mccullagh2019generalized} for more examples.
\end{example}

Sampling from $\pi$ defined in \eqref{eq:target_density} is challenging because both the number of data $n$ and the dimension $d$ can be large.
In addition, the data set $\mathrm{D}$ might be distributed over a cluster, which complicates the inference procedure.

\subsection{Instrumental Hierarchical Bayesian Model}
\label{subsec:model_algorithm}

To address these issues, \cite{Vono2019} and \cite{Rendell2018} introduced an artificial/instrumental Bayesian hierarchical model to ease posterior computation. 
The idea is to introduce an auxiliary variable $\B{z}_{i} \in \mathbb{R}^{d_i}$ for some factors $i\in [b]$ such that, under an instrumental prior distribution, these variables are conditionally independent given $\btheta$.
Depending on the structure of the initial posterior distribution, different instrumental hierarchical models have been considered by the aforementioned authors.
In this paper, we will study the instance where $\B{z}_{i} \sim \mathcal{N}(\B{A}_i\btheta,\rho^2\B{I}_{d_i})$ for some $\rho>0$. 
Under this model, the artificial joint posterior distribution $\joint(\btheta,\B{z}_{1:b})\propto \exp(-U(\btheta,\B{z}_{1:b}))$ admits a potential function $U$ defined by
\begin{align}
  U(\btheta,\B{z}_{1:b}) = \sum_{i=1}^b U_i(\B{z}_{i}) + \dfrac{\nr{\B{z}_i-\B{A}_i\btheta}^2}{2\rho^2}\ . \label{eq:split_density_generalized}
\end{align}

Figure \ref{fig:DAG} shows the directed acyclic graph (DAG) associated to this instrumental Bayesian hierarchical model.
\begin{figure}
\centering
{\begin{tikzpicture}
  \node[circle,draw=gray,inner sep=2mm] (x) at (0,2) {$\btheta$};
  \node[circle,draw=gray,inner sep=2mm] (y_j) at (0,0) {$y_i$};
  \draw[->] (x) -- (y_j) node[pos=0.5,left] {};
  \node (dim_j) at (0,-0.7) {\small $i \in [n]$};
  \node[draw,fit=(y_j) (dim_j)] {};
\end{tikzpicture}} \hspace{1cm}
{\begin{tikzpicture}
  \node[circle,draw=gray,inner sep=2mm] (x) at (0,3) {$\btheta$};
  \node[circle,draw=gray,inner sep=2mm] (z_1) at (-2,1.5) {$\B{z}_1$};
  \node[circle,draw=gray,inner sep=2mm] (y_1) at (-2,0) {$\boldsymbol{y}_1$};
  \draw[->] (z_1) -- (y_1) node[pos=0.5,left] {};
  \node at (-1,0) {$\hdots$};
  \node[circle,draw=gray,inner sep=2mm] (z_j) at (0,1.5) {$\B{z}_i$};
  \node[circle,draw=gray,inner sep=2mm] (y_j) at (0,0) {$\boldsymbol{y}_i$};
  \draw[->] (z_j) -- (y_j) node[pos=0.5,left] {};
  \node at (1,0) {$\hdots$};
  \node[circle,draw=gray,inner sep=2mm] (z_b) at (2,1.5) {$\B{z}_b$};
  \node[circle,draw=gray,inner sep=2mm] (y_b) at (2,0) {$\boldsymbol{y}_b$};
  \draw[->] (z_b) -- (y_b) node[pos=0.5,left] {};
  \node at (1,1.4) {$\hdots$};
  \node at (-1,1.4) {$\hdots$};
  \draw[->] (x) -- node [below,midway] {} (z_1);
  \draw[->] (x) -- node [below,midway] {} (z_j);
  \draw[->] (x) -- node [below,midway] {} (z_b);
\end{tikzpicture}}
\caption{DAGs for (left) the original model \eqref{eq:target_density} and (right) the instrumental model \eqref{eq:split_density_generalized}.
For any $i \in [b]$, the notation $\boldsymbol{y}_i$ refers to the subset $\{y_j \mid y_j \in \mathrm{D}_i\}$.
Note that we do not illustrate the dependencies on covariates which can be used to define the matrices $\{\B{A}_i\}_{i=1}^b$ as in Example \ref{example:logistic}.}
\label{fig:DAG}
\vspace{-0.3cm}
\end{figure}
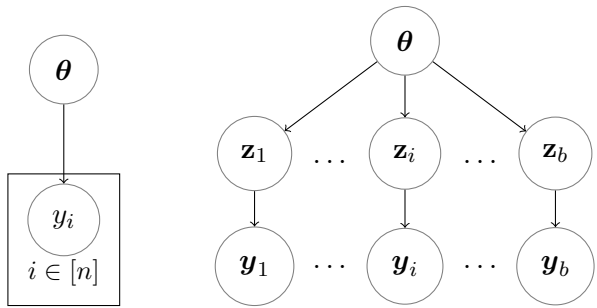
We could have considered an alternative prior for $\B{z}_{i}$ as in \citet{dai2012sampling,Rendell2018,Vono2019_AXDA} but this choice is motivated by the fact that the corresponding quadratic potential enjoys attractive properties such as smoothness and strong convexity. 
Conditions ensuring that $\joint(\btheta,\B{z}_{1:b})$ is a probability density function are detailed in Propositions  \ref{prop:integrability_simpler_condition} and \ref{prop:integrability_ergodicity}.

A key property of this artificial posterior distribution is that the resulting \emph{marginal posterior distribution} \begin{equation}\label{eq:pirhodef}
    \marginal(\btheta) = \int \joint(\btheta,\B{z}_{1:b})\mathrm{d}\B{z}_{1:b}
\end{equation}
converges to the posterior distribution of interest $\pi(\btheta)$ in total variation norm as $\rho\rightarrow0$. 
This follows directly from the fact that $\mathcal{N}(\B{z}_i;\B{A}_i\btheta,\rho^2\B{I}_{d_i})$ weakly converges towards the Dirac distribution $\delta_{\B{A}_i\btheta}(\B{z}_i)$ when $\rho \rightarrow 0$ by Scheff\'e's lemma \citep{Scheffe1947}.

Another key property of the marginal posterior distribution $\marginal$ is that it can be expressed in terms of convolutions in an explicit form.  
Suppose that $\min_{i\in [b]}\inf_{
\bz_i}U_i(\B{z}_i)>-\infty$ and let
\begin{align}
\nonumber U_i^{\rho}(\B{A}_i\btheta)&:=-\log \int_{\R^{d_i}}\exp\l(-U_i(\B{z}_i)-\frac{\|\B{z}_i-\B{A}_i\btheta\|^2}{2\rho^2}\r)\cdot \frac{\mathrm{d} \B{z}_i}{(2\pi \rho^2)^{d_i/2}} \text{ for }i\in [b] \ , \text{ and}\\
U^{\rho}(\btheta)&:=\sum_{i=1}^{b} U_i^{\rho}(\B{A}_i\btheta) \ .
\label{eq:Urhodef}
\end{align}
Then, Proposition \ref{prop:integrability_ergodicity} shows that $\pi_{\rho}(\btheta)\propto \exp(-U^{\rho}(\btheta))$ whenever $\exp(-U^{\rho}(\btheta))$ is integrable on $\R^d$.

The instrumental potential \eqref{eq:split_density_generalized} slightly generalizes the approach from \citet{Vono2019,Rendell2018}. 
In \cite{Rendell2018}, only the case $\B{A}_i=\B{I}_{d}$ is considered so that $\B{z}_i\in\mathbb{R}^{d_i}$ where $d_i=d$. This can be very inefficient. 
In many applications, we can indeed define auxiliary variables $\B{z}_i$ taking values in $\mathbb{R}^{d_i}$ where $d_i\ll d$. 
For instance, in the logistic regression example presented in Example \ref{example:logistic}, we have $d_i=1$ for $i\in[n]$ while $d$ can be large. 
Hence, simulation from $\joint(\B{z}_i|\btheta)$ is expected to be much cheaper. 
Efficient sampling from such conditionals is a key ingredient to SGS as described in the next section.

\subsection{Split Gibbs Sampler}
\label{subsec:algorithm}

The main benefit of working with the artificial target distribution $\joint(\btheta,\B{z}_{1:b})$ defined by \eqref{eq:split_density_generalized} instead of $\pi(\btheta)$ is the fact
that, under $\joint$, the conditional distribution of the auxiliary variables $\B{z}_{1:b}$ given $\btheta$ factorizes across $i\in[b]$, that is $\joint(\B{z}_{1:b}|\btheta) = \prod_{i=1}^b \joint(\B{z}_i|\btheta)$.
Hence these simulation steps can be performed in parallel.
This suggests using a Gibbs sampler to sample from $\joint(\btheta,\B{z}_{1:b})$. 
The resulting so-called split Gibbs sampler is described in Algorithm \label{eq:conditionalunderartificialtargetz}\ref{algo:Gibbs}.
Simple conditions ensuring the ergodicity of SGS are given in \ref{appendix:integrability_ergodicity} 
In the following paragraphs, we detail such conditional sampling problems and discuss the applicability of SGS.
\renewcommand{\baselinestretch}{0.8}
\begin{algorithm}
    \caption{Split Gibbs Sampler (SGS)}
    \label{algo:Gibbs}
     \SetKwInOut{Input}{Input}
     \Input{Potentials $\{U_{i}\}_{i\in [b]}$, penalty parameter $\rho$, initialization $\btheta^{[0]}$ and nb. of iterations $T$.}
   \For{$t \leftarrow 1$ \KwTo $T$}{%
   \For{$i\gets1$ \KwTo $b$}{
   $\B{z}_i^{[t]} \sim  \joint(\B{z}_i|\btheta^{[t-1]})$ (see Equation \ref{eq:conditionalunderartificialtargetzi})
   }
   $\btheta^{[t]} \sim  \joint(\btheta|\B{z}_{1:b}^{[t]})$ (see Equation \ref{eq:conditionalunderartificialtargettheta})
   }
\end{algorithm}
\renewcommand{\baselinestretch}{1.3}

\subsubsection{Sampling the Auxiliary Variables}
\label{subsubsec:sampling_zi}

As emphasized previously, SGS is an attractive sampler since the auxiliary variables $\{\bz_i\}_{i=1}^b$ can be sampled in parallel given $\btheta$ from the conditional distributions
\begin{align}
\label{eq:conditionalunderartificialtargetzi}
  \quad \joint(\B{z}_i|\btheta)\propto \exp\pr{-U_i(\B{z}_i)-\frac{1}{2\rho^2}\nr{\B{z}_i-\B{A}_i\btheta}^2}\ .
\end{align}
Additionally each conditional $\joint(\B{z}_i|\btheta,\B{y})$ only depends on $\B{y} = \{y_j\}_{j=1}^n$ through the subset of observations $\mathrm{D}_i$, see Figure \ref{fig:DAG}.
This is particularly interesting in scenarios where observations $\by$ are distributed over a set of nodes within a cluster such that each node involves the subset $\boldsymbol{y}_i$, see the work by \cite{Rendell2018} for more details.
Not only parallel and possibly distributed sampling from \eqref{eq:conditionalunderartificialtargetzi} is possible but this conditional distribution is far simpler than the original target distribution \eqref{eq:target_density}.
Indeed, while the latter involves a composite potential function $U$ with matrices $\{\B{A}_i\}_{i=1}^b$ acting on $\btheta$, \eqref{eq:conditionalunderartificialtargetzi} only involves a single potential $U_i$ and an isotropic Gaussian term without any matrix acting on $\bz_i$.
Hence, sampling from \eqref{eq:conditionalunderartificialtargetzi} is expected to be easier and cheaper.

In the literature, sampling from this conditional distribution has been performed via two main approaches: exact sampling and Metropolis-Hastings schemes.
% As their names suggest, exact sampling approaches produce a sample $\bz_i$ which is distributed according to \eqref{eq:conditionalunderartificialtargetzi}.
For example, the authors in \citet{Vono2019} considered a linear Gaussian inverse problem where \eqref{eq:conditionalunderartificialtargetzi} was a Gaussian distribution. Apart from the Gaussian case, exact and efficient sampling is for instance possible when considering generalized (non-)linear models and Gaussian prior distributions for $\btheta$.
Similarly to the Bayesian logistic regression example presented in Example \ref{example:logistic}, one can indeed assign one potential per univariate observation $y_i$ leading to univariate potentials $\{U_i\}_{i=1}^n$.
Hence, in this case one can sample from $\joint(\bz_i|\btheta,y_i)$ for $i \in [n]$ by using adaptive rejection sampling \citep{Gilks1992,Martino2011}.

When exact sampling was not possible in practice, the authors in \citet{Rendell2018} considered a Metropolis-Hastings scheme to sample from $\Pi_{\rho}(\bz_i|\btheta)$.
% Contrary to exact sampling, such Metropolis-Hastings algorithms build a Markov chain which admits \eqref{eq:conditionalunderartificialtargetzi} as invariant distribution. 
This defines a Metropolis-within-SGS scheme which can be shown to admit $\Pi_{\rho}$ as stationary distribution under mild assumptions.

In this paper, we are interested in providing explicit and non-asymptotic theoretical convergence guarantees for Algorithm  \ref{algo:Gibbs}.
Since no explicit convergence result exists for special instances of Algorithm \ref{algo:Gibbs}, we choose to focus on the simplest scenario where exact sampling from \eqref{eq:conditionalunderartificialtargetzi} is considered.
One of the aim of our theoretical analysis is to show that, under this exact sampling assumption, we are able to sample efficiently from a close approximation of $\pi$ in high-dimensional settings involving a large number of data.
To this end, we have to ensure that each conditional sampling step involved in Algorithm \ref{algo:Gibbs} can be performed efficiently. This is established in Proposition \ref{prop:rejectionsamplingcomplexity} below which shows that, if $\rho$ is sufficiently small, sampling $\bz_i$ given $\btheta$ can be performed using rejection sampling with $\mathcal{O}(1)$ expected evaluations of $U_i$ and its gradient.

\begin{proposition}[Complexity of rejection sampling]
\label{prop:rejectionsamplingcomplexity}
For any $i \in [b]$, suppose that $U_i$ is $M_i$-gradient Lipschitz for some $M_i > 0$ and that $U_i$ is $m_i$-strongly convex for some $m_i\ge 0$ (possibly zero). 
Let \[V_i(\bz_i):=U_i(\bz_i)+\frac{\|\B{A}_i\btheta-\bz_i\|^2}{2\rho^2},\]
$\bz_i^*(\btheta)$ be the unique minimizer of $V_i$, and $\tbz_i(\btheta)$ be another point (an approximation of $\bz_i^*(\btheta)$). We let 
\begin{equation*}\label{eq:tAidef}
    \tA_i=\frac{1}{\rho^2}+m_i+\frac{\|\grad V_i(\tbz_i(\btheta))\|^2}{2d_i}-\sqrt{\frac{\|\grad V_i(\tbz_i(\btheta))\|^4}{4d_i^2} + \frac{\l(1/\rho^2+m_i\r)\|\grad V_i(\tbz_i(\btheta))\|^2}{d_i}}\ ,
\end{equation*}
and set $\nu_{\btheta}(\bz_i):=\mathcal{N}(\bz_i;\tbz_i(\btheta),(\tA_i)^{-1}\cdot \B{I}_{d_i})$. 

Suppose that we take samples $\bZ_1,\bZ_2,\ldots$ from $\nu_{\btheta}$, and accept them with probability
\begin{align*}\PP(\bZ_j\text{ is accepted})=
\exp\left(-\frac{\|\nabla V_i(\tbz_i(\btheta))\|^2}{2(1/\rho^2+m_i-\tilde{A}_i)} - [V_i(\bZ_j)-V_i(\tbz_i(\btheta))] +\frac{\tilde{A}_i \|\bZ_j-\tbz_i(\btheta)\|^2}{2}\right).
\end{align*}
Then, these accepted samples are distributed according to $\joint(\bz_i|\btheta)$. Moreover, the expected number of samples taken until one is accepted is equal to 
\begin{equation}\label{eq:cor2Eidef}E_i:=\l(\frac{1/\rho^2+M_i}{\tilde{A}_i}\r)^{d_i/2}\cdot \exp\l[\frac{\|\grad V_i(\tbz_i(\btheta))\|^2}{2}\l(\frac{1}{1/\rho^2+m_i-\tA_i}-\frac{1}{1/\rho^2+M_i}\r)\r] \ ,
\end{equation}
which is less than or equal to 2 if 
\begin{equation}\label{eq:lessthan2bnd}\rho^2(2d_i (M_i-m_i)-m_i)\le 1\quad\text{ and }\quad \|\grad V_i(\tbz_i(\btheta))\|\le \frac{2}{7}\cdot \frac{\sqrt{1/\rho^2+m_i}}{\sqrt{d_i}}\ .
\end{equation}
\end{proposition}
\begin{proof}
    The proof is postponed to \ref{sec:RS}
\end{proof}
\begin{remark}
The choice of the approximate minimizer  $\tbz_i(\btheta)$ that we are using in our implementation is built via a few steps of gradient descent started from $\tbz_i^{[0]}(\btheta)=\B{A}_i\btheta$, with step size $\frac{1}{1/\rho^2+M_i}$, that is for $j\ge 1$,
\[\tbz_i^{[j]}(\btheta)=\tbz_i^{[j-1]}(\btheta)-\grad V_i(\tbz_i^{[j-1]}(\btheta))\cdot \frac{1}{1/\rho^2+M_i}\ .\]
We stop once the condition $\|\grad V_i(\tbz_i^{[j]}(\btheta))\|\le \frac{2}{7}\cdot \frac{\sqrt{1/\rho^2+m_i}}{\sqrt{d_i}}$ is satisfied, and set $\tbz_i$ to $\tbz_i^{[j]}$. 
Since the condition number of the function $V_i$ equals 
$\kappa_i=\frac{1+\rho^2 M_i}{1+\rho^2 m_i}$, and the gradient descent decreases the norm of the gradient by a factor of $1-1/\kappa_i$ at each iteration, it follows that we need at most 
\[\l\lceil \frac{\log \|\grad V_i(\B{A}_i\btheta)\|-\log\l(\frac{2}{7}\cdot \frac{\sqrt{1/\rho^2+m_i}}{\sqrt{d_i}}\r)}{\log(1/(1-1/\kappa_i))} \r\rceil\]
iterations before stopping.
\end{remark}

Proposition \ref{prop:rejectionsamplingcomplexity} shows that if $\rho^2 \le 1/(2 d_i M_i)$, then one can use rejection sampling to sample efficiently from \eqref{eq:conditionalunderartificialtargetzi}.
We would like to emphasize that this condition on $\rho^2$ is not limiting if our goal is to sample from a close approximation of $\pi$ using Algorithm \ref{algo:Gibbs}.
Indeed, it follows from Propositions \ref{prop:IUUrho} and \ref{proposition:2} that the bias between $\pi_{\rho}$ defined in \eqref{eq:pirhodef} and $\pi$ in total variation is of the order $\mathcal{O}(\rho^2)\sum_{i=1}^b d_i M_i$.
Hence, if we want to ensure that the bias in total variation is at most $\epsilon$, for $\epsilon >0$, $\rho^2$ has to be chosen of the order $\mathcal{O}(\epsilon)/(\sum_{i=1}^b d_i M_i)$ which is more restrictive than \eqref{eq:lessthan2bnd} in Proposition \ref{prop:rejectionsamplingcomplexity}.

\subsubsection{Sampling the Master Parameter}
\label{subsubsec:sampling_theta}

Regarding the master parameter $\btheta$, it follows from elementary calculations that the conditional distribution of $\btheta$ given $\B{z}_{1:b}$ is Gaussian, that is 
\begin{align}
\label{eq:conditionalunderartificialtargettheta}
  \quad \joint(\btheta|\B{z}_{1:b})=\mathcal{N}(\boldsymbol{\mu}_{\btheta}(\B{z}_{1:b}),\B{\Sigma}_{\btheta})\ ,
\end{align}
with $\B{\Sigma}_{\btheta} = \rho^2 (\sum_{i=1}^b\B{A}_i^\top\B{A}_i)^{-1}$ and $\boldsymbol{\mu}_{\btheta}(\B{z}_{1:b}) = (\sum_{i=1}^b\B{A}_i^\top\B{A}_i)^{-1} \sum_{i=1}^b\B{A}_i^\top\B{z}_i$. 
To ensure that this normal distribution is non-degenerate, the block matrix $[\B{A}_1^\top \hdots \B{A}_b^\top]$ must have full row rank. 
The matrix $\B{\Sigma}_{\btheta}$ is constant across iterations so its Cholesky decomposition, necessary to sample from $\eqref{eq:conditionalunderartificialtargettheta}$, can be pre-computed in a preliminary step.
In cases where the cost of Cholesky decomposition becomes prohibitive (for instance, in high-dimensional scenarios), a lot of methods have been proposed to sample exactly or approximately from a given Gaussian distribution \citep{Vono_2020_Gaussian}.
For instance, we could use samplers inspired from numerical linear algebra such as conjugate-gradient and Lanczos samplers \citep{Ilic2004,Parker2012,Chow2014}. 

\subsection{Connections with Optimization Methods}

\label{subsec:connections_quadratic_penalty}

The SGS whose main steps are described in Algorithm \ref{algo:Gibbs} can be related to common optimization approaches.
More precisely, it can be seen as the stochastic counterpart of alternating minimization (AM) algorithms based on the classical quadratic penalty method \citep[Chapter 7]{NoceWrig06}.
Instead of minimizing a given composite objective function, these algorithms transform this unconstrained minimization problem into a constrained one via a so-called variable splitting technique.
This constraint is then relaxed by adding a ``seemingly naive'' quadratic term to the initial objective function before performing alternating minimization.
In the sequel, we detail such an optimization approach and draw connections between the latter and Algorithm \ref{algo:Gibbs}.

\textit{Quadratic penalty method.} We consider the maximum a posteriori estimation problem under the posterior distribution $\pi$ in \eqref{eq:target_density}, that is
\begin{align}
	\min_{\btheta \in \mathbb{R}^d} \sum_{i=1}^b U_i(\B{A}_i\btheta)\ . \label{eq:joint_minimization_1}
\end{align}
Similarly to direct sampling from $\pi$, solving directly this minimization problem might be computationally demanding because the objective function is a sum of $b$ composite terms, the presence of linear operators acting on $\btheta$, non-differentiability or a possible distributed architecture.
To bypass these issues, some authors \citep{Wang2008,Afonso2010,Leeuwen_2015} proposed to build on variable splitting by introducing a set of auxiliary variables $\{\bz_i\}_{i \in [b]}$ to reformulate \eqref{eq:joint_minimization_1} into the constrained minimization problem
\begin{align*}
	\begin{split}
	&\min_{\btheta \in \mathbb{R}^d, \bz_{1} \in \mathbb{R}^{d_1},\hdots, \bz_{b} \in \mathbb{R}^{d_b}} \sum_{i=1}^b U_i(\B{z}_i) \\
	&\text{subject to} \quad \B{z}_i = \B{A}_i\btheta, i \in [b]\ . %\label{eq:joint_minimization_2}
	\end{split}
\end{align*}
The constraint $\B{z}_i = \B{A}_i\btheta$ is then relaxed by adding a quadratic penalty term in the objective function.
This yields the approximate joint minimization problem
\begin{align*}
	\min_{\btheta \in \mathbb{R}^d, \bz_{1} \in \mathbb{R}^{d_1},\hdots, \bz_{b} \in \mathbb{R}^{d_b}} U(\btheta,\bz_{1:b}) \coloneqq \sum_{i=1}^b U_i(\B{z}_i) + \dfrac{\nr{\B{z}_i - \B{A}_i\btheta}^2}{2\rho^2}\ .
\end{align*}
This optimization problem can be solved by alternating minimization \citep{Beck2015}. 
For fixed $\btheta \coloneqq \btheta^{[t-1]}$, one first minimizes $U(\btheta,\bz_{1:b})$ w.r.t. $\bz_i$ for each factor $i \in [b]$ before minimizing, for fixed $\bz_{1:b} \coloneqq \bz_{1:b}^{[t]}$, $U(\btheta,\bz_{1:b})$ w.r.t. $\btheta$.
Similarly to SGS and at the price of an approximation, the main benefit of this approach is that the minimization problems w.r.t. each auxiliary variable now only involve the sum of a single potential $U_i$ without any operator and a quadratic term.

\textit{SGS and quadratic penalty methods.}
Interestingly, these AM steps stand for the deterministic counterpart of the conditional sampling steps in Algorithm \ref{algo:Gibbs}.
Indeed, instead of drawing a random variable following each conditional, these minimization steps only find the mode associated to each conditional probability distribution and can be related to iterated conditional modes in image processing \citep{Besag1986}.
This shows another interesting link between optimization and simulation and complements earlier connections between these two fields.
For instance, we can mention the celebrated one-to-one equivalence between gradient descent and discretized Langevin dynamics \citep{Roberts1996,Pereyra2016B,Durmus2018} and more recently the use of Hamiltonian dynamics to define first-order descent schemes achieving linear convergence \citep{Duane1987,Maddison2018}.

\textit{Connections and differences with ADMM.} Similar to SGS and quadratic penalty approaches introduced above, the alternating direction method of multipliers (ADMM) also builds on a variable splitting trick to ease an inference task, see \cite{Boyd2011} for a recent comprehensive overview.
However, the connection between SGS and ADMM stops here.
Indeed, contrary to SGS and quadratic penalty methods, ADMM resorts to the so-called augmented Lagrangian and as such involves some dual variables $\B{u}_{1:b}$ in the quadratic penalty terms and their iterative updates via dual ascent steps, see Algorithm \ref{algo:ADMM}.

\renewcommand{\baselinestretch}{0.8}
\begin{algorithm}
    \caption{Alternating Direction Method of Multipliers (ADMM)}
    \label{algo:ADMM}
     \SetKwInOut{Input}{Input}
     \Input{Potentials $U_{i}$ for $i\in[b]$, penalty parameter $\rho$, initialization $\btheta^{[0]}$, $\B{u}_{1:b}^{[0]}$ and nb. of iterations $T$.}
   \For{$t \leftarrow 1$ \KwTo $T$}{%
   \For{$i\gets1$ \KwTo $b$}{
   $\B{z}_i^{[t]} = \underset{\bz_i}{\arg \min} \ U_i(\bz_i) + \dfrac{1}{2\rho^2}\nr{\B{z}_i - \B{A}_i\btheta^{[t-1]} + \B{u}_i^{[t-1]}}^2$ 
   }
   $\btheta^{[t]} = \underset{\btheta}{\arg \min}  \dfrac{1}{2\rho^2} \displaystyle\sum_{i=1}^b\nr{\B{z}_i^{[t]} - \B{A}_i\btheta + \B{u}_i^{[t-1]}}^2$ \\
   \For{$i\gets1$ \KwTo $b$}{
   $\B{u}_i^{[t]} = \B{u}_i^{[t-1]} + \bz_i^{[t]} - \B{A}_i\btheta^{[t]}$
   }
   }
\end{algorithm}
\renewcommand{\baselinestretch}{1.3}

\section{Quantitative Results on the Bias of the Approximate Model}
\label{sec:non_asymptotic_properties}

In order to establish explicit non-asymptotic mixing time bounds for SGS described in Algorithm \ref{algo:Gibbs}, we will first provide in this section quantitative bounds on the bias between $\pi_{\rho}$ and $\pi$ in both total variation and 1-Wasserstein distances. 

\subsection{Results}
\label{subsec:results}

To prove these non-asymptotic results, we shall introduce various regularity conditions listed in Assumption \ref{assumptions}.
\begin{assumption}[General assumptions]
  \label{assumptions}
    \text{}\\[-1.9em]
    \begin{enumerate}[label=(\subscript{A}{{\arabic*} })]
    \setlength\itemsep{0.05em}
    \setcounter{enumi}{-1}
    \item For any $i\in [b]$, $U_i:\R^{d_i}\to \R$ is Borel measurable,  $\inf_{\bz_i\in \R^{d_i}}U_i(\bz_i)>-\infty$, and $\exp(-U^{\rho}(\btheta))$ is integrable on $\R^d$ ($U^{\rho}$ was defined in Equation \ref{eq:Urhodef}).
    %and for at least one $i\in [b]$ we have $d_i=d$, $\B{A}_i$ is full rank, and $\exp(-U_i(\bz_i))$ integrable on $\R^{d}$.
    \item For any $i\in [b]$, $U_i$ is $L_i$-Lipschitz, that is there exists $L_i \geq 0$ such that $|U_i(\B{z}_i')-U_i(\B{z}_i)| \leq L_i\nr{\B{z}_i'-\B{z}_i}$, $\forall \B{z}_i,\B{z}_i' \in \mathbb{R}^{d_i}$.
    \item For any $i\in [b], \B{z}_i \in \mathbb{R}^{d_i}$, $U_i$ is twice continuously differentiable and $-M_i \B{I}_d\preceq \grad^2 U_i(\bz_i)\preceq M_i \B{I}_d$.
    \item For any $i\in [b]$, $U_i$ is convex, that is for any $\alpha\in [0,1]$, $\B{z}_i,\B{z}_i'\in \R^{d_i}$, we have
    $U_i(\alpha \B{z}_i+(1-\alpha)\B{z}_i')\le 
    \alpha U_i( \B{z}_i)+(1-\alpha) U_i( \B{z}_i')$.
    \item For any $i\in [b]$, $U_i$ is $m_i$-strongly convex, that is there exists $m_i \ge 0$ such that $U_i(\B{z}_i)-\frac{m_i\|\B{z}_i\|^2}{2}$ is convex. 
    \item $d_1=\ldots=d_b=d$ and $\B{A}_1=\ldots=\B{A}_b=\B{I}_d$.
    \item For any $i\in [b]$, $U_i$ is centered, that is $\nabla U_i(\B{A}_i\bthetastar) = \B{0}_d$, where $\bthetastar$ is the global minimum of $U$.
  \end{enumerate}
  
\end{assumption}
If some potentials $U_i$ do not verify \assseven, one can first find the global minimum $\bthetastar$ using optimization (typically it takes only a small number of iterations to do this up to machine precision for smooth and strongly convex potentials), and perform a linear shift and define their centered version $\tilde{U}_i$ as $\tilde{U}_i(\B{A}_i\btheta) = U_i(\B{A}_i\btheta) - \langle\B{A}_i\btheta, \nabla U_i(\B{A}_i\bthetastar)\rangle$.  \assseven\  is not just a technical assumption that requires additional work in implementation, without making any difference. It is not difficult to construct an example when $b=2$, $\B{A}_1=\B{A}_2=\B{I}_d$, and $U_1$ and $U_2$ are two quadratics with different covariances whose minimizers are not at the same point, but at distance $D$. In such situations, one can show that the bias between $\pi_{\rho}$ and $\pi$ (in Wasserstein and total variational distance) can depend strongly on $D$, and cannot be bounded based only on the usual smoothness and strong convexity parameters $m_i$, $M_i$. Hence centering has a beneficial effect in reducing the bias of $\pi_{\rho}$ in such situations.
In Section \ref{sec:experimental_results}, we will see that working with the centered potentials $\tilde{U}_i$ does not increase significantly the computational complexity of SGS when rejection sampling is used to sample the auxiliary variables $\bz_i$ conditionally upon $\btheta$.

Our first proposition provides a simple way to verify that \asszero\, holds.
\begin{proposition}[Sufficient conditions for integrability]
\label{prop:integrability_simpler_condition}
Suppose that for any $i\in [b]$, $U_i$ is Borel measurable, and we have potentials $V_i:\R^{d_i}\to \R$ satisfying that
\begin{enumerate}
\item $\inf_{\bz_i\in \R^{d_i}} V_i(\bz_i)>-\infty$,
\item $V_i$ is $L_i$-Lipschitz,  that is $|V_i(\bz_i)-V_i(\bz_i')|\le L_i\|\bz_i-\bz_i'\|$ for any $\bz_i, \bz_i'\in \R^{d_i}$,
\item $V_i$ lower bounds $U_i$, that is $V_i(\bz_i)\le U_i(\bz_i)$ for any $\bz_i \in \R^{d_i}$,
\item $\exp\l(-\sum_{i\in [b]} V_i(\B{A}_i \btheta)\r)$ is integrable on $\R^{d}$.
\end{enumerate}
Then \asszero\, holds.
%Under Assumption \asszero, SGS is $\pi_{\rho}$-irreducible and aperiodic.
\end{proposition}
\begin{remark}
It is easy to check that the Lipschitz conditions are satisfied by the potential terms $\{U_i\}_{i\in[b]}$ of the logistic regression, when there is no Gaussian prior. So in this case, $\exp(-U^{\rho}(\btheta))$ is integrable whenever $\exp(-U(\btheta))$ is integrable. If $\min_{i\in[b]}\inf_{\bz_i\in \R^{d_i}} U_i(\bz_i)>-\infty$, and there is a positive definite Gaussian term of dimension $d$ (typically the prior) among the $U_i$s, then this again can be easily lower bounded by a Lipschitz function satisfying the conditions of Proposition \ref{prop:integrability_simpler_condition}, hence  \asszero\, holds. An alternative sufficient condition has been proposed in Proposition 1 of \cite{plassier2021dglmc}.
\end{remark}
\begin{proof}
The proof is postponed to \ref{appendix:integrability_ergodicity}
\end{proof}

The following proposition establishes the ergodicity of SGS.
\begin{proposition}[Integrability and Ergodicity of SGS]
\label{prop:integrability_ergodicity}
Under \asszero,
$\joint(\btheta,\B{z}_{1:b})$ defines a joint probability density function $\pi_{\rho}(\btheta)\propto \exp(-U^{\rho}(\btheta))$, and 
SGS is $\pi_{\rho}$-irreducible and aperiodic.
\end{proposition}
\begin{proof}
The proof is postponed to \ref{appendix:integrability_ergodicity}
\end{proof}

We now provide results giving non-asymptotic bounds on the bias between $\pi_{\rho}$ and $\pi$.
Only assuming a Lipschitz continuity property on the individual potential functions $\{U_{i}\}_{i \in [b]}$, Proposition \ref{proposition:1} shows that this bias is of the order $\mathcal{O}(\rho\sum_i \sqrt{d_i})$ when $\rho$ is sufficiently small.
This result requires neither differentiability nor convexity and covers standard loss functions used in statistical machine learning such as Huber, pinball or logistic losses \citep{Vono2019_AXDA}.
\begin{proposition}
\label{proposition:1}
    Suppose that $\pi$ satisfies $(A_0)$.
    Let $\pi$ and $\pi_{\rho}$ be as defined in \eqref{eq:target_density} and \eqref{eq:split_density_generalized}.
  For any $i \in [b]$, let $U_i$ satisfy \assone.
  Then, for any $\rho > 0$, we have
  \begin{align}
    \nr{\marginal-\pi}_{\mathrm{TV}} \leq 1 - \prod_{i=1}^b\Delta_{d_i}^{(i)}(\rho)\ , \label{eq:theorem_1_main}
  \end{align}
  where for any $i \in [b]$,
  \begin{align*}
    \Delta_{d_i}^{(i)}(\rho) = \dfrac{D_{-d_i}(L_i\rho)}{D_{-d_i}(-L_i\rho)}\ .
  \end{align*}
  The function $D_{-d_i}$ is the parabolic cylinder special function defined at the end of Section \ref{sec:introdution}.
  In addition, for  $\rho$ sufficiently small, the bound \eqref{eq:theorem_1_main} satisfies
  \begin{align*}
      \nr{\pi_{\rho}-\pi}_{\mathrm{TV}} \leq 2\rho \sum_{i=1}^b d_i^{1/2}L_i + o(\rho)\ .
  \end{align*}
\end{proposition}
\begin{proof}
    The proof is a straightforward extension of \citet[Corollary 3]{Vono2019_AXDA} and is omitted.
\end{proof}
Our next result allows us to bound the TV, KL and 2-Wasserstein biases in terms of a single quantity.
\begin{proposition}\label{prop:IUUrho}
Let 
\begin{equation}\label{eq:Ipipirhodef}I(U,U_{\rho}):=\int_{ \mathbb{R}^d} \pi(\btheta) \cdot (U(\btheta)-U^{\rho}(\btheta))_- \mathrm{d}\btheta+ \l(\log\l(\frac{Z_{\pi_{\rho}}}{Z_{\pi}}\r)\r)_+\ ,
\end{equation}
where $Z_{\pi}:=\int_{ \R^d}\exp(-U(\btheta))d\btheta$ and $Z_{\pi_{\rho}}:=\int_{ \R^d}\exp(-U^{\rho}(\btheta))d\btheta$ are the normalizing constants associated with $\pi$ and $\pi_\rho$, respectively. Then, we have
\[ \nr{\marginal-\pi}_{\mathrm{TV}}\le I(U,U_{\rho})\ , \quad \text{ and }\quad \kld{\pi}{\pi_{\rho}}\le I(U,U_{\rho})\ ,\]
that is the same bound holds for total variation distance and KL-divergence. For the 2-Wasserstein distance, assuming that $U$ is $m$-strongly convex for $m>0$, we have
\[W_2(\pi,\pi_{\rho})\le \sqrt{\frac{2}{m}\cdot I(U,V_{\rho})} \ .\]
\end{proposition}
\begin{proof}
    The proof is postponed to \ref{appendix:A_1}
\end{proof}
If the potentials $\{U_i\}_{i\in[b]}$ are now strongly convex and continuously differentiable with a Lipschitz-continuous gradient, the total variation bias is of order $\mathcal{O}(\rho^2\sum_i d_i)$ for $\rho$ sufficiently small.
\begin{proposition}
    \label{proposition:2}
    Let $\pi$ and $\pi_{\rho}$ as defined in \eqref{eq:target_density} and \eqref{eq:split_density_generalized}, respectively. 
    Suppose first that $\pi$ satisfies \asszero, $b=1$, $d_1=d$, $\B{A}_1$ is full rank and that $(A_2)$ holds (convexity is not required in this case). Let $I(U,U_{\rho})$ be as in \eqref{eq:Ipipirhodef}. Then for any $\rho>0$,
   \begin{align*}
    %\nr{\marginal-\pi}_{\mathrm{TV}} &\le 
    I(U,U_{\rho})\le \frac{\rho^2  d M_1}{2}\ . 
    \label{eq:proof_theorem_2_b1}
  \end{align*}
  In the general multiple splitting case, suppose that Assumptions \asszero, \asstwo, \assfive \ and \assseven \ hold, and  $\det(\sum_{i=1}^b m_i \B{A}_i^\top \B{A}_i)>0$.
  Then $U$ is $m_U$-strongly convex for 
  \begin{equation*}\label{eq:mdef}
    m_U=\lambda_{\min}\left(\sum_{i=1}^b m_i \B{A}_i^\top \B{A}_i\right).
  \end{equation*}
Let $\B{A}=[\B{A}_1^\top, \hdots, \B{A}_b^\top]^\top$ ($\B{A}_1,\ldots, \B{A}_b$ are stacked one upon another), and
\begin{equation}\label{eq:sigma2def}\sigma^2_U:=\|\B{A}^\top \B{A}\| (\max_{i\le b}M_i)^2 \cdot m_U^{-1}\ .\end{equation}
  Then, for  $0<\rho^2\le \frac{1}{6\sigma_U^2}$, we have
  \begin{align}
    %\nr{\marginal-\pi}_{\mathrm{TV}} 
I(U,U_{\rho})    &\le\frac{\rho^2}{2}\pr{\sum_{i=1}^b d_i M_i} +\l(2+\frac{3}{2}d\r)\rho^4 \sigma_U^4\ .
    \label{eq:proof_theorem_2}
  \end{align}
\end{proposition}
\begin{remark}
It is possible to reduce the bias and the constraint on $\rho^2$ in situations where some of the $U_i$'s are quadratic, such as when there is a Gaussian prior (and also more generally in situations where for some indices $j \in [b]$, $e^{-U_j(\bz_j)}$ can be written as the 
convolution of another function and a Gaussian density). In this case, instead of applying the algorithm on the original $U$, we can replace $U_j$ by another quadratic potential  $U_j'$ such that $(U_j')^{\rho}=U_j$, that is the convolution of $e^{-U_j'(\bz_j)}$ and $\frac{1}{(2\pi \rho^2)^{d_j/2}} e^{-\|\bz_j\|^2/(2\rho^2)}$ equals $e^{-U_j(\bz_j)}$. 
Let $[b]^{n.q.}$ denote the set of indices that correspond to non-quadratic potentials. By a straightforward modification of the proof of Proposition \ref{proposition:2} (elimination of the error terms caused by the difference between $U_j$ and $U_j^{\rho}$ for quadratics), one can show that the results hold with $\B{A}$ changed to only contain
$\B{A}_i$ for $i\in [b]^{n.q.}$, $\sigma^2_U$ updated to $\sigma^2_U:=\|\B{A}^\top \B{A}\| (\max_{i\in  [b]^{n.q.}}M_i)^2 \cdot m_U^{-1}$, and the final bound changed to
\[
I(U,U_{\rho}) \le\frac{\rho^2}{2}\pr{\sum_{i \in  [b]^{n.q.}} d_i M_i} +\l(2+\frac{3}{2}d\r)\rho^4 \sigma_U^4 \ .\]
This can improve the dimension dependence in situations when the number of data points is smaller than the dimension $d$, which is often the case for latent Gaussian models.
\end{remark}
\begin{proof}
    The proof is postponed to \ref{appendix:A_1}
\end{proof}

In the single splitting case corresponding to $b=1$, Proposition \ref{thm:Wassersteinpipirhosinglesplitting} builds on the heat equation to derive an explicit and simple bound on the bias between $\pi_{\rho}$ and $\pi$ in 1-Wasserstein distance.
%
%\newpage
\begin{proposition}
    \label{thm:Wassersteinpipirhosinglesplitting}
    Suppose that \asszero \ and \asssix \ hold, $b=1$, and $U_1$ is twice continuously differentiable, and satisfies 
    \begin{equation}\label{eq:fradunbnd}
    U_1(\btheta)\ge a_1+a_2\|\btheta\|^{\alpha} \quad \text{ and }\quad \|\grad U_1(\btheta)\|\le a_3+a_4\|\btheta\|^{\beta}\ ,
    \end{equation}
    for some $a_2>0$, $\alpha>0$, $\beta>0$, $a_1,a_3,a_4\in \R$. 
    Then, we have 
    \begin{equation*}\label{eq:W1bndgeneral}
    W_1(\pi, \pi_{\rho})\le \min\l(\rho \sqrt{d}, \frac{1}{2}\rho^2 \int_{\mathbb{R}^d} \|\grad U_1(\btheta)\| \pi(\btheta) \mathrm{d}\btheta\r)\ .
    \end{equation*}
    Moreover, if $U_1$ satisfies Assumptions \asstwo \ and \assfive \ (gradient Lipschitz and strong convexity properties), then \eqref{eq:fradunbnd} holds, and we have
   \begin{equation*}\label{eq:W1bndstrongconvexsmooth}W_1(\pi, \pi_{\rho})\le \min\l(\rho \sqrt{d}, \frac{1}{2}\rho^2 \sqrt{M_1 d}\r)\ .
   \end{equation*}
\end{proposition}
\begin{proof}
    The proof is given in \ref{appendix:subsec:bounds_bias_Wasserstein} Although we believe that a similar bound also holds for multiple splitting, unfortunately we have not found a way to obtain such a result (the heat equation argument used within the proof is not straightforward to adapt to the $b>1$ case). 
    Nevertheless, Proposition \ref{proposition:2} provides an alternative Wasserstein bound for strongly convex $U$ via Proposition \ref{prop:IUUrho} (which provides a bound not as sharp as this result in the single splitting case).
\end{proof}

Table \ref{table:equivalents} summarizes the non-asymptotic bounds on the bias we have obtained.
\begin{table}
\renewcommand{\arraystretch}{1.3}
\centering
\begin{tabular}{l l l}
\thickhline
\bfseries Distance & \bfseries Assumptions & \bfseries Upper bound \\
\hline 
\multirow{3}{*}{$\nr{\marginal-\pi}_{\mathrm{TV}}$} & \asszero, \assone & $\rho\displaystyle\sum_{i=1}^b 2\sqrt{d_i} L_i + o(\rho)$ \\
%& \multirow{1}{*}{\asszero,\asstwo, $b=1$ & $\displaystyle \frac{\rho^2 M_1 d}{2}$ \\
& \multirow{1}{*}{\asszero, \asstwo, $b=1$} & $\displaystyle \frac{1}{2}\rho^2 M_1d$
\\
& \multirow{1}{*}{\asszero, \asstwo, \assfive, \assseven} & $\displaystyle \frac{1}{2}\rho^2\sum_{i=1}^b M_i d_i + o(\rho^2)$
\\
\hline
$W_1(\marginal,\pi)$ & \asszero, \asstwo, \assfive, \asssix, $b=1$ & $\min\l(\rho \sqrt{d}, \frac{1}{2}\rho^2 \sqrt{M_1 d}\r)$  \\
\thickhline
\end{tabular}
\caption{Non-asymptotic bounds given in Propositions \ref{proposition:1}, \ref{proposition:2} and \ref{thm:Wassersteinpipirhosinglesplitting}.}
\label{table:equivalents}
\end{table}
In the single splitting case with $\B{A}_1=\B{I}_d$, we have the following two steps when moving from $\btheta^{[t]}$ to $\btheta^{[t+1]}$.
 \begin{enumerate}
     \item Sample $\bz^{[t]}\sim \joint(\bz|\btheta^{[t]})\propto \exp\left(-U(\bz)-\|\bz-\btheta^{[t]}\|^2/(2\rho^2)\right)$.
     \item Sample $\btheta^{[t+1]}\sim \mathcal{N}(\bz^{[t]},\rho^2 \B{I}_d)$.
 \end{enumerate}
By Taylor's expansion, we can see that in the $\rho\to 0$ limit, we have $\joint(\bz|\btheta^{[t]})\approx \mathcal{N}(\btheta^{[t]}-\rho^2 \grad U(\btheta^{[t]}), \rho^2 \B{I}_d)$, and hence $\btheta^{[t+1]}$ is approximately distributed as $\mathcal{N}(\btheta^{[t]}-\rho^2 \grad U(\btheta^{[t]}), 2\rho^2 \B{I}_d)$, which corresponds to one ULA step with stepsize $h=\rho^2$. 
One can show that the same approximation holds in the multiple splitting case as well, after appropriate preconditioning. 
Thus, in the $\rho\to 0$ limit, SGS can be interpreted as another discretization of the Langevin diffusion.
Hence it is natural to compare the bias bounds of Table \ref{table:equivalents} with the best available bias bounds for ULA; as far as we know, these were stated in \cite{durmus2019analysis}.
We do this in Table \ref{table:equivalentsULA}, for ULA with step size $h=\rho^2$. 
\begin{table}
\renewcommand{\arraystretch}{1.3}
\centering
\begin{tabular}{l l l}
\thickhline
\bfseries Distance & \bfseries Assumptions & \bfseries Upper bound \\
\hline 
$\nr{\pi_{\mathrm{ULA}}-\pi}_{\mathrm{TV}}$ & \asstwo, \assfour & $\sqrt{2 M_1 d} \cdot \rho$ \\
%& \multirow{1}{*}{\asszero,\asstwo, $b=1$ & $\displaystyle \frac{\rho^2 M_1 d}{2}$ \\
\hline
$W_1(\pi_{\mathrm{ULA}},\pi)$ & \asstwo, \assfive & $\sqrt{2 \frac{M_1}{m_1} d}\cdot \rho$ \\
\thickhline
\end{tabular}
\caption{Non-asymptotic bounds for ULA with step size $h=\rho^2$, see Theorem 12 of \cite{durmus2019analysis}.}
\label{table:equivalentsULA}
\end{table}
SGS has a significantly smaller bias than ULA both in total variation, and Wasserstein distances, suggesting a higher order approximation. Indeed the bias is $\mathcal{O}(\rho^2)$ for 1-Wasserstein and total variation distances for SGS, while it is $\mathcal{O}(\rho)$ for both distances for ULA. 
This means that SGS can be seen as a discretization of the Langevin diffusion, whose stationary distribution is significantly less biased compared to ULA. The key reason why we have been able to provide these bounds is that we have access to a quite explicit form of the stationary distribution for SGS (see Equation \ref{eq:Urhodef} and Proposition \ref{prop:integrability_ergodicity}), while there is no such explicit form available for ULA. Another valuable property of SGS in the single splitting case is that the posterior mean of $\marginal$ is the same as the posterior mean of the target $\pi$, since $\marginal$ is formed by the convolution of $\pi$ and a Gaussian.
Finally, we would like to highlight that SGS in the single splitting case can be used to sample from $\pi$ without adding any bias by considering the marginal distribution of $\bz$ under $\Pi_{\rho}(\btheta,\bz)$ (which equals $\pi$) instead of $\pi_{\rho}(\btheta)$ \citep{pmlr-v134-lee21a,Liang2021}. We do not consider this alternative since our primary focus is on the multiple splitting scenario where this exact sampling approach is not possible.

\subsection{Illustrations on a Toy Gaussian Model}
\label{subsec:toy_Gaussian_model}

We perform a sanity check of the tightness of the upper bounds derived in Section \ref{subsec:results} on a toy Gaussian model for which a closed-form expression is available for both $\pi_{\rho}$ and the considered  statistical distances.
The target distribution is chosen as a scalar Gaussian 
\begin{equation*}
    \pi(\theta) = \mathcal{N}\pr{\theta;\mu,\frac{\sigma^2}{b}},
\end{equation*} 
where $b \geq 1$ and $\sigma > 0$.
In the sequel, we set $\mu=0$, $\sigma = 3$ and $b=10$.
To satisfy the assumptions associated to each distance (see Table \ref{table:equivalents} for a summary), we consider two splitting strategies.

\textit{Splitting strategy 1.} Since the bound on $\nr{\pi-\pi_{\rho}}_{\mathrm{TV}}$ is valid for any number of splitting operations, we set $U_i(\theta) = (2\sigma^2)^{-1}(\theta-\mu)^2$ for any $i \in [b]$. 
The marginal of $\theta$ under the instrumental hierarchical model in \eqref{eq:split_density_generalized} has the closed-form expression 
\begin{equation*}
    \pi_{\rho}(\theta) = \mathcal{N}\pr{\theta;\mu,\frac{\sigma^2 + \rho^2}{b}}.
\end{equation*}

\textit{Splitting strategy 2.} As the bound in 1-Wasserstein distance has only been established for a single splitting operation, we set $U(\theta) \coloneqq U_1(\theta) = b(2\sigma^2)^{-1}(\theta-\mu)^2$. This yields
\begin{equation*}
    \pi_{\rho}(\theta) = \mathcal{N}\pr{\theta;\mu,\frac{\sigma^2}{b} + \rho^2}.
\end{equation*}

Figure \ref{fig:toy_gaussian_1} illustrates the bounds derived in Section \ref{subsec:results} for both TV (with splitting strategy 1) and 1-Wasserstein (with splitting strategy 2) distances.
The 1-Wasserstein distance has been calculated by numerical integration using the identity $W_1(\pi,\pi_{\rho}) = \int_{\mathbb{R}}|F(u) - F_{\rho}(u)|\mathrm{d}u$
where $F$ and $F_{\rho}$ are the cumulative distribution functions (c.d.f.) associated to $\pi$ and $\pi_{\rho}$, respectively.
For this simple problem, these bounds manage to achieve the correct decay in $\mathcal{O}(\rho^2)$ for small values of $\rho$: the quantitative bound on the 1-Wasserstein distance is particularly tight.
\begin{figure}
\centering
  \text{}\hspace{-0.5cm}
  \mbox{{\includegraphics[scale=0.5]{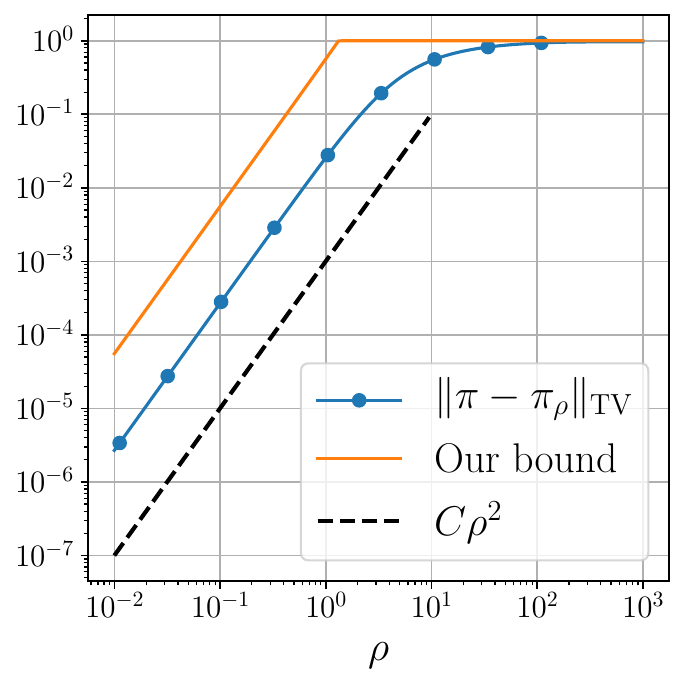}}}
  \mbox{{\includegraphics[scale=0.5]{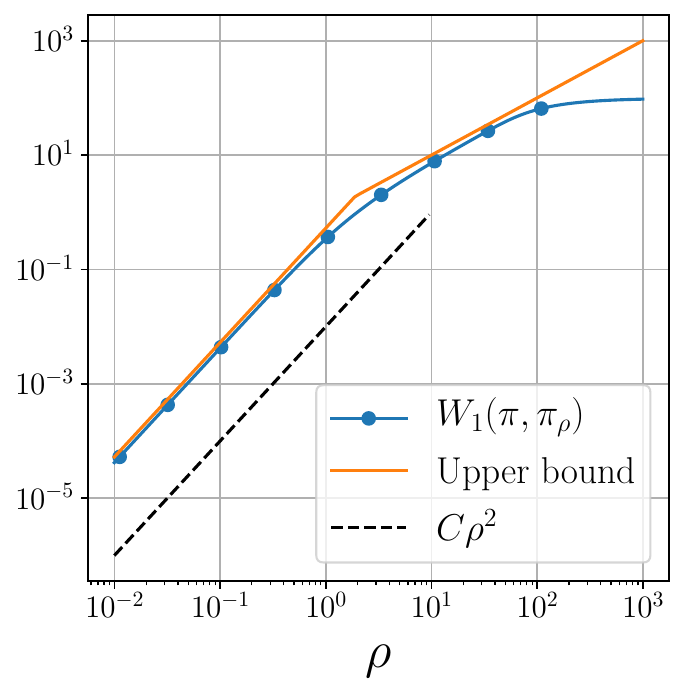}}}\\
  \caption{$\nr{\pi-\pi_{\rho}}_{\mathrm{TV}}$ (left) and $W_1(\pi,\pi_{\rho})$ (right) as a function of $\rho$ along with the bounds established in Section \ref{subsec:results} for the toy Gaussian model considered in Section \ref{subsec:toy_Gaussian_model}.
  The dashed line shows the slope associated to a decay in $\rho^2$ in log-log scale (the parameter $C$ stands for a positive constant).}
  \label{fig:toy_gaussian_1}
\end{figure}

\section{Main Results: Explicit Mixing Time Bounds}
\label{sec:main_result_convergence_rates}

We now state our main results regarding the non-asymptotic bounds on the mixing time of SGS.

\subsection{Explicit Convergence Rates}

In this section, we first prove a key result related to the Ricci curvature of SGS which allows us to derive explicit convergence rates for this algorithm.

\subsubsection{Lower Bound on the Ricci Curvature of the SGS Kernel}
The SGS sampler described in Algorithm \ref{algo:Gibbs} generates a Markov chain $(\btheta^{[t]})_{t\geq 1}$ of transition kernel $\B{P}_{\mathrm{SGS}}$ defined by
\begin{equation*}
    \B{P}_{\mathrm{SGS}}(\btheta,\btheta')=\int_{\bz_{1:b}}\joint(\bz_{1:b}|\btheta) \joint(\btheta'|\bz_{1:b})\mathrm{d}\bz_{1:b}\ ,
\end{equation*}
where the conditional distributions associated to $\Pi_{\rho}$ are defined in \eqref{eq:conditionalunderartificialtargetzi} and \eqref{eq:conditionalunderartificialtargettheta}. 
For any $\btheta\neq \btheta'\in \R^d$, given a metric $
w: \R^d\times \R^d\to \R^{+}$, the coarse Ricci curvature $K(\btheta,\btheta')$ of $\B{P}_{\mathrm{SGS}}$, introduced by \cite{ollivier2009ricci}, equals
\begin{equation*}\label{eq:KdefW1}
    K(\btheta,\btheta')=1-\frac{W_{1}^{w}(\B{P}_{\mathrm{SGS}}(\btheta,\cdot),\B{P}_{\mathrm{SGS}}(\btheta',\cdot))}{w(\btheta,\btheta')}\ ,
\end{equation*}
for any $\btheta\ne \btheta'\in \R^d$.
We can also define this quantity for $p$-Wasserstein  distances for any $1
\le p \le \infty$ by
\begin{equation*}\label{eq:KdefWp}
    K_p(\btheta,\btheta')=1-\frac{W_{p}^{w}(\B{P}_{\mathrm{SGS}}(\btheta,\cdot),\B{P}_{\mathrm{SGS}}(\btheta',\cdot))}{w(\btheta,\btheta')}\ .
\end{equation*}
In Theorem \ref{thm:RicciSGS}, we show under Assumption \assfive\, that for any 1$\le p\le \infty$ and a suitable metric $w$, $K_p(\btheta,\btheta')$ is lower bounded by a simple quantity having an explicit dependence w.r.t. the tolerance parameter $\rho$ and the strong convexity constants of the potential functions $\{U_{i}\}_{i\in [b]}$.

\begin{theorem}
    \label{thm:RicciSGS}
     Suppose that $\pi$ satisfies $(A_0)$ and that \assfive \ holds. Define the metric \begin{equation}\label{ew:wmetricdef}w(\btheta,
    \btheta')=\left\|\left(\sum_{i=1}^{b}\B{A}_{i}^\top \B{A}_i\right)^{1/2}(\btheta-\btheta')\right\|\ .\end{equation}
    Let
    \begin{equation}\label{eq:kappasgsdef}K_{\mathrm{SGS}}:=1-
    \left\| \left(\sum_{i=1}^{b} \B{A}_i^\top \B{A}_i\right)^{-1/2} \left(\sum_{i=1}^{b}\frac{\B{A}_i^\top \B{A}_i}{1+m_i\rho^2}\right)
    \left(\sum_{i=1}^{b} \B{A}_i^\top \B{A}_i\right)^{-1/2}
    \right\|\ . 
    \end{equation}
    Then for the transition kernel $\B{P}_{\mathrm{SGS}}$ of SGS, $K_p(\btheta,\btheta')\ge K_{\mathrm{SGS}}$ for any $\btheta\neq \btheta'\in \R^d$ and any $1\le p\le \infty$. 
\end{theorem}
\begin{proof}
    The proof is postponed to \ref{appendix_proof_theorem_K_SGS}
\end{proof}

As shown in the following corollary, Theorem \ref{thm:RicciSGS} implies that the convergence rate of SGS towards its invariant distribution is governed by the constant $K_{\mathrm{SGS}}$ defined in \eqref{eq:kappasgsdef}.

\begin{corollary}
    \label{coro:convergence_rates_SGS}
     Suppose that $\pi$ satisfies $(A_0)$ and that \assfive \ holds. 
    Then, for any $1\le p\le \infty$ and any initial distribution $\nu$ on $\R^d$, we have
    \begin{align}
        \label{eq:W1convergencesec4}
        W_{p}^{w}(\nu \B{P}_{\mathrm{SGS}}^t, \pi_{\rho})&\le W_{p}^{w}(\nu, \pi_{\rho})\cdot (1-K_{\mathrm{SGS}})^t\ ,\\
        \label{eq:TVconvergencesec4}\|\nu \B{P}_{\mathrm{SGS}}^t-\pi_{\rho}\|_{\mathrm{TV}}&\le \mathrm{Var}_{\pi_{\rho}}\l(\frac{\mathrm{d}\nu}{\mathrm{d} \pi_{\rho}}\r)\cdot (1-K_{\mathrm{SGS}})^t\ ,\nonumber
    \end{align}
    where $W_p^w$ denotes the Wasserstein distance of order $p$ w.r.t. the metric $w$ defined in \eqref{ew:wmetricdef}.
\end{corollary}
\begin{proof}
    The proof is postponed to \ref{appendix_proof_coro_convergence_rates}
\end{proof}

An attractive property of the convergence rate $K_{\mathrm{SGS}}$ is that it is dimension-free: it only depends on $b$, $\rho^2$ and the strong convexity parameter $m_i$, and neither requires differentiability nor smoothness of the potential functions $\{U_{i}\}_{i \in [b]}$.
This is of interest since Corollary \ref{coro:convergence_rates_SGS} can be applied to many problems where non-differentiable potential functions are considered; see \cite{Li2010_elasticnet,Gu_2014,Xu2015_grouplasso}.

\subsubsection{Illustrations on the Toy Gaussian Example}

Before proving our mixing time bounds for the SGS, we perform a simple sanity check on the toy Gaussian example considered in Section \ref{subsec:toy_Gaussian_model} in order to assess the tightness of the convergence bounds stated in Corollary \ref{coro:convergence_rates_SGS}. In this case, the $\theta$-chain follows a Gaussian auto-regressive process of order 1. We can thus compute analytically the Markov transition kernel $\nu P_{\mathrm{SGS}}^t$ and the total variation and 1-Wasserstein distances between this kernel and the invariant distribution $\pi_{\rho}$; see \ref{appendix:toy_Gaussian_example_2} for details.
For this toy Gaussian example, the convergence rate of SGS is governed by 
\begin{align*}
    K_{\mathrm{SGS}} &= \frac{\rho^2}{\sigma^2+\rho^2}\ , \ \text{for the splitting strategy 1}, \\
    K_{\mathrm{SGS}} &= \frac{b\rho^2}{\sigma^2+b\rho^2}\ , \ \text{for the splitting strategy 2}.
\end{align*}
Figure \ref{fig:toy_gaussian_2} illustrates our convergence bounds for each splitting strategy and associated statistical distance.
For the total variation case, the slope in log-scale associated to our bound, which equals $\log(1 - K_{\mathrm{SGS}})$, appears to be sharp since it matches the slope associated to the observed convergence rate.
Regarding the Wasserstein scenario, the slope associated to our bound is roughly equal to twice the real slope in log-scale, and hence is a bit conservative.
\begin{figure}
\centering
  \mbox{{\includegraphics[scale=0.5]{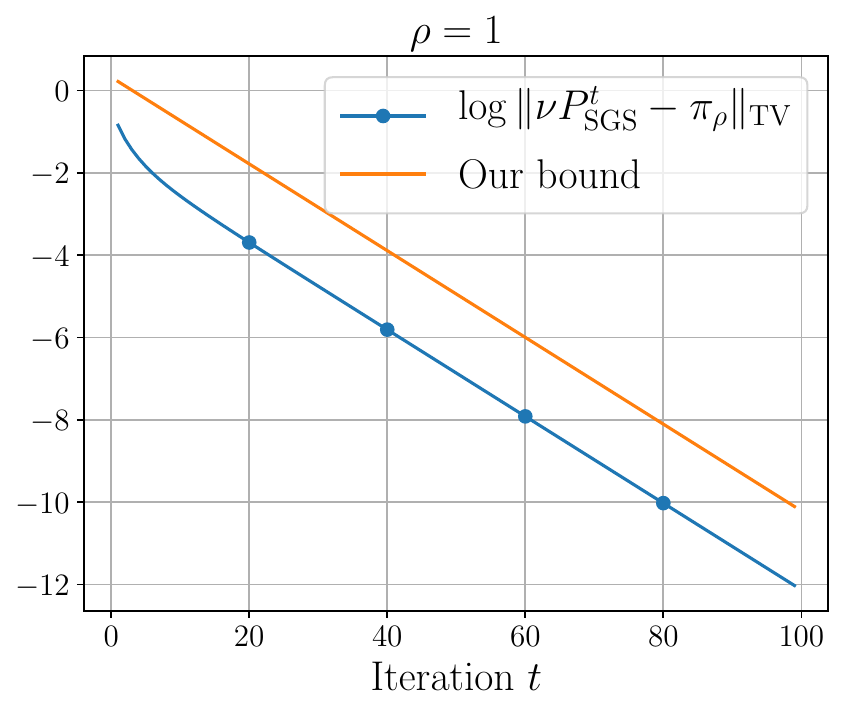}}}
  \mbox{{\includegraphics[scale=0.5]{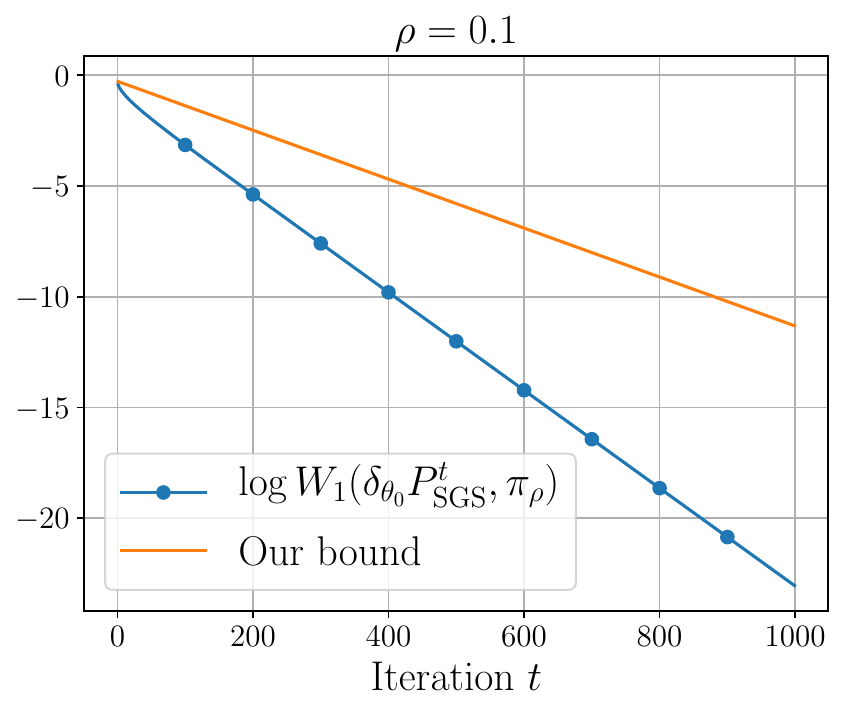}}}
  \caption{From left to right: $\nr{\nu P_{\mathrm{SGS}}^t-\pi_{\rho}}_{\mathrm{TV}}$ with $\nu(\theta) = \mathcal{N}(\theta;\mu,\sigma^2/b)$ and $W_1(\delta_{\theta_0} P_{\mathrm{SGS}}^t,\pi_{\rho})$ with $\theta_0 = 0$ along with the bounds shown in Theorem \ref{thm:RicciSGS} for the toy Gaussian model considered in Section \ref{subsec:toy_Gaussian_model}.}
  \label{fig:toy_gaussian_2}
\end{figure}

We are now ready to prove our main results, namely mixing time bounds associated to SGS when apply to an initial target density $\pi$ which is both smooth and strongly log-concave.
These assumptions will be weakened in Section \ref{subsec:weakly_log_concave}.

\subsection{User-Friendly Mixing Time Bounds}
\label{subsec:mixing_time_bounds}

We consider two cases, namely the single splitting strategy where $b=1$ and the multiple one where the initial density $\pi$ involves $b \geq 1$ composite potential functions.
In both cases, we derived explicit expressions for the mixing time of SGS and compared them to the ones recently obtained in the MCMC literature including results associated to common subsampling MCMC approaches.

\subsubsection{Single Splitting Strategy}
\label{subsec:single_splitting}

We begin by considering the case $b=1$ corresponding to a single splitting operation of the potential function $U \coloneqq U_1$.
Since sampling from the conditional \eqref{eq:conditionalunderartificialtargetzi} is as difficult as sampling from the initial target $\pi$, this scheme is not particularly relevant from a practical point of view.
Nevertheless, the convergence analysis of the scheme and its comparison with state-of-the-art MCMC approaches are still worth studying from a theoretical point of view.
Indeed, the non-asymptotic results we derive in the single splitting case are simpler and allow practitioners to have a first theoretical understanding of the convergence behavior of Algorithm \ref{algo:Gibbs} in high-dimensional settings.
By combining our error bounds on $W_1(\pi,\pi_{\rho})$ from Proposition \ref{thm:Wassersteinpipirhosinglesplitting} to the convergence bound \eqref{eq:W1convergencesec4} in Corollary \ref{coro:convergence_rates_SGS}, we obtain the following complexity result.

\begin{table}
\centering
\begin{tabular}{llll}
\thickhline
     \multicolumn{1}{c}{\bfseries Reference} &
      \multicolumn{1}{c}{\bfseries Method} &
      \multicolumn{1}{c}{\bfseries Validity} &
      \multicolumn{1}{c}{\bfseries Evals}\\
    \hline
    \cite{durmus2019analysis} & Unadjusted Langevin & $0< \epsilon\le 1$& $\mathcal{O}^*\l( \frac{\kappa}{ \epsilon^2}\r)$\\
   \cite{Dalalyan2019} & SGLD & $0< \epsilon\le 1$& $\mathcal{O}^*\l( \frac{\kappa}{ \epsilon^2}\r)$\\
    \cite{pmlr-v75-cheng18a} & Underdamped Langevin& $0< \epsilon\le 1$ & 
    $\mathcal{O}^*\l( \frac{\kappa^2}{\epsilon}\r)$
    \\
    \cite{dalalyan2018sampling}& Underdamped Langevin & $0<\epsilon\le \frac{1}{\sqrt{\kappa}}$ & $\mathcal{O}^*\pr{\frac{\kappa^{3/2}}{\epsilon}}$\\
    \cite{chen2019optimal} & Hamiltonian Dynamics & $0<\epsilon\le 1$ & $\mathcal{O}^*\pr{\frac{\kappa^{3/2}}{\epsilon}}$\\
    this paper & SGS with single splitting & $0<\epsilon\le \frac{1}{d \sqrt{\kappa}}$ & $\mathcal{O}^*\pr{\frac{\kappa^{1/2}}{\epsilon}}$\\
    \hline
\thickhline
\end{tabular}
\caption{Comparison of convergence rates in Wasserstein distance with the literature, starting from the minimizer $\btheta^\star$ of the $m_1$-strongly convex and $M_1$-smooth potential $U_1(\btheta)$, with condition number $\kappa=M_1/m_1$. SGS with single splitting is implemented based on rejection sampling. $\mathcal{O}^*(\cdot)$ denotes $\mathcal{O}(\cdot)$ up to polylogarithmic factors. In the last column, the complexity stands for the number of gradient and function evaluations to get a $W_1$ error of $\frac{\epsilon \sqrt{d}}{\sqrt{m}_1}$.
The acronym SGLD refers to stochastic gradient Langevin dynamics.}
\label{table:comparison_wass}
\end{table}

\begin{theorem}[Complexity bound for Wasserstein distance]\label{prop:compcomplexityWass}
 Suppose that $\pi$ satisfies $(A_0)$.
 Suppose that $b=1$ and that Assumptions \asstwo, \assfive \ and \asssix \ hold. Let $\bthetastar$ be the unique minimizer of $U_1(\btheta)$ and let $\nu = \delta_{\bthetastar}$ be the initial distribution. Suppose that $\epsilon\le 1$. Then, with the choice 
\begin{equation}
    \label{eq:rho2singlesplittingwass}
    \rho^2=\max\l(\frac{\epsilon^2}{4 m_1},\frac{\epsilon}{\sqrt{m_1M_1}}\r)\ ,
\end{equation}
and number of iterations $t \ge t_{\mathrm{mix}}(\epsilon\sqrt{d/m_1};\nu)$ where
\begin{equation}
    \label{eq:ksinglesplittingwass}
    t_{\mathrm{mix}}(\epsilon\sqrt{d/m_1};\nu) =  \frac{\log\l(\frac{3}{\epsilon}\r)}{\log\l(1+\max\l(\frac{\epsilon^2}{4},\epsilon\sqrt{\frac{m_1}{M_1}}\r)\r)}\ ,
\end{equation}
we have
\begin{equation*}
    W_1(\nu P_{\mathrm{SGS}}^{t},\pi)\le \frac{\epsilon}{\sqrt{m_1}}\sqrt{d}\ .
\end{equation*}
This implies that, using $t$ steps of SGS, we can obtain a sample that has a  Wasserstein distance from the target $\pi$ at most equal to $\frac{\epsilon \sqrt{d}}{\sqrt{m_1}}$.
\end{theorem}
\begin{proof}
    The proof is postponed to \ref{subsec:mixing_time_bounds_Wasserstein}
\end{proof}

Several comments can be made on the result stated in Theorem \ref{prop:compcomplexityWass}.
The expressions of both the choice of the tolerance parameter \eqref{eq:rho2singlesplittingwass} and the mixing time \eqref{eq:ksinglesplittingwass} are simple and can be computed in practice.
These nice properties along with the explicit dependencies of the mixing time of SGS w.r.t. the condition number $\kappa \coloneqq M_1/m_1$ of $U_1$ and the desired precision $\epsilon$ make Theorem \ref{prop:compcomplexityWass} of particular interest for practitioners.
In addition, under smoothness and strong convexity of the potential $U_1$ (see Assumption \ref{assumptions}), one can show that $W_1(\delta_{\bthetastar},\pi) \le \sqrt{d/m_1}$  \citep[Proposition 1]{durmus2018high}.
This quantity can be interpreted as the typical deviation associated to the sampling problem.
Under the assumptions of Theorem \ref{prop:compcomplexityWass}, it follows that $W_1(\nu P_{\mathrm{SGS}}^{t},\pi)$ is upper bounded by $\epsilon$ times this typical deviation.
Note that considering the relative precision $\epsilon \sqrt{d/m_1}$ yields a mixing time bound \eqref{eq:ksinglesplittingwass} which is invariant to the scaling of $U$ (that is replacing $U$ by $\alpha U$ with $\alpha$ > 0).  

For a fixed condition number $\kappa$ and a sufficiently small precision $\epsilon$, \eqref{eq:ksinglesplittingwass} implies that the mixing time of SGS scales as $\mathcal{O}(\sqrt{\kappa} \epsilon^{-1}\log(3\epsilon^{-1}))$.
To be competitive with other MCMC algorithms, such as those based on Langevin or Hamiltonian dynamics, we have to ensure that the auxiliary variable $\bz_1$ can be efficiently drawn at each iteration of Algorithm \ref{algo:Gibbs}.  
In Proposition \ref{prop:rejectionsamplingcomplexity} in Section \ref{subsubsec:sampling_zi}, we established that this is possible by showing that if $\epsilon\le 1/(d\sqrt{\kappa})$, then sampling $\bz_1$ given $\btheta$ can be performed by rejection sampling with $\mathcal{O}(1)$ expected evaluations of $U_1$ and its gradient. 
Based on this rejection sampling scheme, Table \ref{table:comparison_wass} compares our complexity result for SGS with single splitting with the ones derived recently in the literature. It shows that SGS compares favourably to standard MCMC methods when $0<\epsilon\le 1/(d\sqrt{\kappa})$ including the commonly-used subsampling approach called stochastic gradient Langevin dynamics (SGLD) \citep{Welling2011}. 
An explanation for this improved performance in terms of precision $\epsilon$ and condition number $\kappa$ is the fact that SGS admits a stationary distribution with an explicit form, which we were able to exploit to establish smaller bounds on the bias for the same step size.

\begin{theorem}[Complexity bound for TV distance, single splitting]
\label{PROP:COMPCOMPLEXITYTV_SINGLE}
Suppose that $b=1$, $d_1=d$, $\B{A}_1$ is invertible, and that Assumptions \asszero, \asstwo \ and \assfive \ hold, with $m_1>0$. Let $\bthetastar$ be the unique minimizer of $\btheta \mapsto U(\btheta)=U_1(\B{A}_1\btheta)$.
Let $\nu(\btheta):=\mathcal{N}(\btheta;\bthetastar,(M_1\B{A}_1^\top\B{A}_1)^{-1})$ be the initial distribution. Then for any $0< \epsilon\le 1$, with the choice 
\begin{equation*}%\label{eq:rho2TVquadratic_single}
    \rho^2 \le \frac{\epsilon}{dM_1}\ ,
\end{equation*}
and number of iterations $t \geq t_{\mathrm{mix}}(\epsilon;\nu)$ where
\begin{equation*}
    %\label{eq:kTV_single}
    t_{\mathrm{mix}}(\epsilon;\nu) = \frac{\log\l(\frac{2}{\epsilon}\r) + C/2}{K_{\mathrm{SGS}}}\ ,
\end{equation*} 
for $K_{\mathrm{SGS}}=\frac{m_1\rho^2}{1+m_1\rho^2}$ and
\begin{equation*}
C=\frac{5d}{8} + \frac{d}{2}\log\l(\frac{M_1}{m_1}\r)\ ,
\end{equation*}
we have 
\begin{equation*}
\|\nu P_{\mathrm{SGS}}^{t} - \pi\|_{\mathrm{TV}}\leq \epsilon\ .
\end{equation*}
This means that starting from $\nu$, after $t$ step of SGS, we are at a TV-distance at most $\epsilon$ from $\pi$.
\end{theorem}
\begin{proof}
    The proof is postponed to  \ref{subsec:mixing_time_bounds_TV}
\end{proof}

\subsubsection{Multiple Splitting Strategy}
\label{subsec:multiple_splitting}

In this section, we consider the general case where $b \ge 1$ potential functions have been split as in \eqref{eq:split_density_generalized}.
For this scenario, the following theorem states explicit mixing time bounds in total variation distance.

\begin{theorem}[Complexity bound for TV distance, multiple splitting]
\label{prop:compcomplexityTV}
Assume that \asszero, \asstwo, \assfive \ and \assseven \ hold, and $\mathrm{det}\l(\sum_{i=1}^b m_i\B{A}_i^\top\B{A}_i\r)>0$.
Let $\bthetastar$ be the unique minimizer of  $U(\btheta)=\sum_{i=1}^{b}U_i(\B{A}_i\btheta)$.
Let \[\nu(\btheta):=\mathcal{N}\l(\btheta;\bthetastar,\l(\sum_{i=1}^b M_i\B{A}_i^\top\B{A}_i\r)^{-1}\r)\] be the initial distribution. Then for any $0< \epsilon\le 1$, with the choice 
\begin{equation}\label{eq:rho2TVquadratic}
    \rho^2  \le \frac{\displaystyle\sum_{i=1}^bd_i M_i\pr{\sqrt{1 + 8\epsilon \sigma_U^4\pr{2+\frac{3}{2}d} \pr{\displaystyle\sum_{i=1}^bd_i M_i}^{-2}}-1}}{4\sigma_U^4\pr{2+\frac{3}{2}d}} \wedge \frac{1}{6 \sigma^2_U}
\end{equation}
and number of iterations $t \geq t_{\mathrm{mix}}(\epsilon;\nu)$ where
\begin{equation}
    \label{eq:kTV}
    t_{\mathrm{mix}}(\epsilon;\nu) = \frac{\log\l(\frac{2}{\epsilon}\r) + C/2}{K_{\mathrm{SGS}}}\ ,
\end{equation} 
for  $K_{\mathrm{SGS}}$ defined in \eqref{eq:kappasgsdef}, and
\[C=d\sigma^2_U + \rho^4\l(2+d\r)\sigma_U^4 + \frac{17}{32} \sum_{i=1}^{b}d_i + \frac{1}{2}\log\l(\frac{\mathrm{det}\l(\sum_{i=1}^b M_i\B{A}_i^\top\B{A}_i\r)}{\mathrm{det}\l(\sum_{i=1}^b m_i\B{A}_i^\top\B{A}_i\r)}\r)\ ,
\]
we have 
\[\|\nu P_{\mathrm{SGS}}^{t} - \pi\|_{\mathrm{TV}}\leq \epsilon\ .\]
This means that starting from $\nu$, after $t$ step of SGS, we are at a TV-distance at most $\epsilon$ from $\pi$.
\end{theorem}
\begin{proof}
    The proof is postponed to \ref{subsec:mixing_time_bounds_TV}
\end{proof}

Again, note that both the tolerance parameter \eqref{eq:rho2TVquadratic} and the $\epsilon$-mixing time \eqref{eq:kTV} are explicit and can be computed in practice.
If we denote the condition number of the potential $U$ by $\kappa \coloneqq M/m$, this theorem implies that $t_{\mathrm{mix}}(\epsilon;\nu)$ scales as $\mathcal{O}(d^2\kappa/\epsilon)$ up to polylogarithmic factors.
In this scenario, Table \ref{table:comparison_TV} compares our complexity results for SGS implemented using rejection sampling with existing results in the literature. 
For the same initialization $\nu$, we have better dependencies than ULA w.r.t. both $\kappa$, $d$ and $\epsilon$.
However, MALA seems to have better convergence rates in total variation distance in general, except in badly conditioned situations, where the rates for SGS can be better.
Moreover, compared to MALA and ULA, SGS with multiple splitting is amenable to distributed and parallel computations. In this distributed environment, the complexity results shown in this table suggest that SGS is an attractive approach to sample from a smooth, strongly log-concave and composite target distribution.

\begin{table}
\centering
\begin{tabular}{llll}
\thickhline
     \multicolumn{1}{c}{\bfseries Reference} &
      \multicolumn{1}{c}{\bfseries Method} &
      \multicolumn{1}{c}{\bfseries Validity} &
      \multicolumn{1}{c}{\bfseries Evals} \\
    \hline
    \cite{Durmus2017} & ULA, $\nu = \delta_{\bthetastar}$& $0\le \epsilon\le 1$ & $\mathcal{O}^*\l(\kappa^2 d/\epsilon^2 \r)$\\
     $\left\{
        \begin{array}{l}
            \text{\cite{pmlr-v83-cheng18a}}\\
            \text{\cite{durmus2019analysis}}
        \end{array}
    \right.$ & ULA, $\nu = \nu_m$ & $0< \epsilon\le 1$& $\mathcal{O}^*\l(\kappa^2 d/\epsilon^2\r)$\\
    \cite{Dalalyan2017} & ULA, $\nu = \nu_M$& $0\le \epsilon\le 1$ & $\mathcal{O}^*\l(\kappa^2 d^3/\epsilon^2 \r)$\\
    \cite{durmus2019analysis} & SGLD, $\nu = \nu_M$& $0\le \epsilon\le 1$ & $\mathcal{O}^*\l(\kappa^2 d^3/\epsilon^2 \r)$\\
      \cite{Dwivedi2019}& MALA, $\nu = \nu_M$ & $0< \epsilon\le 1$& $\mathcal{O}\l(\kappa^2 d^2 \log^{1.5}\l(\frac{\kappa}{\epsilon^{1/d}} \r)\r)$\\
this paper & SGS, $\nu = \nu_M$ & $0<\epsilon\le 1$ & $\mathcal{O}^*(\kappa d^2 /\epsilon)$\\
    \hline
\thickhline
\end{tabular}
\caption{Comparison of convergence rates in TV distance with the literature, starting from a Gaussian distribution centered at the minimizer $\btheta^\star$ of the $m$-strongly convex and $M$-smooth potential $U(\btheta)$, with condition number $\kappa=\frac{M}{m}$. SGS is implemented based on rejection sampling. $\mathcal{O}^*(\cdot)$ denotes $\mathcal{O}(\cdot)$ up to polylogarithmic factors, $\nu_m(\btheta) = \mathcal{N}(\btheta;\btheta^\star,\frac{\B{I}_d}{m})$ and $\nu_M(\btheta) = \mathcal{N}(\btheta;\btheta^\star,\frac{\B{I}_d}{M})$.
The notation $\nu$ stands for the initialization of each method.}
\label{table:comparison_TV}
\end{table}

\subsection{Non-Strongly Log-Concave Target Density}
\label{subsec:weakly_log_concave}

The complexity results shown in Section \ref{subsec:mixing_time_bounds} assume that each potential $\{U_{i}\}_{i\in [b]}$ is strongly convex.
In cases where there are $b-1$ convex potential functions and the $b$-th one stands for an isotropic quadratic term (coming from the prior distribution for instance), this strongly-convex assumption can be met.
Indeed, one can decompose the quadratic potential into $b-1$ strongly convex terms and add each of them to each individual convex potential $\{U_{i}\}_{i\in[b-1]}$.
Nevertheless, the strongly-convex assumption is still restrictive.
In this section, we extend our explicit mixing time bound in the multiple splitting scenario to densities which are smooth (see Assumption \asstwo) but such that each individual potential $\{U_{i}\}_{i\in[b]}$ only satisfies the standard convexity assumption \assfour \ instead of satisfying the strong convexity assumption \assfive.

Similarly to \cite{Dalalyan2017} and \cite{Dwivedi2019}, we will weaken our strongly-convex assumption \assfive \ by approximating each potential $U_i$ with a strongly convex one and then applying our previous proof techniques to this approximation.
More precisely, instead of the initial target density $\pi$ in \eqref{eq:target_density}, we now consider the approximate density $\tilde{\pi}(\btheta) \propto \exp(-\tilde{U}(\btheta))$ with
\begin{equation}
    \tilde{U}(\btheta) = \sum_{i=1}^b U_i(\B{A}_i\btheta) + \dfrac{\lambda}{2}\nr{\btheta - \bthetastar}^2,\label{eq:tildepi}
\end{equation}
where $\lambda > 0$ and $\bthetastar$ stands for a minimizer of $U$.
This approximation allows us to apply Theorems \ref{PROP:COMPCOMPLEXITYTV_SINGLE} and \ref{prop:compcomplexityTV} with the new smooth and strongly-convex constants $\tilde{M}_i = M_i + \lambda$ and $\tilde{m}_i = \lambda$ in order to find the minimum number of SGS steps such that the TV distance from $\tilde{\pi}$ is less than $\epsilon$.
To achieve an $\epsilon$ TV-distance from the initial target density $\pi$, we have to consider an additional error term to bound, namely $\nr{\pi - \tilde{\pi}}_{\mathrm{TV}}$.
If $\int_{\mathbb{R}^d} \nr{\btheta-\bthetastar}^4\pi(\btheta)\mathrm{d}\btheta \leq d^2R^2$ with $R > 0$, then with the choice $\lambda = 4\epsilon/(3bdR)$, we have $\nr{\pi - \tilde{\pi}}_{\mathrm{TV}} \leq \epsilon/3$ \citep[Lemma 3]{Dalalyan2017}.
Combining this result with Theorem \ref{PROP:COMPCOMPLEXITYTV_SINGLE}, the following corollary states a complexity result for the single splitting strategy.
The one corresponding to the multiple splitting one can be obtained using Theorem \ref{prop:compcomplexityTV} in a similar manner but is omitted here for simplicity.

\begin{corollary}[Complexity bound for TV distance, no strong convexity]
\label{coro:weakly_strongly_convex}
Suppose that $b=1$, Assumptions \asszero, \asstwo \ and \assfour \ hold and 
\[\int_{\mathbb{R}^d} \nr{\btheta-\bthetastar}^4\pi(\btheta)\mathrm{d}\btheta \leq d^2R^2\]
for some $R > 0$. Let $\tilde{\pi}$ be defined as in \eqref{eq:tildepi}. 
Let $\nu(\btheta):=\mathcal{N}(\btheta;\bthetastar,(\tilde{M_1}\B{A}_1^\top\B{A}_1)^{-1})$ be the initial distribution with $\tilde{M_1} = M_1 + \lambda$. Then for any $0< \epsilon\le 1$, with the choices $\lambda = 4\epsilon/(3dR)$ and
\begin{equation*}\label{eq:rho2TVquadratic_single}
    \rho^2 \le \frac{2\epsilon}{3d(M_1+\lambda)}\ ,
\end{equation*}
and number of iterations 
\begin{equation*}
    \label{eq:kTV_single}
    t \geq \frac{\log\l(\frac{3}{\epsilon}\r) + C/2}{K_{\mathrm{SGS}}}\ ,
\end{equation*} 
for
\begin{equation*}
 K_{\mathrm{SGS}} = \frac{\lambda \rho^2}{1 + \lambda \rho^2} \text{ and } C=\frac{5d}{8} + \frac{d}{2}\log\l(\frac{M_1 + \lambda}{\lambda}\r)\ ,
\end{equation*}
we have 
\begin{equation*}
\|\nu P_{\mathrm{SGS}}^{t} - \pi\|_{\mathrm{TV}}\leq \epsilon\ .
\end{equation*}
This means that starting from $\nu$, after $t$ step of SGS applied to the approximate density $\tilde{\pi}$, we are at a TV-distance at most $\epsilon$ from $\pi$.
\end{corollary}
\begin{proof}
    The proof is straightforward. It follows from the triangle inequality and Theorem \ref{PROP:COMPCOMPLEXITYTV_SINGLE}.
\end{proof}

Compared to our mixing time bound derived under the assumption that the potential $U_1$ is strongly convex, Corollary \ref{coro:weakly_strongly_convex} shows that relaxing the strongly convex assumption affects negatively the dependence w.r.t. both the dimension $d$ and the precision $\epsilon$, as it scales as $\mathcal{O}^*(M_1 d^2/\epsilon^2)$.  Nevertheless, this complexity result improves upon that in \citet{Dalalyan2017,Dwivedi2019} for the unadjusted Langevin algorithm (ULA) and the Metropolized random walk (MRW), which respectively scale as $\mathcal{O}^*(M_1^2d^3/\epsilon^4)$ and $\mathcal{O}^*(M_1^2d^3/\epsilon^2)$.

\subsection{Comparison with Existing Divide-and-Conquer and Subsampling-Based MCMC Schemes}

So far, we have mainly compared the theoretical behavior of SGS in high-dimensional scenarios with common MCMC schemes such as those derived from Langevin and Hamiltonian dynamics, see Tables \ref{table:comparison_wass}
 and \ref{table:comparison_TV}.
 In this section, we discuss and compare when possible the theoretical results associated to SGS with those associated to existing divide-and-conquer and subsampling-based MCMC approaches. Regarding divide-and-conquer approaches, although a lot of algorithms have been proposed over the past ten years (see Section \ref{sec:introdution}), very few non-asymptotic and explicit convergence results exist up to the authors' knowledge \citep{plassier2021dglmc}.
Among available results, we can cite those associated to non-parametric approaches proposed by \citet{Wang2013,Neiswanger2014,Wang2015} which showed that these methods scaled exponentially with respect to the dimension $d$ because of the use of kernel density estimates.
For general subsampling-based approaches which do not resort to the Bernstein-von Mises approximation, explicit bounds have been recently derived for (variance-reduced)  stochastic gradient MCMC algorithms such as SGLD \citep{Dalalyan2019,durmus2019analysis}.
As illustrated in Tables \ref{table:comparison_wass} and \ref{table:comparison_TV}, our non-asymptotic theoretical analysis shows that SGS is competitive with other state-of-the-art MCMC algorithms.

\section{Numerical Illustrations}
\label{sec:experimental_results}

This section aims at illustrating the main theoretical results of Section \ref{sec:main_result_convergence_rates}. 
We consider three different examples which satisfy all the assumptions required in our main statements. 
The first experiment considers the case where the target $\pi$ is a multivariate Gaussian density while the second one sets $\pi$ to be a mixture of two multivariate Gaussian densities.
Finally, the third experiment considers a Bayesian binary logistic regression problem with a Gaussian prior.
For all approaches and experiments, the initial distribution will be set to $\nu = \mathcal{N}(\bthetastar,(\sum_{i=1}^b M_i\B{A}_i^\top\B{A}_i)^{-1})$ for the TV distance and to $\nu = \delta_{\bthetastar}$ for the 1-Wasserstein one.
Although SGS is amenable to a distributed implementation \citep{Rendell2018}, all the experiments have been run on a serial computer to emphasize that it is even beneficial in this context.
The experiments have been carried out on a Dell Latitude 7390 laptop equipped with an Intel(R) Core(TM) i5-8250U 1.60 GHz processor, with 16.0 GB of RAM, running Windows 10. 

\subsection{Multivariate Gaussian Density}

In this example, we want to verify empirically the dependencies of the mixing times derived in Section \ref{sec:main_result_convergence_rates} w.r.t. the dimension $d$, the desired precision $\epsilon$ and the condition number $\kappa$ of the potential $U$.
We consider a target zero-mean Gaussian density on $\mathbb{R}^d$
\begin{equation*}
    \pi(\btheta) \propto \exp\pr{-\dfrac{1}{2}\btheta^\top\B{Q}\btheta},
\end{equation*}
where $\B{Q} \in \mathbb{R}^{d \times d}$ is a positive definite precision matrix. 
In the sequel, $\B{Q}$ will be chosen to be diagonal and anisotropic, that is $\B{Q} = \mathrm{diag}(q_1,\hdots,q_d)$, with $q_i \neq q_j$ for $ i \neq j$.
The resulting potential function $U \coloneqq U_1 = \btheta^\top\B{Q}\btheta / 2$ is strongly convex and smooth with parameters $m = \min_{i \in [d]} q_i$ and $M = \max_{i \in [d]} q_i$.
Since computing the total variation distance between continuous and multidimensional measures is challenging, we discretized the latter over a set of bins and consider the error between the empirical marginal densities associated to the least favorable direction, that is along the eigenvector associated to $m$.
In the following, we will illustrate our mixing time results for both 1-Wasserstein and total variation distances in the strongly log-concave case.

\begin{figure}
\centering
  \mbox{{\includegraphics[scale=0.4]{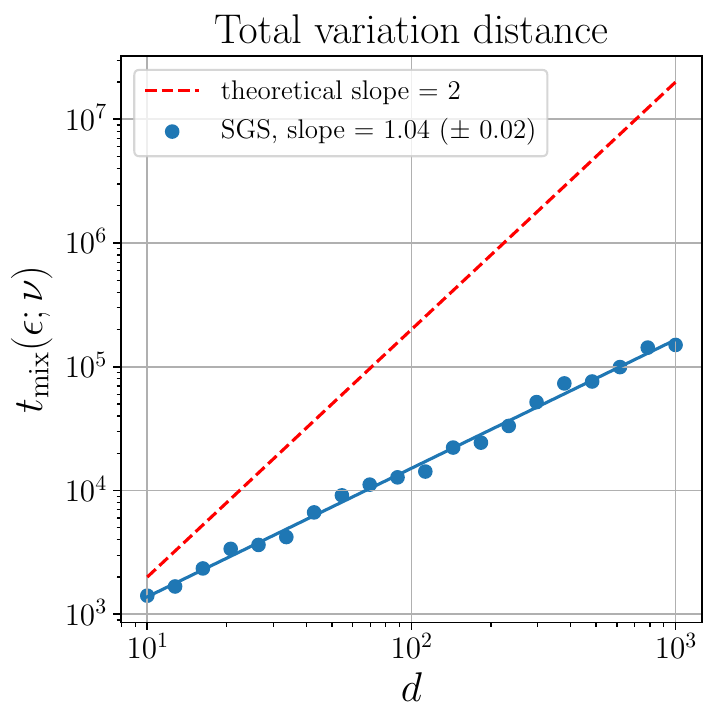}}}
  \mbox{{\includegraphics[scale=0.4]{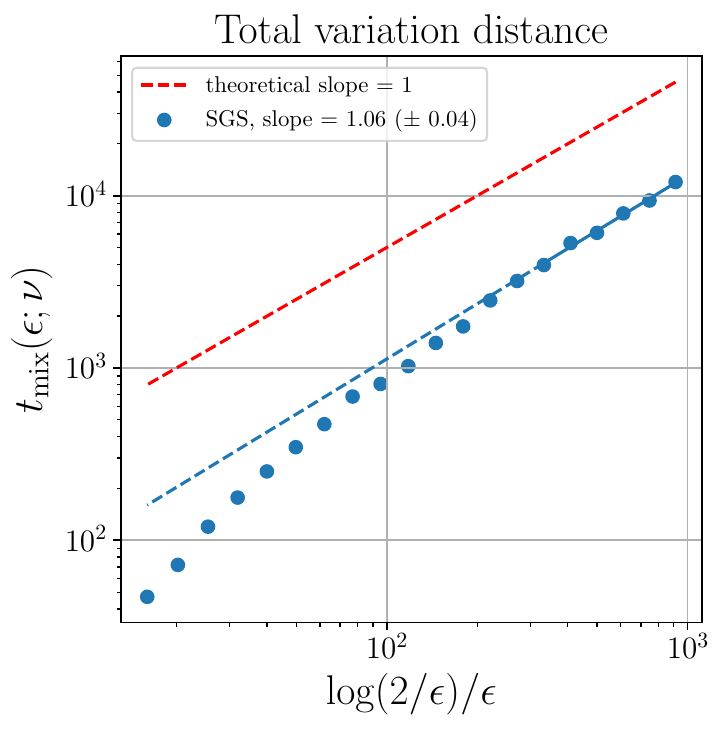}}}
  \mbox{{\includegraphics[scale=0.4]{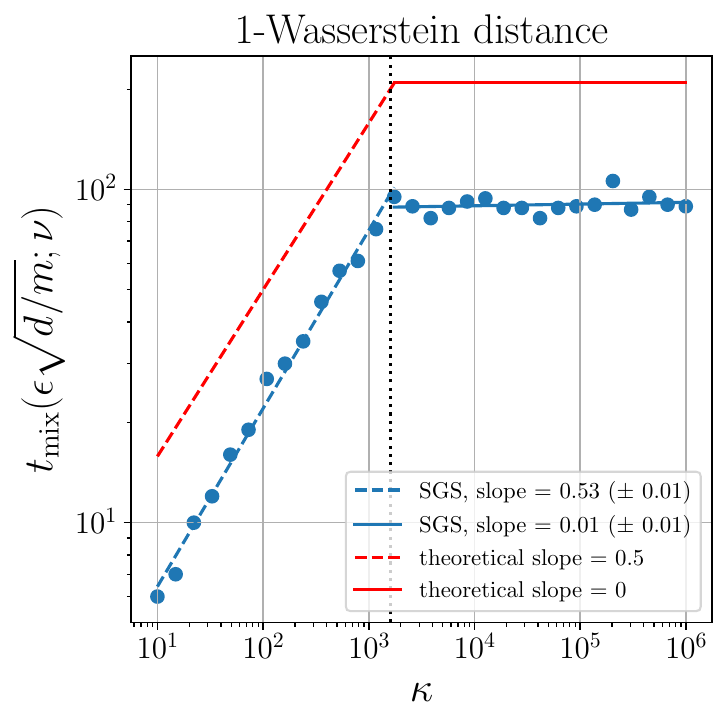}}}
  \caption{Multivariate Gaussian. (left and middle) $\epsilon$-mixing times for the total variation distance and (right) $\epsilon \sqrt{d/m}$-mixing times for the 1-Wasserstein distance.}
  \label{fig:exp1_strongly}
\end{figure}

\textit{Dimension dependence.} We set here $\epsilon = 0.1$,  $m = 1/4$, $M = 1$ such that $\kappa \coloneqq M/m = 4$ and are interested in the dimension dependence of our $\epsilon$-mixing time result for SGS.
We let the dimension $d$ vary between $10^1$ and $10^3$ and ran SGS for each case.
We measured its $\epsilon$-mixing time by recording the smallest iteration such that the discrete total variation error falls below the desired precision $\epsilon$.
The mixing time has been averaged over $10$ independent runs.
Figure \ref{fig:exp1_strongly} illustrates the behavior of the mixing time of SGS w.r.t. the dimension $d$ in log-log scale.
In order to assess the dimension dependency, we performed a linear fit and reported the slope of the linear model.
According to Table \ref{table:comparison_TV}, the dimension dependence is of order $\mathcal{O}(d^2)$.
Interestingly, we found in this example that the dimension dependence of the mixing time of SGS is linear w.r.t. $d$.

\textit{Precision dependence.} We set here $d = 2$ and  $\kappa = 3$ while the prescribed precision $\epsilon$ varies between $6\times10^{-3}$ and $1.6\times10^{-1}$, and ran SGS for each case.
As before, we measured its $\epsilon$-mixing time by recording the smallest iteration such that the discrete total variation error falls below the desired precision $\epsilon$.
Figure \ref{fig:exp1_strongly} shows the behavior of the mixing time of SGS w.r.t. $\log(2/\epsilon)\epsilon^{-1}$ in log-log scale.
For sufficiently small precisions, this figure confirms our theoretical result which states that the mixing time of SGS scales as $\mathcal{O}(\log(2/\epsilon)\epsilon^{-1})$.

\textit{Condition number dependence.} Regarding the 1-Wasserstein distance and the complexity results depicted in Table \ref{table:comparison_wass}, the main difference between existing MCMC approaches is the dependence w.r.t. the condition number $\kappa$ of the potential function $U$.
Here, we aim at verifying the latter numerically. 
To this purpose, we set $d = 10$, $\epsilon = 0.1$ and let $\kappa$ vary between $10^1$ and $10^6$.
From \eqref{eq:ksinglesplittingwass}, it appears that the dependence of the mixing time of SGS depends on $\max\{\epsilon^2/4,\epsilon/\sqrt{\kappa}\}$.
This quantity equals $\epsilon/\sqrt{\kappa}$ for $\kappa \leq 1600$ and $\epsilon^2/4$ otherwise. Hence, we are expecting to retrieve a dependence in $\kappa^{1/2}$ for small and moderate $\kappa$ and a mixing time only depending on $\epsilon$ for larger values of the condition number.
We performed 50 independent runs of SGS and stopped them when their empirical Wasserstein error fell below $\epsilon\sqrt{d/m}$.
The results are depicted on Figure \ref{fig:exp1_strongly} in log-log scale.
As before, we did a linear fit to assess the dependency of the mixing time w.r.t. the condition number $\kappa$.
The slope of the linear model for SGS equals $0.53$ for $\kappa \leq 1600$ (depicted with a black dotted vertical line) which confirms the theoretical dependence of the order $\mathcal{O}(\kappa^{1/2})$.
As expected, the mixing time of SGS becomes independent of $\kappa$ for larger values.

\subsection{Gaussian Mixture}
\label{subsec:gaussian_mixture}

\begin{figure}
    \centering
    \mbox{{\includegraphics[scale=0.39]{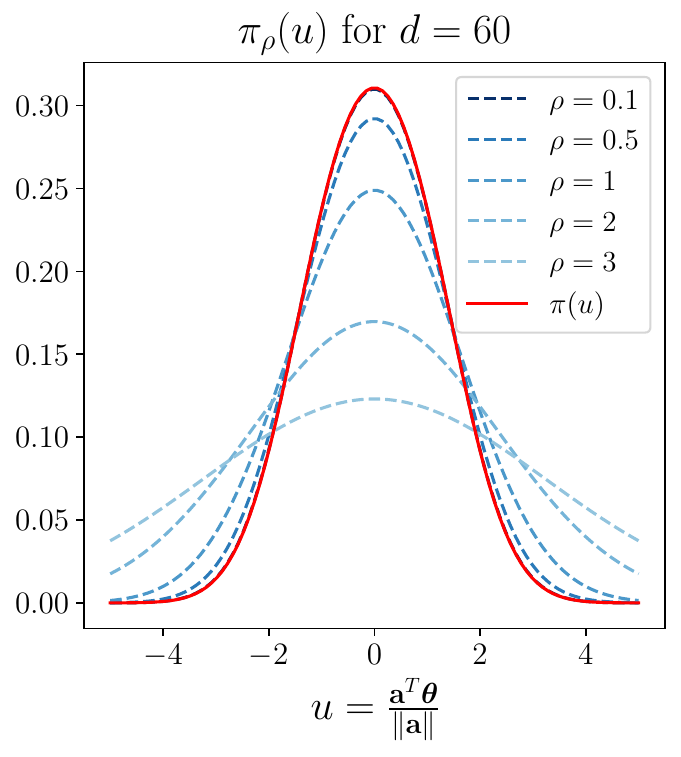}}}
    \mbox{{\includegraphics[scale=0.39]{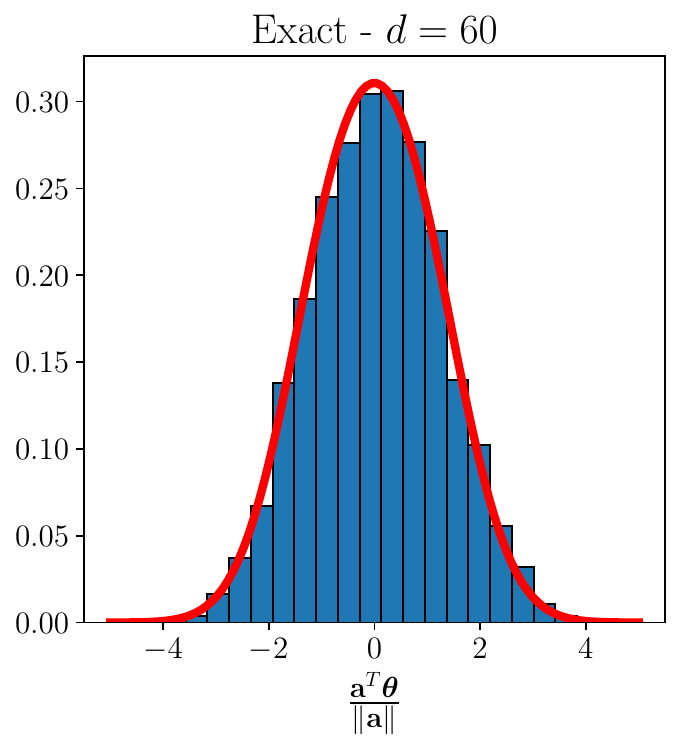}}}
    \mbox{{\includegraphics[scale=0.39]{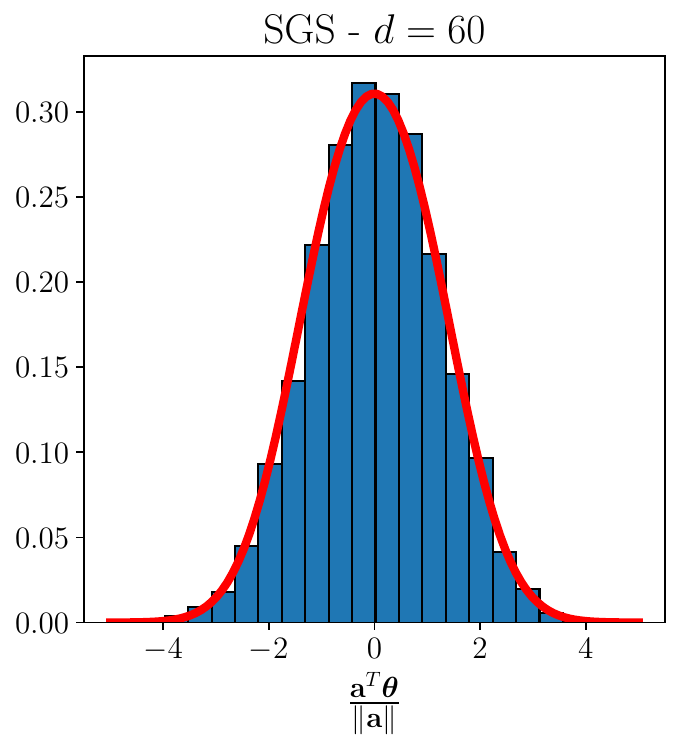}}}
  \caption{Gaussian mixture with $d=60$. From left to right: behavior of $\pi_{\rho}(u)$ w.r.t. $\rho$ with $u = \B{a}^\top\btheta/\nr{\B{a}}$; empirical distribution obtained by exact sampling from $\pi$; empirical distribution obtained by sampling from $\pi_{\rho}$ with the guidelines recommended in Theorem \ref{prop:compcomplexityTV}.
  The histograms have been computed using 2500 independent samples and the precision has been set to $\epsilon = 0.1$.
  In all figures, the red curve stands for $\pi(u)$.}
  \label{fig:exp2}
\end{figure}

In this second experiment, also considered by \cite{Dalalyan2017} and \cite{Dwivedi2019}, we show that the values of the tolerance parameter $\rho$ and the mixing time $t_{\mathrm{mix}}(\epsilon;\nu)$ recommended by Theorem \ref{PROP:COMPCOMPLEXITYTV_SINGLE} indeed yield approximate samples having a distribution close to $\pi$.
We also verify that the running time required to generate such samples is reasonable, and compare it to the running time of ULA to achieve the same prescribed precision $\epsilon$.
To this purpose, let us consider the simple problem of generating samples from a mixture of two Gaussian densities with density $\pi$ defined, for all $\btheta \in \mathbb{R^d}$, by
\begin{align*}
    \pi(\btheta) &= \dfrac{1}{2(2\pi)^{d/2}}\pr{\exp\pr{- \dfrac{\nr{\btheta - \mathbf{a}}^2}{2}} + \exp\pr{- \dfrac{\nr{\btheta + \mathbf{a}}^2}{2}}} \nonumber \\
    &\propto \exp\pr{- U(\btheta)},
\end{align*}
where 
\begin{equation*}
    U(\btheta) = \dfrac{1}{2}\nr{\btheta - \mathbf{a}}^2 - \log\pr{1 + e^{-2\btheta^\top\mathbf{a}}},
\end{equation*}
and $\mathbf{a} \in \mathbb{R}^d$ is a fixed vector involved in the mean of each Gaussian density.
If $\nr{\mathbf{a}} < 1$, one can show that $U$ is $M$-smooth and $m$-strongly convex with $m = 1 - \nr{\mathbf{a}}^2$ and $M=1$.
In the sequel, we choose $\B{a}$ such that $\nr{\B{a}} = 1/\sqrt{2}$, which also implies that the global minimizer of $U$ is $\bthetastar = \B{0}_d$.
Since $\pi$ admits a finite second order moment, all the assumptions required in Theorem \ref{PROP:COMPCOMPLEXITYTV_SINGLE} are verified.
We now consider a single splitting strategy on $U$ leading to the joint approximate density $\joint(\btheta,\B{z})$ defined in \eqref{eq:split_density_generalized} with $b=1$ and $\B{A}_1 = \B{I}_d$.
Under this distribution, the marginal density $\pi_{\rho}(\btheta)$ writes
\begin{align*}
    \pi_{\rho}(\btheta) &= \dfrac{1}{2(2\pi(1+\rho^2))^{d/2}}\pr{\exp\pr{- \dfrac{\nr{\btheta - \mathbf{a}}^2}{2(1+\rho^2)}} + \exp\pr{- \dfrac{\nr{\btheta + \mathbf{a}}^2}{2(1+\rho^2)}}},
\end{align*}
and simply corresponds to a mixture of the two initial Gaussian densities but with respective variance now inflated by a factor $\rho^2$.
The one-dimensional approximate density $\pi_{\rho}(u)$ of $u = \B{a}^\top\btheta/\nr{\B{a}}$ is depicted in Figure \ref{fig:exp2} for $d = 60$ and compared to the true target $\pi(u)$.

\textit{Illustrations of Theorem \ref{PROP:COMPCOMPLEXITYTV_SINGLE}.}
We now illustrate the guidelines for $\rho$ and the number of iterations $t$, stated in Theorem \ref{PROP:COMPCOMPLEXITYTV_SINGLE}, to achieve an $\epsilon$-error in total variation distance.
To this purpose, we set $\epsilon = 0.1$, $d = 60$ and launched $2500$ independent runs of SGS.
The conditional distribution of $\B{z}$ given $\btheta$ is a mixture of two Gaussians with common covariance matrix $\B{\Sigma}=\rho^2/(1+\rho^2)\B{I}_d$, respective mean vectors $\boldsymbol{\mu}_1 = (\btheta + \B{a}\rho^2)/(1+\rho^2)$ and $\boldsymbol{\mu}_2 = (\btheta - \B{a}\rho^2)/(1+\rho^2)$ and respective weights $w_1 = 1$ and $w_2 = \exp(-4\btheta^\top\B{a}/(2(1+\rho^2))$.
We can sample exactly from this mixture by first drawing a Bernoulli random variable $B$ with probability $p = w_1/(w_1+w_2)$ and then setting $\B{z} = B (\boldsymbol{\xi} + \boldsymbol{\mu}_1) + (1-B) (\boldsymbol{\xi} + \boldsymbol{\mu}_2)$ where $\boldsymbol{\xi} \sim \mathcal{N}(\B{0}_d,\B{\Sigma})$.
In order to assess the relevance of the samples generated with SGS, we generated $2500$ independent samples directly from $\pi$ by an exact sampler similar to the one used to sample $\B{z}$.
To provide an illustration of the quality of the samples drawn with SGS, we computed the one-dimensional projection $u = \B{a}^\top\btheta/\nr{\B{a}}$ and showed its empirical distribution in Figure \ref{fig:exp2}. The empirical distribution of the samples drawn using SGS is indeed close to $\pi$ and is visually indistinguishable from the one of the exact samples.

\begin{table}
    \centering
    \begin{tabular}{cllllllll}
        \thickhline
        Dimension $d$ & 4 & 8 & 12 & 16 & 20 & 30 & 40 & 60 \\
        \hline
        $t_{\mathrm{mix}}(\epsilon;\nu)$ ($\times 10^3$) for SGS & 3 & 10 & 23 & 40 & 62 & 138 & 244 & 548 \\
        $t_{\mathrm{mix}}(\epsilon;\nu)$ ($\times 10^3$) for ULA & 29 & 87 & 184 & 330 & 532 & 1,350 & 2,729 & 7,742 \\
        Efficiency of SGS w.r.t. ULA & 10.8 & 8.6 & 8.2 & 8.3 & 8.6 & 9.8 & 11.1 & 14.1 \\
        \thickhline
        CPU time [s] for SGS & 1 & 7 & 29 & 62 & 114 & 335 & 749 & 2,416 \\
        CPU time [s] for ULA & 6 & 31 & 135 & 302 &
       589 & 1,974 & 4,766 & 15,096 \\
        Efficiency of SGS w.r.t. ULA & 5.6 & 4.6 & 4.7 & 4.9 & 5.2 & 5.9 & 6.4 & 6.2 \\
        \thickhline
    \end{tabular}
    \caption{Gaussian mixture. Comparison between SGS and ULA for a prescribed precision $\epsilon = 0.1$. For SGS, $t_{\mathrm{mix}}(\epsilon;\nu)$ has been computed by using Theorem \ref{PROP:COMPCOMPLEXITYTV_SINGLE} while for ULA, the mixing time bound derived in \citet[Corollary 1]{Dalalyan2017} has been used.
    CPU time information corresponds to the running time necessary to draw $10^3$ independent samples having a distribution at most $\epsilon$ total variation distance from $\pi$.}
    \label{table:exp2}
\end{table}

\textit{Computational complexity of SGS.}
We now verify empirically the computational complexity of SGS, that is the number of iterations and the overall running time for generating samples with some prescribed precision $\epsilon$. We compare this complexity to that of ULA \citep{Dalalyan2017}.
Starting from the same initial distribution $\nu = \mathcal{N}(\bthetastar,M^{-1}\B{I}_d)$ and with $\epsilon = 0.1$, Table \ref{table:exp2} reports the number of iterations $t_{\mathrm{mix}}(\epsilon;\nu)$ required, in theory, to obtain a sample whose distribution is at most $\epsilon$ in total variation from $\pi$ and the CPU time needed to generate $10^3$ such samples.
For ULA, $t_{\mathrm{mix}}(\epsilon;\nu)$ has been computed by using the mixing time bound derived in \citet[Corollary 1]{Dalalyan2017}.
We observe that for $d \in [4,60]$, both the number of iterations and the running time for generating $10^3$ independent samples with SGS are much smaller than that of ULA.

This second experiment confirms our theoretical statement that SGS is able to generate accurate samples for a reasonable computational budget compared to popular alternatives such as ULA.

\subsection{Bayesian Binary Logistic Regression}
\label{subsec:logistic_regression}

The previous two sections illustrated our theoretical results for a single splitting strategy and $\B{A}_1 = \B{I}_d$.
In this section, we consider a more challenging problem, namely Bayesian binary logistic regression.
This model involves $b > 1$ potential functions, matrices $\{\B{A}_{i}\}_{i\in[b]}$ not equal to the identity, and is such that the observations might be distributed over a set of $b$ nodes within a cluster.
As introduced in Section \ref{sec:split_Gibbs_sampler}, SGS is of interest for this scenario since it allows to sample from the posterior distribution of interest in such distributed environments.
In the sequel, we will also show the benefits of splitting $\B{A}_i\btheta$ instead of only splitting the variable of interest $\btheta$ as in \citet{Vono2019,Rendell2018}. 
This goal will be conducted by illustrating numerically our mixing time bounds and assessing the efficiency of the rejection sampling procedure (see Proposition \ref{prop:rejectionsamplingcomplexity}) used to sample the auxiliary variables $\bz_{1:b}$, in both cases.

\subsubsection{Problem Formulation}

As introduced in Example \ref{example:logistic}, the logistic regression problem considers a set of observed data $\{\B{x}_i,y_i\}_{i \in [n]}$ where the binary labels $y_i \in \{0,1\}$ are related to the unknown regression parameter $\btheta$ via the model \eqref{eq:logistic_likelihood}. In a Bayesian framework, a standard approach consists of assigning a zero-mean Gaussian prior to $\btheta$ with diagonal precision matrix $\B{\Sigma}^{-1} = \tau \B{I}_d$ as in \eqref{eq:logistic_prior}; see  \citet{Albert1993,Holmes2006}.
Instead, we set here $\B{\Sigma}^{-1} = \alpha \sum_{i=1}^n\B{x}_i\B{x}_i^\top$, with $\alpha = 3d/(\pi^2 n)$ which corresponds to a Zellner prior \citep{SabanesBove2011,Hanson2014}.
Such a choice leads to a posterior density $\pi(\btheta) \propto \exp(-U(\btheta))$ with
\begin{equation*}
    U(\btheta) = \sum_{i=1}^n y_i\B{x}_i^\top\btheta + \log\br{1 + \exp\pr{-\B{x}_i^\top\btheta}} + \frac{\alpha}{2}\nr{\B{x}_i^\top\btheta}^2.
\end{equation*}
Sampling from this posterior density can be conducted by exploiting the mixture representation of the binomial likelihood which involves the Polya-Gamma distribution, and then performing Gibbs sampling \citep{Polson2013}.
Nevertheless, although this algorithm has been shown to be uniformly ergodic w.r.t. the TV distance, the best known explicit result for its ergodicity constant degrades exponentially quickly with $n$ and $d$, see \cite{Choi2013}. 
We propose here to sample approximately from the posterior using SGS. We will consider and compare two splitting strategies.

\textit{Splitting strategy 1.} The first strategy sets $b=n$, $\B{A}_i = \B{x}_i^\top$ for $i \in [b]$ and leads to the approximate posterior density \eqref{eq:split_density_generalized} with
\begin{align*}
    U_i(z_i) = y_i z_i + \log\br{1 + \exp\pr{-z_i}} + \frac{\alpha}{2}z_i^2\ , \quad \forall i \in [b]\ .
\end{align*}
In this case, $U_i$ is $m_i$-strongly convex and $M_i$-smooth with $m_i = \alpha$ and $M_i = \alpha + 1/4$, respectively, and hence verifies \asstwo \ and \assfive. As detailed after Proposition \ref{proposition:2}, we can verify \assseven \ by centering $U_i$ with a simple linear shift.
In the sequel, we will assume that such a shifting has been performed which implies that all the assumptions in Theorem \ref{prop:compcomplexityTV} are verified.
The interest of this first splitting strategy is that the conditional posterior probability densities of $z_i$ given $\btheta$ are univariate and easy to sample.

\textit{Splitting strategy 2.} The second strategy mimics the one used by \cite{Rendell2018} and considers that the data $\{\B{x}_i,y_i\}_{i \in [n]}$ is divided into $b$ shards $\{\mathrm{D}_{i}\}_{i \in [b]}$.
For simplicity, we will assume that $n$ is a multiple of $b$ such that card($\mathrm{D}_i$) $= n/b$ for all $i \in [b]$. 
In contrast to the first splitting strategy, we use here $\B{A}_i = \B{I}_d$ for $i \in [b]$. This yields
\begin{equation*}
    U_i(\bz_i) = \sum_{j \in \mathrm{D}_i} y_j\B{x}_j^\top\bz_i + \log\br{1 + \exp\pr{-\B{x}_j^\top\bz_i}} + \frac{\alpha}{2}\nr{\B{x}_j^\top\bz_i}^2\ , \quad \forall i \in [b]\ .
\end{equation*}
Here $U_i$ is $m_i$-strongly convex and $M_i$-smooth with $m_i = \alpha \lambda_{\mathrm{min}}(\sum_{j \in \mathrm{D}_i}\B{x}_j\B{x}_j^\top)$ and $M_i = (\alpha + 1/4)\lambda_{\mathrm{max}}(\sum_{j \in \mathrm{D}_i}\B{x}_j\B{x}_j^\top)$, where $\lambda_{\mathrm{min}}(\B{M})$ and $\lambda_{\mathrm{max}}(\B{M})$ stand for the smallest and largest eigenvalues of a matrix $\B{M}$, respectively.
As before, we assume that an appropriate centering of $U_i$ has been performed to satisfy \assseven.
In some scenarios, such a splitting strategy is expected to be less efficient than the first one for two main reasons.
First, the conditional density of $\bz_i$ given $\btheta$ is $d$-dimensional and sampling from it is as difficult as sampling from $\pi$.
Second, the condition number $\kappa = \sum_i M_i/\sum_i m_i$ associated with this strategy (denoted $\kappa_2$) might be very large compared to the one associated to the splitting strategy 1 (denoted $\kappa_1$).
Indeed, the ratio of these two condition numbers is 
\begin{equation}
    \frac{\kappa_2}{\kappa_1} = \dfrac{\displaystyle\sum_{i=1}^b\lambda_{\mathrm{max}}\pr{\sum_{j \in \mathrm{D}_i}\B{x}_j\B{x}_j^\top}}{\displaystyle\sum_{i=1}^b\lambda_{\mathrm{min}}\pr{\sum_{j \in \mathrm{D}_i}\B{x}_j\B{x}_j^\top}}. \label{eq:ratio}
\end{equation}
This ratio is expected to be large when $d$ is large and the correlation between the covariates within each group is high. The splitting strategy 1 can be thought of as a preconditioning technique whose efficiency is measured by the ratio \eqref{eq:ratio}.

\subsubsection{Efficient Sampling with the SGS}

\begin{figure}
    \centering
    \mbox{{\includegraphics[scale=0.35]{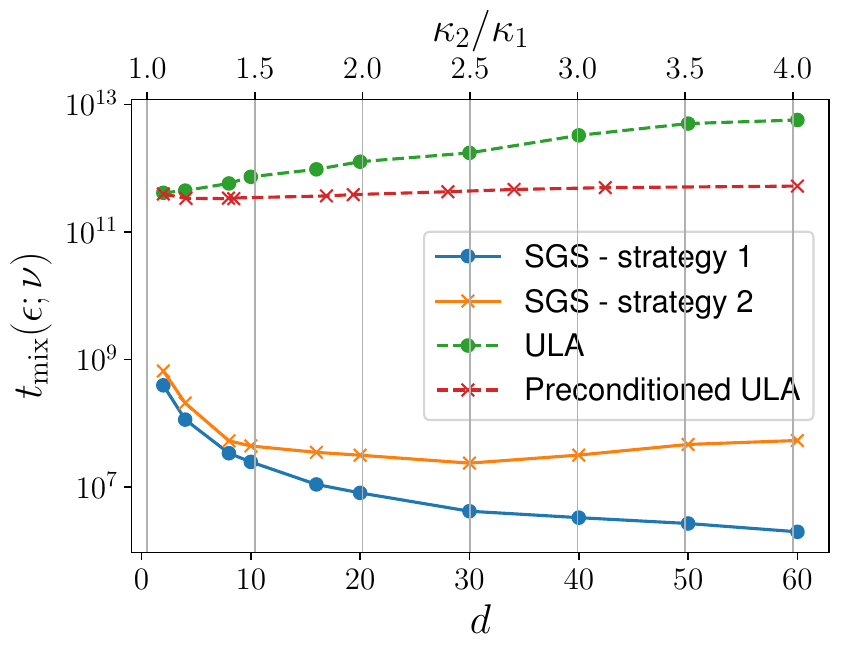}}}
    \mbox{{\includegraphics[scale=0.35]{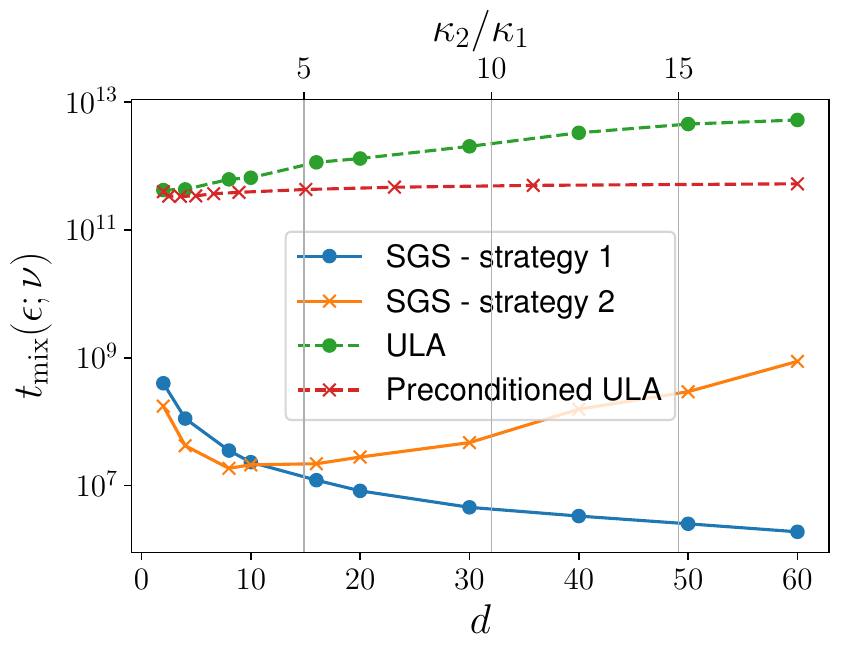}}}
    \mbox{{\includegraphics[scale=0.35]{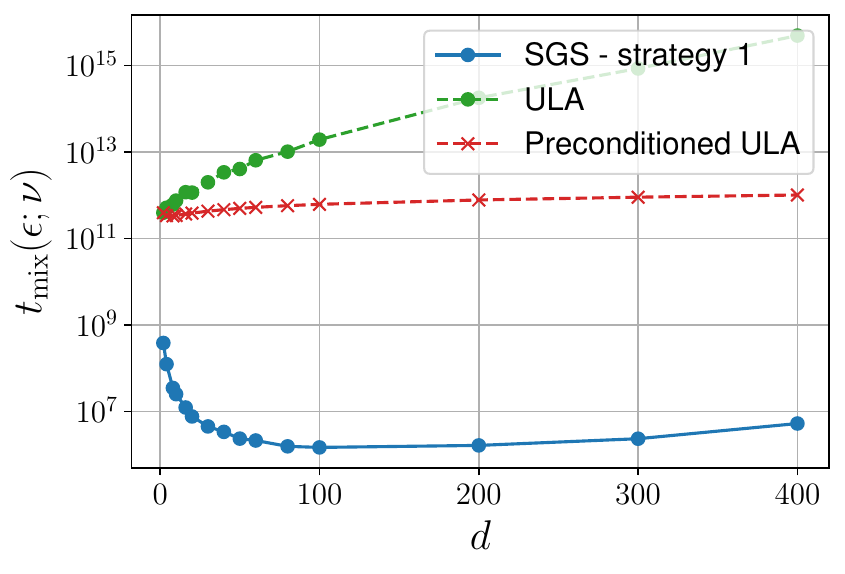}}}
  \caption{Logistic regression. (left and middle) Behavior of the mixing times of the two splitting strategies w.r.t. the ratio $\kappa_1/\kappa_2$. For strategy 2, $b=5$ (left) and $b=10$ (middle), while for strategy 1, $b=n=1000$ on both figures.
  (right) Behavior of the splitting strategy 1 in higher dimensions.}
  \label{fig:exp3}
\end{figure}

For this experiment, also considered in \cite{Dalalyan2017}, we generated a synthetic data set $\{\B{x}_i,y_i\}_{i=1}^n$ by drawing the covariates $\B{x}_i$ from a Rademacher distribution before normalizing the latter such that $\nr{\B{x}_i} = 1$.
Each binary label $y_i$ was then drawn from a Bernoulli distribution with probability of success equal to $\sigma(\B{x}_i^\top\btheta_{\mathrm{true}})$, where $\sigma(\cdot)$ is the sigmoid function and $\btheta_{\mathrm{true}} = \B{1}_d$.

\textit{Mixing times.} In this first sub-experiment, we compare our mixing time bounds in Theorem \ref{prop:compcomplexityTV} for the two splitting strategies detailed previously.
We set $\epsilon = 0.01$, $n=1,000$ and let $b$ vary from $b=5$ to $b=10$ for the splitting strategy 2.
The theoretical $\epsilon$-mixing times for the TV distance associated to the two splitting strategies are reported in Figure \ref{fig:exp3}.
To give an idea of the order of magnitude of these mixing times, the ones associated to ULA and its preconditioned version \citep{Dalalyan2017} using the same starting distribution $\nu$ are also displayed.
As expected, the splitting strategy 1 needs, in theory, less iterations than the splitting strategy 2 to achieve a prescribed precision $\epsilon$ when the ratio $\kappa_2/\kappa_1$ is large.
Finally, it is clear that the mixing times of SGS are again competitive compared to the ones derived in the recent literature for other MCMC algorithms.

\textit{Efficiency of the rejection sampling scheme.} In this second sub-experiment, we complement the previous analysis by showing that drawing the auxiliary variables $\B{z}_i$ given $\btheta$ can indeed be conducted efficiently with rejection sampling.
For each instance of SGS (with either splitting strategy 1 or 2), we used the rejection sampling scheme detailed in Proposition \ref{prop:rejectionsamplingcomplexity} in Section \ref{subsubsec:sampling_zi} to sample the auxiliary variables.
To this purpose, we considered different scenarios where $d \in \{2,10,50\}$ and $n \in \{200,10^3,10^4\}$.
We ran SGS over $T = 100$ iterations and averaged the number of rejection steps over these iterations for each auxiliary variable $\B{z}_i$.
In Table \ref{table:exp3}, we reported the largest average number of rejection steps per iteration obtained among the $b$ auxiliary variables for each splitting strategy.
For the nine different scenarios, the average number of rejection steps per iteration is near 1 which confirms the theoretical results of Proposition \ref{prop:rejectionsamplingcomplexity}.
Overall, SGS appears to be a promising and efficient approach to sample from smooth and strongly log-concave distributions. 

\begin{table}
    \centering
    \begin{tabular}{clllllllllllll}
        \thickhline
        $d$ & \multicolumn{3}{c}{2} & \multicolumn{3}{c}{10} & \multicolumn{3}{c}{50} \\
        $n$ & 200 & 1,000 & 10,000 & 200 & 1,000 & 10,000 & 200 & 1,000 & 10,000 \\
        \hline
        SGS 1 ($b=n$)& 1.04 & 1.04 & 1.03 & 1.03 & 1.05 & 1.03 & 1.06 & 1.05 & 1.03 \\
        SGS 2 ($b=2$) & 1.13 & 1.16 & 1.34 & 1.26 & 1.08 & 1.16 & 1.06 & 1.27 & 1.03  \\
        SGS 2 ($b=5$) & 1.22 & 1.14 & 1.26 & 1.08 & 1.14 & 1.08 & 1.00 & 1.10 & 1.34 \\
        SGS 2 ($b=10$) & 1.12 & 1.07 & 1.05 & 1.31 & 1.17 & 1.16 & 1.41 & 1.07 & 1.13  \\
        \thickhline
    \end{tabular}
        \caption{Logistic regression. Average number of samples proposed until one is accepted per iteration for SGS 1 (associated to splitting strategy 1) and SGS 2 (associated to splitting strategy 2).}
    \label{table:exp3}
\end{table}

\section{Conclusion}

In this paper, we have provided a detailed theoretical study of a recent and promising MCMC algorithm, namely SGS, which is amenable to a distributed implementation and shares strong similarities with quadratic penalty approaches in optimization.
Under a strong log-concavity assumption, we have obtained explicit dimension-free convergence rates for this sampler under both Wasserstein and total variation distances.
Combined with quantitative bounds on the bias induced by this algorithm, we have derived explicit bounds on its mixing time under reasonable assumptions which can be easily verified in practice.
In addition to be amenable to distributed and parallel computations, these results showed that SGS can compete and even improve upon standard MCMC schemes in terms of computational complexity.
Our theoretical results have been supported with numerical illustrations which confirmed the efficiency of SGS even on a serial computer.

There are a few additional interesting questions to address. 
All our theoretical results assume that the auxiliary variables $\B{z}_i$ are drawn from the exact conditional probability density at each iteration of SGS. Although this is possible for interesting models such as logistic regression, one might have to sample approximately these variables using Metropolis-Hastings or proximal MCMC scheme \citep{Pereyra2016B,Durmus2018,Vargas2019} in more complex scenarios and it would be interesting to extend our results to such settings.
Another interesting extension would be to consider whether using a sequence $\{\rho_t\}_{t \in \mathbb{N}}$ instead of a fixed parameter $\rho$ could be beneficial by determining convergence rates in this scenario.

\acks{This material is based upon work supported in part by the U.S. Army Research Laboratory and the U.S. Army Research Office, by the U.K. Ministry of Defence (MoD), and by the U.K. Engineering and Physical Research Council (EPSRC) under grant number EP/R013616/1. It is also supported by EPSRC grants EP/R034710/1 and EP/R018561/1. The authors thank the GdR ISIS and R\'emi Bardenet from Universit\'e de Lille for funding MV's visit to Oxford through the internationalization grant ``Effet tunnel''.
Part of this work has been supported by the ANR-3IA Artificial and Natural Intelligence Toulouse Institute (ANITI).
We thank Solomon Jacobs and Andreas Eberle for pointing out an error in the proof of Proposition \ref{proposition:2} in a previous version of this paper.}

%\newpage

\appendix
\renewcommand{\thesection}{Appendix \Alph{section}}
\renewcommand{\thesubsection}{Appendix \Alph{section}.\arabic{subsection}.}
\renewcommand{\thesubsubsection}{Appendix \Alph{section}.\arabic{subsection}.\arabic{subsubsection}.}

\section{Additional Details and Proofs for Section \ref{sec:split_Gibbs_sampler}}
\label{appendix:sec_2}

This section aims at proving the results claimed in Section \ref{sec:split_Gibbs_sampler}.

\subsection{Integrability of $\pi_{\rho}$ and Ergodicity of SGS}
\label{appendix:integrability_ergodicity}
\begin{proof}[Proof of Proposition \ref{prop:integrability_ergodicity}]
Let $U(\btheta,\B{z}_{1:b})$ be defined as in \eqref{eq:split_density_generalized}, then
by repeated application of Tonelli's theorem \citep{tonelli1909sull} to integrate out $\bz_1,\ldots \bz_b$, we have
\begin{align*}
&\int_{\btheta, \B{z}_{1:b}} \exp(-U(\btheta,\B{z}_{1:b}))\mathrm{d}\btheta \mathrm{d}\B{z}_{1:b}
=\int_{\btheta, \B{z}_{1:b}}\exp\left(- \sum_{i=1}^b U_i(\B{z}_{i}) + \dfrac{\nr{\B{z}_i-\B{A}_i\btheta}^2}{2\rho^2}\right)\mathrm{d}\btheta \mathrm{d}\B{z}_{1:b}\\
&=\int_{\btheta}\exp\left(-\sum_{i=1}^b U_i^{\rho}(\B{A}_i\btheta)\right)\mathrm{d}\btheta=\int_{\btheta}\exp(-U^{\rho}(\btheta))\mathrm{d}\btheta.
\end{align*}
By Assumption \asszero, $\exp(-U^{\rho}(\btheta))$ is integrable, hence $\exp(-U(\btheta,\B{z}_{1:b})$ is also integrable, and $\joint(\btheta,\B{z}_{1:b})$ is a probability density.
The $\pi$-irreducibility and aperiodicity of SGS follows because SGS defined on the extended state space including $\bz_{1:b}$ is a Gibbs sampler with systematic scan, and it satisfies the positivity condition of Gibbs sampling (since the densities are always positive); see for instance \cite{roberts1994simple}.
\end{proof}
\begin{proof}[Proof of Proposition \ref{prop:integrability_simpler_condition}]
Note that we have
\begin{align*}
\exp(-U_i^{\rho}(\B{w}_i))&=\int_{\B{z}_i\in \R^d}\exp\l(-U_i(\B{z}_i)-\frac{\|\B{z}_i-\B{w}_i\|^2}{2\rho^2}\r)\cdot \frac{\mathrm{d} \B{z}_i}{(2\pi \rho^2)^{d_i/2}}\\
&\le \int_{\B{z}_i\in \R^d}\exp\l(-V_i(\B{z}_i)-\frac{\|\B{z}_i-\B{w}_i\|^2}{2\rho^2}\r)\cdot \frac{\mathrm{d} \B{z}_i}{(2\pi \rho^2)^{d_i/2}}\\
&\le \exp\l(-V_i(\B{w}_i)\r) \cdot \int_{\B{z}_i\in \R^d}\exp\l(L_i\|\B{z}_i-\B{w}_i\| -\frac{\|\B{z}_i-\B{w}_i\|^2}{2\rho^2}\r)\cdot \frac{\mathrm{d} \B{z}_i}{(2\pi \rho^2)^{d_i/2}}.
\end{align*}
It is easy to show that 
$L_i\|\B{z}_i-\B{w}_i\| -\frac{\|\B{z}_i-\B{w}_i\|^2}{2\rho^2}\le -\frac{\|\B{z}_i-\B{w}_i\|^2}{4\rho^2}$ whenever $\|\B{z}_i-\B{w}_i\|\ge 2\rho^2 L_i$, hence this integral is finite, and 
$\exp(-U_i^{\rho}(\B{w}_i))\le \exp\l(-V_i(\B{w}_i)\r) C_i$ for some $C_i<\infty$. Hence, we have that $\exp(-U^{\rho}(\btheta))\le \exp\l(-\sum_{j\in [b]} V_j(\B{A}_j \btheta)\r)\cdot \prod_{j\in [b]} C_j$. The integrability of $\exp(-U^{\rho}(\btheta))$ now follows from our assumption that $\exp\l(-\sum_{j\in [b]} V_j(\B{A}_j \btheta)\r)$ is integrable.
\end{proof}

\section{Proofs for the Results of Section \ref{sec:non_asymptotic_properties}}
\label{appendix:A}

This section gives the proofs and technical details associated to the results presented in Section \ref{sec:non_asymptotic_properties}.

\subsection{Non-Asymptotic Bound for $I(U,U^{\rho})$}
\label{appendix:A_1}
In this section, we are going to bound the bias of the stationary distribution of SGS ($\pi_{\rho}$) from $\pi$. We start by the proof of Propostion \ref{prop:IUUrho}, which shows that we can bound the total variation, KL and Wasserstein-2 distances between $\pi_{\rho}$ and $\pi$ in terms of $ I(U,U_{\rho})$.
\begin{proof}[Proof of Proposition \ref{prop:IUUrho}]
By using the notations $f(\btheta)_{-} = -\min(f(\btheta),0)$ and $f(\btheta)_{+} = \max(f(\btheta),0)$, note that 
\begin{align}
    \nr{\marginal-\pi}_{\mathrm{TV}} &=
    \dfrac{1}{2}\int_{\btheta \in \mathbb{R}^d} |\pi(\btheta) - \marginal(\btheta)|\mathrm{d}\btheta \nonumber\\
    &=\int_{\btheta \in \mathbb{R}^d} \l(\pi(\btheta) - \marginal(\btheta)\r)_{-}\mathrm{d}\btheta
    =\int_{\btheta \in \mathbb{R}^d} \l(\pi(\btheta) - \marginal(\btheta)\r)_{+}\mathrm{d}\btheta\nonumber\\
    &=\int_{\btheta \in \mathbb{R}^d} \pi(\btheta)\l( 1 - \frac{\marginal(\btheta)}{\pi(\btheta)}\r)_{+}\mathrm{d}\btheta\label{eq:TVpospartbnd},
\end{align}
since 
\begin{align*}
    &\int_{\btheta \in \mathbb{R}^d} \l(\pi(\btheta) - \marginal(\btheta)\r)_{+}\mathrm{d}\btheta-\int_{\btheta \in \mathbb{R}^d} \l(\pi(\btheta) - \marginal(\btheta)\r)_{-}\mathrm{d}\btheta=\int_{\btheta \in \mathbb{R}^d} \l(\pi(\btheta) - \marginal(\btheta)\r)\mathrm{d}\btheta=0,\\
    &|\pi(\btheta) - \marginal(\btheta)|=(\pi(\btheta) - \marginal(\btheta))_+ + (\pi(\btheta) - \marginal(\btheta))_-.
\end{align*}
Using the definitions of $\pi$ and $\pi_{\rho}$, we have
\begin{align}
    \nr{\marginal-\pi}_{\mathrm{TV}}&= 
    \int_{\btheta \in \mathbb{R}^d} \pi(\btheta)\l(1 - \exp\l(U(\btheta)-U^{\rho}(\btheta)\r) 
    \cdot \frac{Z_{\pi}}{Z_{\pi_{\rho}}}
    \r)_{+}\mathrm{d}\btheta \nonumber
    \\
   \intertext{using the fact that $(1-\exp(x))_+ \le x_{-}$ for any $x\in \R$,}
   \label{eq:TVpipirhoZpiZhroUUrho2}&\le \int_{\btheta \in \mathbb{R}^d} \pi(\btheta)\l(
   \log\l(\frac{Z_{\pi}}{Z_{\pi_{\rho}}}\r)+
   U(\btheta)-U^{\rho}(\btheta) 
    \r)_{-}\mathrm{d}\btheta
    \\
    \label{eq:TVintermediate}&\le \l(\log\l(\frac{Z_{\pi_{\rho}}}{Z_{\pi}}\r)\r)_+
    +\int_{\btheta \in \mathbb{R}^d} \pi(\btheta)
   \l(U^{\rho}(\btheta)-U(\btheta)
    \r)_{+} \mathrm{d}\btheta = I(U,U_{\rho}),
\end{align}
hence the TV bound follows. For KL-divergence, note that
\eqref{eq:TVpipirhoZpiZhroUUrho2} satisfies that
\[\kld{\pi}{\pi_{\rho}}=
\int_{\mathbb{R}^d} \pi(\btheta) \log\left(\frac{\pi(\btheta)}{\pi_{\rho}(\btheta)}\right)\mathrm{d} \btheta
\le \int_{\mathbb{R}^d} \pi(\btheta)\l(
   \log\l(\frac{Z_{\pi}}{Z_{\pi_{\rho}}}\r)+
   U^{\rho}(\btheta)-U(\btheta)
    \r)_{+}\mathrm{d}\btheta,
\]
hence this is also bounded by $I(U,U_{\rho})$.
Finally, the Wasserstein bound follows from the KL bound and Lemma 9 of \cite{pmlr-v83-cheng18a}.
\end{proof}

Now we will state a few definitions and prove some auxiliary lemmas, and then prove Proposition \ref{proposition:2}. 
We have 
$U(\btheta)=\sum_{i=1}^{b} U_i(\B{A}_i\btheta)$, and 
\[\pi(\btheta)=\frac{\exp(-U(\btheta))}{Z_{\pi}}, \text{ for a normalising constant } Z_{\pi}=\int_{\btheta}\exp(-U(\btheta)) \mathrm{d}\btheta.\]
Similarly, by Proposition \ref{prop:integrability_ergodicity}, we have \[\marginal(\btheta)=\frac{\exp(-U^{\rho}(\btheta))}{Z_{\pi_{\rho}}}.\]
The following lemma states some bounds on $U(\btheta)-U^{\rho}(\btheta)$.

\begin{lemma}\label{lem:UrhoUratiobnd}
Let
\begin{align*}%\label{eq:UrhoUratiobnd}
     \overline{B}(\btheta) &\coloneqq \dfrac{\rho^2}{2}\sum_{i=1}^{b}\nr{\nabla U_i(\B{A}_i\btheta)}^2,\\
    \underline{B}(\btheta) &\coloneqq 
    \sum_{i=1}^{b}\pr{\dfrac{\rho^2}{2(1+\rho^2M_i)}\nr{\nabla U_i(\B{A}_i\btheta)}^2 -\frac{d_i}{2} \log(1+\rho^2M_i)}.
\end{align*}
Then assuming \asszero \ and \asstwo, we have $\underline{B}(\btheta)\le U(\btheta)-U^{\rho}(\btheta)$.
Assuming \asszero, \asstwo \ and \assfive, we have $U(\btheta)-U^{\rho}(\btheta)\le \overline{B}(\btheta)$.
\end{lemma}
\begin{proof}
First, note that
\[\exp(U(\btheta)-U^{\rho}(\btheta))=\exp\l(\sum_{i=1}^{b} 
    \l(U_i(\B{A}_i\btheta)- U_i^{\rho}(\B{A}_i\btheta)\r) \r).\]
From \eqref{eq:Urhodef}, it is clear that
\begin{equation}\label{eq:Uirhodiff}
\exp(U_i(\B{A}_i\btheta)-U_i^{\rho}(\B{A}_i\btheta))= \int_{\B{z}_i\in \R^d}\exp\l(U_i(\B{A}_i\btheta)-U_i(\B{z}_i)-\frac{\|\B{z}_i-\B{A}_i\btheta\|^2}{2\rho^2}\r)\cdot \frac{\mathrm{d} \B{z}_i}{(2\pi \rho^2)^{d_i/2}}.
\end{equation}
Using \asstwo, and second order Taylor expansion, for each $i \in [b]$, we have
\begin{align*}
      U_i(\B{A}_i\btheta) - U_i(\B{z}_i) \geq \nabla U_i(\B{A}_i\btheta)^\top(\B{A}_i\btheta-\B{z}_i)-\frac{M_i}{2} \|\B{A}_i\btheta-\B{z}_i\|^2. 
  \end{align*}
  Hence, using \eqref{eq:Uirhodiff}, we have
  \begin{align}
      &\exp\l(\sum_{i=1}^{b} 
    \l(U_i(\B{A}_i\btheta)- U_i^{\rho}(\B{A}_i\btheta)\r) \r) \nonumber\\
    &\geq \prod_{i=1}^b (2\pi\rho^2)^{-d_i/2}\int_{\B{z}_i \in \mathbb{R}^{d_i}}\exp\pr{\nabla U_i(\B{A}_i\btheta)^\top(\B{A}_i\btheta-\B{z}_i) - \pr{\dfrac{1+\rho^2M_i}{2\rho^2}}\nr{\B{A}_i\btheta-\B{z}_i}^2}\mathrm{d}\B{z}_i \nonumber\\
      &= \prod_{i=1}^b\pr{\exp\pr{\dfrac{\rho^2}{2(1+\rho^2M_i)}\nr{\nabla U_i(\B{A}_i\btheta)}^2}\pr{\dfrac{1}{1+\rho^2M_i}}^{d_i /2}} \label{eq:proof_theorem_2_4_needed_for_Lemma_26}\\
      &=\exp\pr{\sum_{i=1}^{b}\pr{\dfrac{\rho^2}{2(1+\rho^2M_i)}\nr{\nabla U_i(\B{A}_i\btheta)}^2 -\frac{d_i}{2} \log(1+\rho^2M_i)} }
      = \exp(\underline{B}(\btheta)), \label{eq:proof_theorem_2_4}
  \end{align}  
  hence the lower bound follows.

For the upper bound, we now use \assfive \ (convexity of the individual potential functions $U_i$ for $i \in [b]$), which by Taylor expansion yields that for every $i \in [b]$,
 \begin{align*}
      U_i(\B{A}_i\btheta) - U_i(\B{z}_i) \leq \nabla U_i(\B{A}_i\btheta)^\top(\B{A}_i\btheta-\B{z}_i)-\frac{m_i}{2}\|\B{A}_i\btheta-\B{z}_i\|^2.
  \end{align*}
Then, it follows that
    \begin{align}
      &\exp(U(\btheta)-U^{\rho}(\btheta))=\exp\l(\sum_{i=1}^{b} 
    \l(U_i(\B{A}_i\btheta)- U_i^{\rho}(\B{A}_i\btheta)\r) \r) \nonumber\\
      &\leq \prod_{i=1}^{b}(2\pi\rho^2)^{-d_i/2}\int_{\B{z}_i \in \mathbb{R}^{d_i}}\exp\pr{\nabla U_i(\B{A}_i\btheta)^\top(\B{A}_i\btheta-\B{z}_i) - \dfrac{1+\rho^2m_i}{2\rho^2}\nr{\B{A}_i\btheta-\B{z}_i}^2}\mathrm{d}\B{z}_i \nonumber\\
      &=\prod_{i=1}^{b}\frac{1}{\l(1+\rho^2 m_i\r)^{d_i/2}} \cdot \exp\l(\sum_{i=1}^b \frac{\rho^2 \nr{\grad U_i(\B{A}_i\btheta)}^2}{2 (1+\rho^2 m_i)}\r)
      \label{eq:proof_theorem_2_2_needed_for_Lemma_26}
      \\
      &\le \exp\pr{\dfrac{\rho^2}{2}\sum_{i=1}^{b}\nr{\nabla U_i(\B{A}_i\btheta)}^2}
      = \exp(\overline{B}(\btheta)),
      \label{eq:proof_theorem_2_2}
  \end{align}
  hence the upper bound follows.
\end{proof}
\begin{lemma}\label{lem:betaLipschitz}
Suppose that Assumptions \asstwo\, and \assfive\, hold. Let 
\[\beta(\btheta):=\left(\sum_{i=1}^{b}\nr{\nabla U_i(\B{A}_i\btheta)}^2\right)^{1/2},\] and $\B{A}$ be the $(d_1+\ldots d_b)\times d$ matrix created by stacking $\B{A}_1,\ldots ,\B{A}_b$ one upon another, starting with $\B{A}_1$ on the top and ending with $\B{A}_b$. Then $\beta$ is a Lipschitz function with respect to the Euclidean distance, with Lipschitz constant
\begin{equation*}\label{eq:Lbeta}
    L_{\beta}=\|\B{A}^\top \B{A}\|^{1/2} \max_{i\le b}M_i.
\end{equation*}
\end{lemma}
\begin{proof}
Assuming that $\btheta\neq \bthetastar$, we have $\beta(\btheta)>0$, and thus $\grad \beta(\btheta)$ exists, and it has the form
\begin{align*}
    \grad \beta(\btheta)=\frac{\sum_{i=1}^{b} \B{A}_i^\top \grad^2 U_i(\B{A}_i \btheta) \grad U_i(\B{A}_i \btheta)}{\left(\sum_{i=1}^b \|\grad U_i(\B{A}_i \btheta)\|^2\right)^{1/2}}.
\end{align*}
Let $\B{w}\coloneqq (\grad U_1(\B{A}_1\btheta),\ldots,  \grad U_b(\B{A}_b\btheta) )\in \R^{d_1+\ldots+d_b}$, and $\B{D}\coloneqq \mathrm{diag}(\grad^2 U_1(\B{A}_1 \btheta), \ldots, \grad^2 U_b(\B{A}_b \btheta))$ (a block matrix with diagonal blocks corresponding $\grad^2 U_1(\B{A}_1 \btheta), \ldots, \grad^2 U_b(\B{A}_b \btheta)$). Then we have 
\begin{align*}
    \grad \beta(\btheta)=&\frac{\B{A}^\top \B{D} \B{w}}{\|\B{w}\|},\text{ hence}\\
    \|\grad \beta(\btheta)\|^2&=\frac{\B{w}^\top \B{D} \B{A} \B{A}^\top  \B{D} \B{w}}{\|\B{w}\|^2}\\
    &\le \|\B{D}\|^2 \|\B{A} \B{A}^\top\|\\
    &\le \|\B{A}^\top \B{A}\| (\max_{i\le b}M_i)^2,
\end{align*}
so $\|\grad \beta(\btheta)\|\le L_{\beta}$. Since this bound holds everywhere except possibly at $\bthetastar$, the Lipschitz property  is easy to show by a limiting argument.
\end{proof}

The next result is a technical lemma that will be used in the proof.
\begin{lemma}\label{lem:f2momgenbnd}
Suppose that $\mu(\btheta) \propto \exp(-U(\btheta))$ is a density on $\R^d$ satisfying that $U$ is twice continuously differentiable, and $\grad^2 U(\btheta)\succeq m \B{I}_d$ for every $\btheta\in \R^d$ for some $m>0$. Let $f:\R^d
\to \R$ be an $L$-Lipschitz function with $\E_{\mu}(f)=0$. Then for $0\le \lambda\le \frac{m}{6L^2}$, we have
\begin{equation}\label{eq:f2momgenbnd}\E_{\mu} \left[\econst^{ \lambda (f^2-\E_{\mu}(f^2))^2}\right] \le \econst^{\frac{4\lambda^2  L^4 }{m^2}}.\end{equation}
\end{lemma}
\begin{proof}
First, we are going to assume that $f$ is continuously differentiable. Then by the Lipschitz property, we have $\|\grad f(\btheta)\|\le L$ for every $\btheta \in \R^d$, and using Corollary 3.4 of \cite{huang2020nonlinear}, we can bound the moments of $f^2(\btheta)-\E_{\mu}(f^2)$ as follows.
For integers $p\ge 2$, by using this Corollary 3.4 twice 
(first on the function $f^2-\E_{\mu}(f^2)$, and then on $f$ in the third line), we have \begin{align*}\nonumber
\E_{\mu} \left[(f^2-\E_{\mu}(f^2))^p\right] &\le \E_{\mu} \left[\l|f^2-\E_{\mu}(f^2)\r|^p\right] \\
&\le \frac{1}{m^{p/2}} (p-1)^{p/2} \E_{\mu} \left[(2\|\grad f\|^2 f^2)^{p/2}\right]\\
&\le \frac{1}{m^{p/2}} 2^{p/2} (p-1)^{p/2} L^{p}\E_{\mu} \left[ \l|f\r|^p\right]\\
&\le \frac{1}{m^{p}} 2^{p/2} (p-1)^{p} L^{2p}.
\label{eq:mombndodd}
\end{align*}
Therefore, using the power series representation of the exponential function, we have
\begin{align}\nonumber
&\E_{\mu} \left[\econst^{ \lambda (f^2-\E_{\mu}(f^2))^2}\right] 
\le 1 +\sum_{p=2}^{\infty}   \l(\frac{\sqrt{2} \lambda L^2 }{m}\r)^{p}  \frac{(p-1)^{p}}{p!}\\
&\nonumber \le 1+\frac{1}{2}\l(\frac{\sqrt{2} \lambda L^2 }{m}\r)^{2}+
\sum_{p=3}^{\infty}   \l(\frac{\sqrt{2} \lambda L^2 }{m}\r)^{p}  \frac{(p-1)^{p}}{(p/\econst)^p \sqrt{2\pi p}}\\
&\nonumber \le  1+\frac{1}{2}\l(\frac{\sqrt{2} \lambda L^2 }{m}\r)^{2}+\sum_{p=3}^{\infty}   \l(\frac{\sqrt{2} \lambda L^2 }{m}\r)^{p}  \frac{\econst^{p-1}}{\sqrt{2 \pi\cdot 3}},
\end{align}
where we have used the facts that $p!\ge (p/\econst)^p \sqrt{2\pi p}$ by \cite{robbins1955remark}, and $((p-1)/p)^{p}\le 1/\econst$ for every $p\ge 1$. By summing up the geometric series, we have that for $0\le \frac{\econst\sqrt{2} \lambda L^2 }{m}<1$,
\[\E_{\mu} \left[\econst^{ \lambda (f^2-\E_{\mu}(f^2))^2}\right] 
\le  1+\frac{1}{2}\l(\frac{\sqrt{2} \lambda L^2 }{m}\r)^{2}+\frac{\l(\frac{\sqrt{2} \lambda L^2 }{m}\r)^{3}\econst^2}{\sqrt{6\pi}}\cdot \frac{1}{1- \frac{\econst\sqrt{2} \lambda L^2 }{m}}.\]
It is easy to show that for $0
\le \frac{\lambda L^2 }{m}\le \frac{1}{6}$, the above sum is bounded by $\econst^{\frac{4\lambda^2  L^4 }{m^2}}$, implying \eqref{eq:f2momgenbnd}. Finally, the proof without assuming differentiability of $f$ follows by using Theorem 6 of \cite{azagra2007smooth}, and a limiting argument using the dominated convergence theorem.
\end{proof}

The next result presents some bounds on certain moment generating functions.

\begin{lemma}\label{lem:momgenbeta2}
Suppose that Assumptions \asszero, \asstwo, \assfive, and \assseven\, hold. 
Let 
\[
\sigma^2_U:=L_{\beta}^2 \cdot m_U^{-1}=\|\B{A}^\top \B{A}\| (\max_{i\le b}M_i)^2 \cdot m_U^{-1}.\]
Then for $0\le s\le \frac{1}{12 \sigma^2_U}$, we have
\[\int_{\btheta}\pi(\btheta) e^{s\cdot \beta^2(\btheta)} \mathrm{d} \btheta \le \econst^{s \E_{\pi}(\beta^2)}\cdot \econst^{s^2 (8 \sigma^4_U+4 (\E_{\pi}(\beta))^2 \sigma^2_U)}.\]
Moreover, we have $(\E_{\pi}(\beta))^2\le \E_{\pi}(\beta^2)\le d \sigma^2_U$.
\end{lemma}
\begin{proof}
%Let $\E_{\pi}(f)=\int_{\btheta\in \R^d}\pi(\btheta)f(\btheta) \mathrm{d}\btheta$ denote the expectation of a function $f$ with respect to $\pi$. 
From Assumption \assfive, we have
\begin{equation}\label{eq:Uhesslwbnd}\grad^2 U(\btheta)=\sum_{i=1}^b \B{A}_i^\top \grad^2 U_i(\B{A}_i\btheta) \B{A}_i\succeq \sum_{i=1}^b m_i \B{A}_i^\top \B{A}_i\succeq \lambda_{\min}\left(\sum_{i=1}^b m_i \B{A}_i^\top \B{A}_i\right) \B{I}_d=m_U \B{I}_d.\end{equation}
% %%%%%%%%
 By Theorem 5.2 of \cite{Ledouxconcentrationofmeasure}, $\pi$ satisfies a log-Sobolev inequality with constant $C:=m_U^{-1}$. Therefore, by Herbst's argument (see equation (5.8) on page 95 of \cite{Ledouxconcentrationofmeasure}), and the $L_{\beta}$-Lipschitz property of the function $\beta$ shown in Lemma \ref{lem:betaLipschitz}, for every $\lambda\in \R$, we have  
 \begin{align}\label{eq:genfuncbeta}
 &\E_{\pi}(\econst^{\lambda (\beta-\E_{\pi}(\beta))})\le \econst^{C \lambda^2 L_{\beta}^2/2}.
\end{align}
By the decomposition $\beta^2(\btheta)=\E_{\pi}(\beta)^2+(\beta(\btheta)-\E_{\pi}(\beta))^2+2 \E_{\pi}(\beta) (\beta(\btheta)-\E_{\pi}(\beta))$, we have
\begin{align*}
\E_{\pi}(\econst^{s \beta^2})&=
\econst^{s \E_{\pi}(\beta)^2} \cdot \E_{\pi}\l(\econst^{s (\beta-\E_{\pi}(\beta))^2} \cdot \econst^{2 s \E_{\pi}(\beta) (\beta-\E_{\pi}(\beta))} \r)\nonumber\\
\intertext{ by the Cauchy-Schwarz inequality}
&\le \econst^{s \E_{\pi}(\beta)^2} \cdot \l[\E_{\pi}\l(\econst^{2s (\beta-\E_{\pi}(\beta))^2}\r)\r]^{1/2} \cdot \l[\E_{\pi}\l(\econst^{4 s \E_{\pi}(\beta) (\beta-\E_{\pi}(\beta))} \r)\r]^{1/2}.
\end{align*}
Lemma \ref{lem:f2momgenbnd} implies that for $0\le s\le \frac{1}{12 \sigma^2_U}$, we have
\begin{equation}\label{eq:momsbeta2}
\E_{\pi}\l(\econst^{2s (\beta(\btheta)-\E_{\pi}(\beta))^2 } \r)\le  \econst^{2s \E_{\pi}[(\beta(\btheta)-\E_{\pi}(\beta))^2]} \cdot \econst^{16 s^2 \sigma^4_U}.\end{equation}
%%%%%%%%
By \eqref{eq:genfuncbeta} for $\lambda=4 s \E_{\pi}(\beta)$, we have
\begin{align*}
\E_{\pi}\l(\econst^{4 s \E_{\pi}(\beta) (\beta-\E_{\pi}(\beta))} \r)\le \econst^{8 s^2 C  (\E_{\pi}(\beta))^2 L_{\beta}^2}.
\end{align*}
The first claim of the lemma now follows by rearrangement. Additionally we have $(\E_{\pi}(\beta))^2\le \E_{\pi}(\beta^2)$, which can be further bounded as
\begin{align}
    \nonumber\E_{\pi}(\beta^2)&=\E_{\pi}(\sum_{i=1}^{b} \|\grad U_{i}(\B{A}_{i}\btheta)\|^2)
    \intertext{ using Assumptions \asstwo \, and \assseven}
    \nonumber&\le \E_{\pi}(\sum_{i=1}^{b} M_i^2\|\B{A}_{i}(\btheta-\bthetastar)\|^2)\\
    \nonumber&\le (\max_{1\le i\le b}M_i)^2 \|\B{A}^\top \B{A}\| \E_{\pi}(\|\btheta-\bthetastar\|^2)\\
    \nonumber &\le (\max_{1\le i\le b}M_i)^2 \|\B{A}^\top \B{A}\| d m_U^{-1}=d\sigma^2_U,
\end{align}
where we have used the fact that $\E_{\pi}(\|\btheta-\bthetastar\|^2)\le \frac{d}{m_U}$ by Proposition 1 part (ii) of \cite{durmus2018high}.
\end{proof}

\begin{lemma}\label{lem:ratioofnormalizingconstants}
If we assume that $b=1$, $d_1=d$, and $\B{A}_1$ is of full rank, then we have $Z_{\pi}=Z_{\pi_{\rho}}$ for any $\rho$. 
More generally, assume that \asszero, \asstwo, \assfive and \assseven\, hold. Then for $\rho^2\le \frac{1}{6 \sigma^2_U}$,
\begin{equation}\label{eq:Zratiolowerbnd}
    \log\l(\frac{Z_{\pi}}{Z_{\pi_{\rho}}}\r)\ge 
-\E_{\pi}(\overline{B}(\btheta))-\rho^4 (2+d)\sigma^4_U.
\end{equation}

\end{lemma}
\begin{proof}
Firstly, we that $b=1$, $d_1=d$, and $\B{A}_1$ is of full rank. Then one can show that
\begin{align*}
&\int_{\btheta\in \R^d}\exp(-U_1^{\rho}(\B{A}_1\btheta))\mathrm{d}\btheta \\
&=
\int_{\btheta\in \R^d}\int_{\B{z}_1\in \R^d}\exp\l(-U_1(\B{z}_1)-\frac{\|\B{z}_1-\B{A}_1\btheta\|^2}{2\rho^2}\r)\cdot \frac{\mathrm{d} \B{z}_1}{(2\pi \rho^2)^{d/2}}\mathrm{d} \btheta\\
&=\int_{\btheta\in \R^d}\exp(-U_1(\B{A}_1\btheta))\mathrm{d}\btheta.
\end{align*}
Hence, in this case $Z_{\pi}=Z_{\pi_{\rho}}$.

Now we look at the general multiple splitting case. Note that using the fact that $\int_{\btheta \in \mathbb{R}^d} \pi(\btheta)\l( 1 - \frac{\marginal(\btheta)}{\pi(\btheta)}\r)\mathrm{d}\btheta=0$, it follows that
\begin{equation}\label{eq:ZpiZpirho}
    \frac{Z_{\pi}}{Z_{\pi_{\rho}}}=
    \l( \int_{\btheta \in \mathbb{R}^d} \pi(\btheta) \exp\l(\sum_{i=1}^{b} 
    \l(U_i(\B{A}_i\btheta)- U_i^{\rho}(\B{A}_i\btheta)\r)\r)  \mathrm{d}\btheta \r)^{-1}.
\end{equation}
By \eqref{eq:proof_theorem_2_2} and \eqref{eq:ZpiZpirho}, we have
\begin{align}
    \frac{Z_{\pi}}{Z_{\pi_{\rho}}}\nonumber
    &\ge 
    \l( \int_{\btheta \in \mathbb{R}^d} 
    \pi(\btheta)\exp(\overline{B}(\btheta))
    \mathrm{d}\btheta \r)^{-1}.
\end{align}

Note that  $\overline{B}(\btheta)= \frac{\rho^2}{2} \beta(\btheta)^2$, hence by Lemma \ref{lem:momgenbeta2}, we have that for $\rho^2\le \frac{1}{6 \sigma^2_U}$,
\[\int_{\btheta \in \mathbb{R}^d} 
    \pi(\btheta)\exp(\overline{B}(\btheta))
    \mathrm{d}\btheta
    \le \econst^{\E_{\pi}(\overline{B}(\btheta))}\cdot \econst^{2\rho^4 \sigma^4_U+\rho^4 \E_{\pi}(\btheta)^2 \sigma^2_U}\le \econst^{\E_{\pi}(\overline{B}(\btheta))}\cdot \econst^{(2+d)\rho^4 \sigma^4_U}.\]
The claim now follows by rearrangement.
\end{proof}

Now we are ready to prove our bias bound.
\begin{proof}[Proof of Proposition \ref{proposition:2}]
In the single splitting case, assuming \asszero, $b=1$, $d_1=d$, and that $\B{A}_1$ is invertible, by Lemma \ref{lem:ratioofnormalizingconstants}, we have $\frac{Z_{\pi}}{Z_{\pi_{\rho}}}=1$. Combining this and \eqref{eq:proof_theorem_2_4} with our bound \eqref{eq:TVpipirhoZpiZhroUUrho2}, we obtain that
\begin{align*}
    I(U,U^{\rho})
    &\le \int_{\btheta \in \mathbb{R}^d} \pi(\btheta)\l(
    \underline{B}(\btheta) \r)_{-}\mathrm{d}\btheta \\
    &= \int_{\btheta \in \mathbb{R}^d} \pi(\btheta)\l(
    \dfrac{\rho^2\nr{\nabla U_1(\B{A}_1\btheta)}^2}{2(1+\rho^2M_1)} - \frac{d}{2}\log(1+\rho^2M_1) \r)_{-}\mathrm{d}\btheta \\
    &\le \int_{\btheta \in \mathbb{R}^d} \pi(\btheta)\l(- \frac{d}{2}\log(1+\rho^2M_1) \r)_{-}\mathrm{d}\btheta \\
    &\le  \frac{d}{2} M_1\rho^2.
\end{align*}
In the general case, note that the $m_U$-strong convexity follows by \eqref{eq:Uhesslwbnd}. We have the lower bound \eqref{eq:Zratiolowerbnd} on $\log\l(\frac{Z_{\pi}}{Z_{\pi_{\rho}}}\r)$. By combining this and \eqref{eq:proof_theorem_2_4} with our bound \eqref{eq:TVpipirhoZpiZhroUUrho2}, we obtain that for $\rho^2\le \frac{1}{6 \sigma^2_U}$, we have
\begin{align*}
    %&\nr{\marginal-\pi}_{\mathrm{TV}}\nonumber 
    &I(U,U^{\rho})
    \\
    \nonumber
    &\le \sum_{i=1}^b \frac{d_i}{2}\log(1+\rho^2 M_i)+\sum_{i=1}^b \frac{\rho^4 M_i }{2(1+\rho^2 M_i)} \E_{\pi}(\|\grad U_i(A_i\btheta)\|^2)+(2+d)\rho^4 \sigma^4_U\\
    &\le \frac{\rho^2}{2}(\sum_{i=1}^b d_i M_i)+\frac{\rho^4}{2}(\max_{1\le i\le b}M_i)  \E_{\pi}(\beta^2)+(2+d)\rho^4 \sigma^4_U\\
    &\le \frac{\rho^2}{2}(\sum_{i=1}^b d_i M_i) +\l(2+\frac{3}{2}d\r)\rho^4 \sigma^4_U,
\end{align*}
where in the last step we have used the facts that $\E_{\pi}(\beta^2)\le \sigma^2_U d$ (by Lemma \ref{lem:momgenbeta2}) and that $\max_{1\le i\le b}M_i\le \sigma^2_U$ (by the definition of $\sigma^2_U$).
\end{proof}

\subsection{Non-Asymptotic Bounds for the 1-Wasserstein Distance}\label{appendix:subsec:bounds_bias_Wasserstein}

The result shown in Proposition \ref{thm:Wassersteinpipirhosinglesplitting} shows a Wasserstein error rate bound in the single splitting case ($b=1$). 
Its proof is given below.
\begin{proof}[Proof of Proposition \ref{thm:Wassersteinpipirhosinglesplitting}]
The integrability of $U_1$ follows from assumption \eqref{eq:fradunbnd}. Assume without loss of generality that $U_1(\btheta)$ is normalised, i.e.  $\int_{\btheta\in \R^d}\exp(-U_1(\btheta))\mathrm{d}\btheta=1$ (if it is not, we can add to the potential an appropriate constant). Then the distribution 
\[\pi_{\rho}(\btheta)=  \frac{1}{(2\pi\rho^2)^{d/2}}\int_{\bz\in \R^d}\exp\l(-U_1(\bz)-\frac{\|\btheta-\bz\|^2}{2\rho^2}\r)\mathrm{d} \bz\]
is the convolution of $\pi(\btheta) = \exp(-U_1(\btheta))$ and a $d$-dimensional Gaussian random variable with mean zero and covariance $\rho^2 \B{I}_d$. In particular, it is clear that 
\begin{align*}\int_{\btheta\in \R^d} \pi_{\rho}(\btheta)\mathrm{d} \btheta&=\int_{\bz\in \R^d}\frac{1}{(2\pi\rho^2)^{d/2}}\int_{\btheta\in \R^d}\exp\l(-U_1(\bz)-\frac{\|\btheta-\bz\|^2}{2\rho^2}\r)\mathrm{d} \btheta \mathrm{d} \bz
\\
&=\int_{\bz\in \R^d}\exp(-U_1(\bz))=1.
\end{align*}
The first part of the bound follows from the fact that the expectation of the norm of this Gaussian random variable is bounded by $\rho \sqrt{d}$ (since the expectation of the square of the norm is $\rho^2 d$, this follows by Jensen's inequality).

In order to obtain the second part, we are going to use the dual formulation of the 1-Wasserstein distance (see e.g. Remark 6.5 of \cite{Villani2008}),
\begin{align}
    \nonumber W_1(\pi,\pi_{\rho})&=\sup_{g: \|g\|_{\mathrm{Lip}}\le 1} \int_{\btheta} g(\btheta) \pi(\btheta) d\btheta-\int_{\btheta} g(\btheta)\pi_{\rho}(\btheta) \mathrm{d}\btheta\\
    \label{eq:dualform}&=\sup_{g\in C^1(\R^d): \|\grad g\|_{\infty}\le 1}\int_{\btheta} g(\btheta) \pi(\btheta) d\btheta-\int_{\btheta} g(\btheta)\pi_{\rho}(\btheta) \mathrm{d}\btheta,
\end{align}
where the second equality follows from the fact that differentiable functions $g$ with $\|\grad g\|_{\infty}\le 1$ are dense among 1-Lipschitz functions on $\R^d$.

The evolution of a density $\pi_{\rho}$ as we increase the variance $\rho^2$ is known to follow the heat equation, see Section 2.4 of \cite{lawler2010random},
\[
\frac{\mathrm{d}}{\mathrm{d}(\rho^2)}\pi_{\rho}(\btheta)=
\frac{1}{2}\Laplace{\pi_{\rho}}(\btheta),
\]
where $\Laplace \pi_{\rho}(\btheta)=\sum_{i=1}^{d}\frac{\partial^2}{\partial \theta_i^2} \pi_{\rho}(\btheta)$ denotes the Laplacian of $\pi_{\rho}$. 
By integration, we obtain that
\begin{align*}
\sup_{g\in C^1(\R^d): \|\grad g\|_{\infty}<1}\frac{\mathrm{d}}{\mathrm{d}(\rho^2)} \int_{\btheta} g(\btheta)\pi_{\rho}(\btheta) \mathrm{d}\btheta&=
\sup_{g\in C^1(\R^d): \|\grad g\|_{\infty}\le 1}\frac{1}{2} \int_{\btheta\in \R^d} g(\btheta) \Laplace{\pi_{\rho}}(\btheta) \mathrm{d}\btheta.
\end{align*}
Now if we define the functional 
\[\F(\mu):=\sup_{g\in C^1(\R^d): \|\grad g\|_{\infty}\le 1}\frac{1}{2} \int_{\btheta\in \R^d} g(\btheta) \Laplace\mu(\btheta) \mathrm{d}\btheta.\]
Then it is easy to see that this is convex ($\F(\alpha \mu+(1-\alpha)\nu)\le \alpha \F(\mu)+(1-\alpha)\F(\nu)$ for $\alpha\in[0,1]$) and shift-invariant (if $\nu(\B{x})=\mu(\B{x}-\B{a})$ some constant $\B{a}\in \R^d$, then $\F(\nu)=\F(\mu)$). Therefore it follows by the argument on pages 1-2 of \cite{bennett2015generating} (monotonicity property of the heat semigroup for convex functionals) that
$\F(\pi_{\rho})\le \F(\pi)$ for every $\rho\ge 0$. 

Initially, we have
\begin{align*}
\F(\pi)&=\sup_{g\in C^1(\R^d): \|\grad g\|_{\infty}\le 1}\frac{1}{2} \sum_{i=1}^{d}\int_{\btheta\in \R^d} g(\btheta) \frac{\partial^2}{\partial \theta_i^2}\pi(\btheta) \mathrm{d}\btheta.
\end{align*}
After separating $\btheta$ to $\theta_i \in \mathbb{R}$ and $\btheta_{-i} \in \mathbb{R}^{d-1}$ (denoting the rest of the coordinates), we have
\begin{align*}
    &\int_{\btheta\in \R^d} g(\btheta) \frac{\partial^2}{\partial \theta_i^2}\pi(\btheta) \mathrm{d}\btheta\\
    &=
    \int_{\btheta_{-i}\in \R^{d-1}} 
    \l[\int_{\theta_i\in \R}
    g(\btheta) \frac{\partial^2}{\partial \theta_i^2}\pi(\btheta) \mathrm{d}\theta_i\r]\mathrm{d}\btheta_{-i},
\end{align*}
and now integration by parts and using condition \eqref{eq:fradunbnd} and the Lipschitz continuity of $g$ leads to
\begin{align*}&\int_{\theta_i\in \R}
    g(\btheta) \frac{\partial^2}{\partial \theta_i^2}\pi(\btheta) \mathrm{d}\theta_i\\
    &=\l[ -g(\btheta) \frac{\partial}{\partial \theta_i}U_1(\btheta)\cdot \exp(-U_1(\btheta)) \r]_{\theta_i=-\infty}^{\theta_i=\infty}
    \\
    &+\int_{\theta_i\in \R}
    \frac{\partial }{\partial \theta_i}g(\btheta) \frac{\partial}{\partial \theta_i}U_1(\btheta) \exp(-U_1(\btheta)) \mathrm{d}\theta_i\\
    &=\int_{\theta_i\in \R}
    \frac{\partial }{\partial \theta_i}g(\btheta) \frac{\partial}{\partial \theta_i}U_1(\btheta) \pi(\btheta) \mathrm{d}\theta_i.
\end{align*}
By summing up in $i$, we obtain that
\begin{align*}
\F(\pi)&\le \frac{1}{2}\sup_{g\in C^1(\R^d): \|\grad g\|_{\infty}\le 1}\sum_{i=1}^{d}\int_{\btheta\in \R^d}  \frac{\partial}{\partial \theta_i}U_1(\btheta) \frac{\partial}{\partial \theta_i}g(\btheta) \pi(\btheta)\mathrm{d}\btheta\\
&\le \frac{1}{2} \int_{\btheta\in \R^d} \|\grad U_1(\btheta)\| \pi(\btheta) \mathrm{d}\btheta.
\end{align*}
Using the monotonicity property of $F(\pi_{\rho})$, now the second bound of the theorem follows based on formula
\eqref{eq:dualform}. The finiteness of this integral follows from assumption \ref{eq:fradunbnd}.

Now we are going to consider the $m$-strongly convex and $M$-smooth $U_1$ case. In such situations, it is straightforward to see that condition \eqref{eq:fradunbnd} holds with $a_1 = m\nr{\bthetastar}^2/2$, $a_2=m/2$, $a_3=0$, $a_4 = M$, $\alpha=2$ and $\beta=1$; where $\bthetastar$ is the minimum of $U_1$. For the integral of the norm of the gradient, we have by Jensen's inequality
\[\int_{\btheta\in \R^d} \|\grad U_1(\btheta)\| \pi(\btheta) \mathrm{d}\btheta\le \l(\int_{\btheta\in \R^d} \|\grad U_1(\btheta)\|^2 \pi(\btheta) \mathrm{d}\btheta\r)^{1/2}.\]
For some index $1\le i\le d$, we have
\begin{align*}
\int_{\btheta\in \R^d} \l(\frac{\partial}{\partial \theta_i} U_1(\btheta)\r)^2 \pi(\btheta)  \mathrm{d}\btheta =
\int_{\btheta_{-i}\in \R^{d-1}} \l[\int_{\theta_i\in \R} \l(\frac{\partial}{\partial \theta_i} U_1(\btheta)\r)^2 \exp\l(-U_1(\btheta)\r) \mathrm{d}\theta_i \r]\mathrm{d}\btheta_{-i},
\end{align*}
and using integration by parts, and the conditions of strong convexity and smoothness, we have
\begin{align*}
    &\int_{\theta_i\in \R} \l(\frac{\partial}{\partial \theta_i} U_1(\btheta)\r)^2 \exp\l(-U_1(\btheta)\r) \mathrm{d}\theta_i\\
    &=\l[-\exp\l(-U_1(\btheta)\r) \frac{\partial}{\partial \theta_i} U_1(\btheta)\r]_{\theta_i=-\infty}^{\theta_i=\infty}
    +\int_{\theta_i\in \R}\exp\l(-U_1(\btheta)\r)\frac{\partial^2}{\partial \theta_i^2} U_1(\btheta) \mathrm{d}\theta_i\\
    &\le \int_{\theta_i\in \R}\exp\l(-U_1(\btheta)\r)M_1 \mathrm{d}\theta_i.
\end{align*}
By integrating this expression w.r.t. $\btheta_{-i}$ and summing up in $i$, we obtain that \[\int_{\btheta\in \R^d} \|\grad U_1(\btheta)\|^2 \pi(\btheta) d\btheta\le M_1d,\]
so the last claim of the theorem follows.
\end{proof}

\section{Proofs of the Results of Section \ref{sec:main_result_convergence_rates}}

This section aims at proving the results claimed in Section \ref{sec:main_result_convergence_rates}.

\subsection{Proof of Theorem \ref{thm:RicciSGS} }
\label{appendix_proof_theorem_K_SGS}

The following two propositions are going to be used for the proof of Theorem \ref{thm:RicciSGS}. The first one will allow us to bound the Wasserstein distance of two log-concave distributions based on the differences between their gradients. This is achieved by coupling processes evolving according to the Langevin dynamics with common Brownian noise.
\begin{proposition}\label{prop:graddiffWass}
Let $\mu$ and $\mu'$ be two distributions on $\R^d$ that are absolutely continuous with respect to the Lebesgue measure, and whose negative log-likelihoods are continuously differentiable, strongly convex and smooth (gradient-Lipschitz).
Denote the strong convexity constants $m(\mu), m(\mu')$ and smoothness constants $M(\mu)$ and $M(\mu')$. Then the Wasserstein distance of order $1\le p\le \infty$ of these two distributions can be upper bounded as
\begin{equation*}
    W_p(\mu,\mu')\le \frac{\|D_{\mu,\mu'}\|_{L^{p}(\mu)}}{m(\mu')}\quad \text { for } \quad D_{\mu,\mu'}(\B{z})= \grad \log \mu(\B{z})-\grad \log \mu'(\B{z}).
\end{equation*}
\end{proposition}
\begin{proof}
Let $\mu(\bz) = \exp(-U(\bz))$ and $\mu'(\bz) = \exp(-U'(\bz))$.

First, we are going to consider the case $1\le p<\infty$. Note that it is easy to show that under the strong convexity and smoothness assumptions of this proposition, the Wasserstein distance of order $p$ between $\mu$ and $\mu'$ is finite for such $p$. Assume that $(\B{X}_1(0),\B{X}_3(0))$ is an optimal coupling in Wasserstein distance of order $p$ between $\mu$ and $\mu'$, so that $\B{X}_1(0)\sim \mu$, $\B{X}_3(0)\sim \mu'$, and 
\[\l[\E\l(\|\B{X}_1(0)-\B{X}_3(0)\|^p\r)\r]^{1/p}=W_p(\mu,\mu').\]
The existence of such a coupling follows from Theorem 4.1 of \cite{Villani2008}. Let $\B{X}_2(0)=\B{X}_1(0)$.
We now define three Langevin diffusions $(\B{X}_1(t),\B{X}_2(t),\B{X}_3(t))_{t\geq0}$ with a common noise (synchronous coupling)
\begin{align*}
    \mathrm{d}\B{X}_1(t)&=-\grad U(\B{X}_1(t))\mathrm{d}t+\sqrt{2}\mathrm{d}\B{B}_t,\\
    \mathrm{d}\B{X}_2(t)&=-\grad U'(\B{X}_2(t))\mathrm{d}t+\sqrt{2}\mathrm{d}\B{B}_t,\\
    \mathrm{d}\B{X}_3(t)&=-\grad U'(\B{X}_3(t))\mathrm{d}t + \sqrt{2}\mathrm{d}\B{B}_t.
\end{align*}
Under the strong convexity and smoothness assumptions on the log-densities, these SDEs admit unique strong solutions (see Theorem 3.1 of \cite{pavliotis2014stochastic} and \cite{ArnoldSDE}). Since $\B{X_1}(0)\sim \mu$ and $X_3\sim \mu'$, we can see that $\B{X}_1(t)\sim \mu$ and $\B{X}_3(t)\sim \mu'$ for every $t \geq 0$. $\B{X}_{2}(t)$ is initialized at $\mu$ since $\B{X}_2(0)=\B{X}_1(0)$ and converges towards $\mu'$. The proof of this proposition is based on a coupling argument based on these three diffusions. Let
\begin{equation*}
    D_{12}(t)=\B{X}_1(t)-\B{X}_2(t)-t(\grad U'(\B{X}_1(0))-\grad U(\B{X}_1(0))).
\end{equation*}
Then we can decompose $\B{X}_1(t)-\B{X}_3(t)$ as
\begin{equation}\label{eq:Z1Z3decomposition}
    \B{X}_1(t)-\B{X}_3(t)=
    t(\grad U'(\B{X}_1(0))-\grad U(\B{X}_1(0)))+D_{12}(t)+(\B{X}_2(t)-\B{X}_3(t)).
\end{equation}
In the next two paragraphs of the proof, we are going to establish the following auxiliary inequalities
\begin{align}
&\label{eq:Z23bnd}\|\B{X}_2(t)-\B{X}_3(t)\|\le \exp(-m(\mu') t)\cdot \|\B{X}_1(0)-\B{X}_3(0)\|,\\
&\label{eq:D12bnd}\|D_{12}(t)\|\le C_0 t^3+C_1 t^{2}(\|\grad U(\B{X}_1(0))\|+\|\grad U'(\B{X}_1(0))\|) +C_2 t \sup_{0\le s\le t}\|\B{B}_s\| \text{ for }0\le t\le C_3, 
\end{align}
for positive constants $C_0$, $C_1$, $C_2$, $C_3$ that only depend on the dimension $d$ and the convexity parameters $m(\mu),m(\mu'),M(\mu),M(\mu')$.
Let $\|X\|_{L^p}=(\E(\|X\|^p))^{1/p}$ denote the $L^p$ norm of a random variable. By taking the $L^p$ norms of both sides of \eqref{eq:Z1Z3decomposition}, and using Minkowski's inequality, we can see that 
\begin{equation*}
\|\B{X}_1(t)-\B{X}_3(t)\|_{L^p}\le 
    t\|\grad U'(\B{X}_2(0))-\grad U(\B{X}_1(0))\|_{L^p}+\|D_{12}(t)\|_{L^p}+\|\B{X}_2(t)-\B{X}_3(t)\|_{L^p}.
\end{equation*}
By the definition of the Wasserstein distance, we know that
$W_{p}(\mu,\mu')\le \|\B{X}_1(t)-\B{X}_3(t)\|_{L^p}$, and by assuming inequalities \eqref{eq:Z23bnd} and \eqref{eq:D12bnd} are true, we obtain that for $0\le t\le C_3$, 
\begin{align}\label{eq:WpWpexpbnd}
W_{p}(\mu,\mu')&\le 
    t\|\grad U'(\B{X}_1(0))-\grad U(\B{X}_1(0))\|_{L^p}+
W_{p}(\mu,\mu')\exp(-m(\mu') t)\\
\nonumber
&+C_0 t^3+ C_1 t^2(\|\grad U(\B{X}_1(0))\|_{L^p}+\|\grad U'(\B{X}_2(0))\|_{L^p})+C_2 t \l\|\sup_{0\le s\le t}\|\B{B}_s\| \r\|_{L^p}.
\end{align}
It is easy to show that under the strong convexity and smoothness assumptions of this proposition, the terms $\|\grad U(\B{X}_1(0))\|_{L^p}$ and $\|\grad U'(\B{X}_1(0))\|_{L^p}$ are finite. By the reflection principle for the Brownian motion (see \cite{levy1940certains}), in one dimension, the distribution of $\sup_{0\le s\le t}B_s$ is the same as the distribution of $|B_t|$. Using the triangle inequality, and the fact that $\|Y\|_{L^p}\le \sqrt{p}$ for a standard Gaussian random variable $Y$, it follows that 
\[\l\|\sup_{0\le s\le t}\|\B{B}_s\| \r\|_{L^p}\le 2d\sqrt{t}\sqrt{p}.\]
Hence all of the terms bounding 
$\|D_{12}(t)\|_{L^p}$ in  \eqref{eq:WpWpexpbnd} are of order $o(t)$, and the claim of the proposition follows by rearrangement and letting $t\searrow 0$.

Now we are going to prove the two auxiliary inequalities. We start with \eqref{eq:Z23bnd}. 
From It\^o's formula (see Lemma 3.2 of \cite{pavliotis2014stochastic}), 
$\|\B{X}_2(t)-\B{X}_3(t)\|^2$ is differentiable in $t$ and satisfies
\begin{align*}
    \frac{\mathrm{d}}{\mathrm{d}t}\|\B{X}_2(t)-\B{X}_3(t)\|^2
    &=-2\inner{\B{X}_2(t)-\B{X}_3(t)}{\grad U'(\B{X}_2(t)-\grad U'(\B{X}_3(t))} \\
    &\le -2m(\mu')\|\B{X}_2(t)-\B{X}_3(t)\|^2,
\end{align*}
where the last step follows from the strong convexity of $U'$. We obtain \eqref{eq:Z23bnd} by Gr\"{o}nwall's inequality and rearrangement.

We continue with the proof of \eqref{eq:D12bnd}. By It\^o's formula, we can see that
\begin{align*}
D_{12}(t)&=\B{X}_1(t)-\B{X}_2(t)-t(\grad U'(\B{X}_1(0))-\grad U(\B{X}_1(0)))\\
&=\int_{s=0}^{t}(\grad U'(\B{X}_2(t))-\grad U(\B{X}_1(t)))\mathrm{d}s -t(\grad U'(\B{X}_1(0))-\grad U(\B{X}_1(0)))\\
&=\int_{s=0}^{t}[\grad U'(\B{X}_2(t))-\grad U'(\B{X}_1(0))]\mathrm{d}s  +\int_{s=0}^t[\grad U(\B{X}_1(0))-\grad U(\B{X}_1(t))]\mathrm{d}s.
\end{align*}
Using the smoothness assumption for $U$ and $U'$, and the fact that $\B{X}_1(0)=\B{X}_2(0)$, we have
\begin{equation}\label{eq:D12auxbnd}
    \|D_{12}(t)\|\le M(\mu')\int_{s=0}^{t}\|\B{X}_2(t)-\B{X}_2(0)\|\mathrm{d}s+
    M(\mu)\int_{s=0}^t\|\B{X}_1(t)-\B{X}_1(0)\|\mathrm{d}s.
\end{equation}
Let
\begin{align*}
    \B{Y}_1'(t)&=-\grad U(\B{Y}_1(t)),\\
    \B{Y}_2'(t)&=-\grad U'(\B{Y}_2(t)),
\end{align*}
and assume that $\B{Y}_1(0)=\B{Y}_2(0)=\B{X}_1(0)=\B{X}_2(0)$.
Then these ODEs have a unique solution (see page 74 of \cite{Perko}). Now by the triangle inequality, and the fact that $\B{Y}_1(0)=\B{X}_1(0)$, we have
\[\|\B{X}_1(s)-\B{X}_1(0)\|\le \|\B{Y}_1(s)-\B{Y}_1(0)\|+\|\B{Y}_1(s)-\B{X}_1(s)\|.\]
For the first part, by Taylor's expansion, and the smoothness assumption on $U$, we have
\[\|\B{Y}_1(t)-\B{Y}_1(0)\|\le s\|\grad U(\B{Y}_1(0))\|+\frac{1}{2}M(\mu) s^2.\]
For the second part, by It\^o's formula, we have
\begin{align*}
\B{Y}_1(s)-\B{X}_1(s)&=\int_{r=0}^{s}[\grad U(\B{X}_1(r))-\grad U(\B{Y}_1(r))] \mathrm{d}r + \sqrt{2}\B{B}_s,\\
\|\B{Y}_1(s)-\B{X}_1(s)\|&\le M(\mu)\int_{r=0}^{s}\|\B{X}_1(r)-\B{Y}_1(r)\| \mathrm{d}r+\sqrt{2}\|\B{B}_s\|,\\
\sup_{0\le r\le s}\|\B{Y}_1(s)-\B{X}_1(s)\|&\le M(\mu) s\sup_{0\le r\le s}\|\B{Y}_1(s)-\B{X}_1(s)\|+\sqrt{2}\sup_{0\le r\le s}\|\B{B}_r\|.
\end{align*}
Hence for $s\le 1/(2M(\mu))$, we have
$\sup_{0\le r\le s}\|\B{Y}_1(s)-\B{X}_1(s)\|\le 2\sqrt{2} \sup_{0\le r\le s}\|\B{B}_r\|$. By combining the above two bounds, for $0
\le s\le 1/(2M(\mu))$, we have
\[\|\B{X}_1(s)-\B{X}_1(0)\| \le s\|\grad U(\B{Y}_1(0))\|+\frac{1}{2}M(\mu) s^2+2\sqrt{2} \sup_{0\le r\le s}\|\B{B}_r\|,\]
and by the same argument, for $0
\le s\le 1/(2M(\mu'))$,
\[\|\B{X}_2(s)-\B{X}_2(0)\| \le s\|\grad U'(W_2(0))\|+\frac{1}{2}M(\mu') s^2+2\sqrt{2} \sup_{0\le r\le s}\|\B{B}_r\|.\]
Inequality now \eqref{eq:D12bnd} follows by substituting these into \eqref{eq:D12auxbnd} and doing some rearrangement.

Finally, the result for $p=\infty$ follows from a limiting argument. By Proposition 3 of \cite{Wassersteininfinity}, we have
\begin{align*}
W_{\infty}(\mu,\mu')=\lim_{p\to \infty}W_p(\mu,\mu')\le \sup_{1\le p<\infty} \frac{\|D_{\mu,\mu'}\|_{L^{p}(\mu)}}{m(\mu')} \le 
\frac{\|D_{\mu,\mu'}\|_{L^{\infty}(\mu)}}{m(\mu')}.
\end{align*}
\end{proof}

\begin{proposition}\label{prop:contractionSGS}
Let $\btheta, \btheta'\in \R^d$ be two parameter values, and $\mu_i$, resp. $\mu'_i$, denotes the conditional distributions of $\B{z}_i$ given $\btheta$ under $\Pi_{\rho}$, resp. $\btheta'$. Then under Assumption  \assfive, for every $1\le p\le \infty$, we have
\begin{equation}\label{eq:prop:contractionSGS}
W_{p}(\mu_i,\mu'_i)\le \frac{1}{1+\rho^2 m_i}\|\B{A}_i(\btheta-\btheta')\|.
\end{equation}
\end{proposition}
\begin{proof}
    We have 
    $\mu_i(\bz)\propto \exp\left(-U_i(\bz)-\frac{\nr{\B{A}_i\btheta-\bz}^2}{2\rho^2}\right)$ and $\mu'_i(\bz)\propto \exp\left(-U_i(\bz)-\frac{\nr{\B{A}_i\btheta'-\bz}^2}{2\rho^2}\right)$.
    
    Proposition  \ref{prop:graddiffWass} requires the smoothness (gradient Lipschitz) property, so it cannot be applied directly to these potentials under our assumptions. To overcome this difficulty, we are going to use the Moreau-Yosida envelope of $U_i$ \citep{Durmus2018} defined for any  regularisation parameter $\lambda>0$ as
    \begin{equation*}
    U_i^{\lambda}(\bz):=\min_{\B{y}\in \R^{d}} \l\{U_i(\B{y})+(2\lambda)^{-1} \|\B{y}-\bz\|^2\r\}.
    \end{equation*}
    By Theorem 1.25 of \cite{RockafellarVariational}, $U_i^{\lambda}$ converges pointwise to $U_i$, that is for any $\bz\in \R^d$, 
    \begin{equation}\label{eq:MYpointwise}
    \lim_{\lambda\to 0}U_i^{\lambda}(\bz)=U_i(\bz).
    \end{equation}
    Moreover, from Proposition 12.19 of \cite{RockafellarVariational} and Theorem 2.2 of \cite{lemarechal1997practical} it follows that $U_i^{\lambda}$ is $\lambda^{-1}$ gradient Lipschitz and $\frac{m_i}{1+\lambda m_i}$-strongly convex.
    
    Let $\mu_i^{\lambda}(\bz)\propto \exp\left(-U_i^{\lambda}(\bz)-\frac{\nr{\B{A}_i\btheta-\bz}^2}{2\rho^2}\right)$ and ${\mu'_i}^{\lambda}(\bz)\propto \exp\left(-U_i^{\lambda}(\bz)-\frac{\nr{\B{A}_i\btheta'-\bz}^2}{2\rho^2}\right)$, then we have
    
    \[\|\grad \log (\mu_i^{\lambda}(\bz))-\grad \log ({\mu_i'}^{\lambda}(\bz))\|=\frac{\|\B{A}_i\btheta-\B{A}_i\btheta'\|}{\rho^2}.\] 
    Since $-\log \mu_i^{\lambda}(\bz)$ and $-\log\mu_i'^{\lambda}(\bz)$ are $\frac{m_i}{1+\lambda m_i}+\frac{1}{\rho^2}$-strongly convex and $\frac{1}{\lambda}+\frac{1}{\rho^2}$-smooth, it follows from Proposition \ref{prop:graddiffWass} that we have for every $1
    \le p\le \infty$
    \begin{equation}\label{eq:Wpmuilambda}
    W_{p}(\mu_i^{\lambda},{\mu'_i}^{\lambda})\le \frac{\|\B{A}_i\btheta-\B{A}_i\btheta'\|}{1+\rho^2 m_i/(1+m_i\lambda)}.
    \end{equation}
    Now we are going to consider the case $1\le p<
    \infty$ first. To complete the proof, we still need to bound $W_p(\mu_i^{\lambda}, \mu_i)$.
    By Theorem 6.15 of \cite{Villani2008}, we have
    \begin{equation}\label{eq:W1Villanibnd}W_p(\mu_i^{\lambda}, \mu_i)\le \l[\int_{\bz\in \R^d} \|\bz-\btheta\|^p |\mu_i(\bz)-\mu_i^{\lambda}(\bz)|\mrd \bz\r]^{1/p}.
    \end{equation}
    Note that $|\mu_i(\bz)-\mu_i^{\lambda}(\bz)|\le \mu_i(\bz)+\mu_i^{\lambda}(\bz)$.  Moreover, from the definition of the Moreau-Yosida envelope $U_i^{\lambda}$, it follows that $U_i^{\lambda}(\bz)\leq U_i^{\lambda'}(\bz)$ for $\lambda'<\lambda$, hence it is monotone increasing towards $U_i(\bz)$ as $\lambda\to 0$. This implies that the normalising constant
    \[Z_i^{\lambda}=\int_{\bz} \exp\l(-U_i^{\lambda}(\bz)-\frac{\|\bz-\B{A}_i\btheta\|^2}{2\rho^2}\r)\mrd \bz\]
    is monotone decreasing towards $Z_i=\int_{\bz} \exp\l(-U_i(\bz)-\frac{\|\bz-\B{A}_i\btheta\|^2}{2\rho^2}\r)\mrd \bz$ as $\lambda\to 0$ by the monotone convergence theorem.
    Therefore we have for any fixed $\Lambda>0$ and $0< \lambda<\Lambda$ 
    \[\mu_i^{\lambda}(\bz)=\frac{\exp\l(-U_i^{\lambda}(\bz)-\frac{\|\bz-\B{A}_i\btheta\|^2}{2\rho^2}\r)}{Z_i^{\lambda}}
    \le \frac{\exp\l(-U_i^{\Lambda}(\bz)-\frac{\|\bz-\B{A}_i\btheta\|^2}{2\rho^2}\r) }{Z_i}.\]
    This means that for $\lambda<\Lambda$, we have
    \[ \|\bz-\btheta\|^p |\mu_i(\bz)-\mu_i^{\lambda}(\bz)|\le \|\bz-\btheta\|^p \l(\mu_i(\bz)+\frac{\exp\l(-U_i^{\Lambda}(\bz)-\frac{\|\bz-\B{A}_i\btheta\|^2}{2\rho^2}\r)}{Z_i}\r).\]

    Using the strong-convexity of $-\log\mu_i$, it follows that it has a unique minimizer which we denote by $\bz_i^*$. In particular, we have  
    \begin{align*}
        \int_{\bz\in \R^d} \|\bz-\btheta\|^p \mu_i(\bz)\mrd \bz\le \mu_i(\bz_i^*) \int_{z\in \R^{d}}\|\bz-\btheta\|^p \exp\l(-(m_i+1/\rho^2)\|\bz-\bz_i^*\|^2/2\r)\mrd \bz<\infty,
    \end{align*}
    and with the same argument we can also show that 
    \[ \int_{\bz\in \R^d}\|\bz-\btheta\|^p \frac{\exp\l(-U_i^{\Lambda}(\bz)-\frac{\|\bz-\B{A}_i\btheta\|^2}{2\rho^2}\r)}{Z_i}<\infty.\]
    Hence using the pointwise convergence \eqref{eq:MYpointwise} it follows from the dominated convergence theorem and the bound \eqref{eq:W1Villanibnd} that $W_p(\mu_i^{\lambda}, \mu_i)\to 0$ as $\lambda\to 0$. The same also holds for $W_p({\mu_i'}^{\lambda}, \mu_i')$, so we can conclude using \eqref{eq:Wpmuilambda} and the triangle inequality
    \[W_p(\mu_i, \mu_i')\le W_p(\mu_i,\mu_i^{\lambda}) + W_p({\mu_i}^{\lambda},{\mu_i'}^{\lambda})+W_p({\mu_i'}^{\lambda}, \mu_i').\]
    Finally, since we have shown the inequality \eqref{eq:prop:contractionSGS} for $1\le p<\infty$, the bound for $p=\infty$ follows by Proposition 3 of \cite{Wassersteininfinity}. 
\end{proof}

The following result is an elementary fact from linear algebra (proof is included for completeness).
\begin{lemma}\label{lem:symmmxscaling}Suppose that $\B{u},\B{v}\in \R^d$, and $\|\B{v}\|\le \|\B{u}\|$. Then there exists a symmetric matrix $\B{W}\in \R^{d\times d}$ such that $\B{W}\B{u}=\B{v}$, and $-\B{I}\preceq \B{W}\preceq \B{I}$ ($\preceq$ denotes the partial Loewner ordering).
\end{lemma}
\begin{proof}
First we assume that $\|\B{u}\|=\|\B{v}\|$. If $\B{u}=\B{v}$, then $\B{W}=\B{I}$ works, otherwise it is easy to check that 
 \[\B{W}=(\B{u}+\B{v}) (\B{u}+\B{v})^\top /\|\B{u}+\B{v}\|^2 -(\B{u}-\B{v}) (\B{u}-\B{v})^\top /\|\B{u}-\B{v}\|^2\] satisfies the requirements. The general case follows by rescaling.
\end{proof}
Now we are ready to prove our contraction bound.
\begin{proof}[Proof of Theorem \ref{thm:RicciSGS}]
Let $(\bZ_{1:b}, \bZ'_{1:b})$ be a coupling of the two distributions $\joint(\bZ_{1:b}|\btheta)$ and $\joint(\bZ_{1:b}'|\btheta)$ such that 
\begin{equation}\label{eq:ZiZipbnd}
\|\B{Z}_i-\B{Z}_i'\|\le \frac{1}{1+\rho^2 m_i}\|\B{A}_i(\btheta-\btheta')\| \text{ almost surely.}
\end{equation}
The existence of such a coupling follows from Proposition \ref{prop:contractionSGS}.
Given this coupling $(\bZ_{1:b}, \bZ'_{1:b})$, our next step is to couple the two conditional distributions
\begin{align*}
    \joint(\btheta|\bZ_{1:b})\sim \mathcal{N}(\mu_{\btheta}(\bZ_{1:b}),\Sigma_{\btheta}),
    \\
    \joint(\btheta|\bZ'_{1:b})\sim \mathcal{N}(\mu_{\btheta}(\bZ'_{1:b}),\Sigma_{\btheta}),
\end{align*}
where $\Sigma_{\btheta}=\rho^2(\sum_{i=1}^b \B{A}_{i}^\top \B{A}_i)^{-1}$ and $\mu_{b\theta}(\bz_{1:b})=(\sum_{i=1}^b 
\B{A}_{i}^\top \B{A}_i)^{-1} \sum_{i=1}^{b}\B{A}_i^\top \bz_i$.
Since these two Gaussian distributions have the same covariance matrix, coupling them can be done in a straightforward way, and we can see that for the metric $w$ introduced in the statement of Theorem \ref{thm:RicciSGS}, for every $1\le p\le \infty$, we have
\begin{equation}\label{eq:WpPSGSonestepbnd}
W_{p}^{w}(
\B{P}_{\mathrm{SGS}}(\btheta,\cdot),\B{P}_{\mathrm{SGS}}(\btheta',\cdot))\le \l[\E(w(\mu_{\btheta}(\bZ_{1:b}),\mu_{\btheta}(\bZ'_{1:b}))^p)\r]^{1/p},
\end{equation}
where $W_p^w$ denotes Wasserstein distance of order $p$ with respect to the metric $w$. Note that
\[\mu_{\btheta}(\bZ_{1:b})-\mu_{\btheta}(\bZ'_{1:b})=(\sum_{i=1}^b 
\B{A}_{i}^\top \B{A}_i)^{-1} \sum_{i=1}^{b}\B{A}_i^\top (\bZ_i-\bZ'_i).\]
For each $i \in [b]$, we now apply Lemma \ref{lem:symmmxscaling} with $\B{v} = \B{Z}_i - \B{Z}_i'$ and $\B{u} = \B{A}_i(\btheta-\btheta') / (1+\rho^2 m_i)$.
Using \eqref{eq:ZiZipbnd}, the assumption $\nr{\B{v}} \leq \nr{\B{u}}$ of Lemma \ref{lem:symmmxscaling} is satisfied and there exist some symmetric matrices $\B{W}_1,\ldots, \B{W}_b\in \R^{d\times d}$ such that $-\B{I}\preceq \B{W}_i\preceq \B{I}$, and
$$
\bZ_i-\bZ'_i = \B{W}_i \frac{\B{A}_i(\btheta-\btheta')}{1+\rho^2 m_i}, \quad \forall i \in [b].
$$
This yields
\[
\mu_{\btheta}(\bZ_{1:b})-\mu_{\btheta}(\bZ'_{1:b})=\l(\sum_{i=1}^b 
\B{A}_{i}^\top \B{A}_i\r)^{-1} \sum_{i=1}^{b} \frac{\B{A}_i^\top\B{W}_i\B{A}_i}{1+\rho^2 m_i} (\btheta-\btheta')\text{ almost surely}.\]
From the definition of $w$, we can now write
\begin{align*}&w(\mu_{\btheta}(\bZ_{1:b}),\mu_{\btheta}(\bZ'_{1:b}))=\l\|\l(\sum_{i=1}^b 
\B{A}_{i}^\top \B{A}_i\r)^{1/2}(\mu_{\btheta}(\bZ_{1:b})-\mu_{\btheta}(\bZ'_{1:b}))\r\|\\
&= \l\|\l(\sum_{i=1}^b 
\B{A}_{i}^\top \B{A}_i\r)^{-1/2} \sum_{i=1}^{b} \frac{\B{A}_i^\top\B{W}_i\B{A}_i}{1+\rho^2 m_i} (\btheta-\btheta')\r\|\\
&=\l\|\l(\sum_{i=1}^b 
\B{A}_{i}^\top \B{A}_i\r)^{-1/2} \sum_{i=1}^{b} \frac{\B{A}_i^\top\B{W}_i\B{A}_i }{1+\rho^2 m_i} \l(\sum_{i=1}^b 
\B{A}_{i}^\top \B{A}_i\r)^{-1/2} \cdot \l(\sum_{i=1}^b 
\B{A}_{i}^\top \B{A}_i\r)^{1/2}(\btheta-\btheta')\r\|\\
&\le \l\|\l(\sum_{i=1}^b 
\B{A}_{i}^\top \B{A}_i\r)^{-1/2} \l(\sum_{i=1}^{b} \frac{\B{A}_i^\top\B{A}_i }{1+\rho^2 m_i}\r) \l(\sum_{i=1}^b 
\B{A}_{i}^\top \B{A}_i\r)^{-1/2}\r\| w(\btheta,\btheta') \text{ almost surely}.
\end{align*}
Hence the result follows from \eqref{eq:WpPSGSonestepbnd}.
\end{proof}

\subsection{Proof of Corollary \ref{coro:convergence_rates_SGS}}
\label{appendix_proof_coro_convergence_rates}

First, we will show the convergence results in Wasserstein distance of order $p$ for $1\le p<\infty$. Let $(\btheta_0,\btheta'_0)$ be the optimal coupling of the initial distribution $\nu$ and the stationary distribution $\pi_{\rho}$ that achieves the Wasserstein distance of order $p$ for the metric $w$ (see Theorem 4.1 of \cite{Villani2008} for proof of existence), i.e.
\[W_{p}^{w}(\nu,\pi_{\rho})=\|w(\btheta_0,\btheta'_0)\|_{L^p}.
\]
For $i\ge 1$, assuming that $(\btheta_{0:i-1}, \btheta'_{0:i-1})$
has been defined, add two more elements $(\btheta_{i}, \btheta'_{i})$ by defining their conditional distribution based on the past elements as the optimal coupling between $P_{\mathrm{SGS}}(\btheta_{i-1},\cdot)$ and
$P_{\mathrm{SGS}}(\btheta'_{i-1},\cdot)$ achieving the Wasserstein distance of order $p$ for the metric $w$. Using that $K_{p}(\btheta,\btheta')\ge K_{\mathrm{SGS}}$ by Theorem \ref{thm:RicciSGS}, we have
\[\E(w(\btheta_{1},\btheta'_{1})^p|\btheta_0,\btheta'_0)\le (1-K_{\mathrm{SGS}})^p w(\btheta_0,\btheta'_0)^p,\]
and so by the tower property, we have
\[\|w(\btheta_{1},\btheta'_{1})\|_{L^p}\le (1-K_{\mathrm{SGS}}) W_{p}^{w}(\nu,\pi_{
\rho}).\]
Similarly, by induction, it follows that 
\[\|w(\btheta_{i}, \btheta'_{i})\|_{L^p}\le (1-K_{\mathrm{SGS}})^i W_p^w(\nu,\pi_{
\rho}).\]
Now \eqref{eq:W1convergencesec4} for $1
\le p<\infty$ follows by noticing that $\btheta'_{i}\sim \pi_{\rho}$ since the Markov chain $(\btheta'_{j})_{j\ge 0}$ was initialized in its stationary distribution. Finally, the $p=\infty$ case follows from Proposition 3 of \cite{Wassersteininfinity}.

Regarding the convergence rate in total variation distance stated in Theorem \ref{thm:RicciSGS}, we will use Corollary \ref{cor:spectralgapSGS} and Proposition \ref{prop:tvfromspectralgap} detailed below.

\begin{corollary}[Lower bound on the spectral gap of SGS]
\label{cor:spectralgapSGS}
SGS defines a reversible Markov chain. Under Assumptions \assfive and \asssix, its absolute spectral gap $\gamma_{\mathrm{SGS}}^*$ is lower bounded by $K_{\mathrm{SGS}}$, see \eqref{eq:kappasgsdef}. 
\end{corollary}
\begin{proof}
The reversibility follows by a standard argument for data augmentation schemes given in Lemma 3.1 of \cite{liu1994covariance}. 
The lower bound on the absolute spectral gap follows by Proposition 30 of \cite{ollivier2009ricci}.
\end{proof}

The following proposition is well known in the MCMC literature but we have only found a proof for Markov chains on finite state spaces.
Hence for completeness, we include a short proof here.

\begin{proposition}
\label{prop:tvfromspectralgap}
Suppose that $\B{P}(\bz,\cdot)$ is a reversible Markov kernel on a Polish state space $\Omega$ with absolute spectral gap $\gamma^*>0$, and unique stationary distribution $\pi$. Then for any initial distribution $\nu$ that is absolutely continuous with respect to $\pi$, and any number of steps $t\in \Z_+$, we have
\[\nr{\nu \B{P}^{t} - \pi}_{\mathrm{TV}} 
\le \frac{1}{2}\l(\E_{\pi}
\l[  \l(\frac{\mathrm{d} \nu}{\mathrm{d} \pi}\r)^2 \r]  -1\r)^{1/2}
\cdot (1-\gamma^*)^{t}.\]
\end{proposition}
Our proof is based on the following lemma.
\begin{lemma}\label{lemma:revabscont}
Suppose that $\B{Q}(x,\mathrm{d} y)$ is a reversible Markov kernel on a Polish state space $\Omega$ with stationary distribution $\pi$. Then for any distribution $\nu$ that is absolutely continuous with respect to $\pi$, $\nu \B{Q}$ is also absolutely continuous with respect to $\pi$, and for $\pi$-almost every $x\in \Omega$, we have
\[\frac{\mathrm{d} (\nu\B{Q})}{\mathrm{d} \pi}(x)=\l(\B{Q} \l(\frac{\mathrm{d} \nu}{\mathrm{d} \pi}\r)\r)(x).\]
\end{lemma}
\begin{proof}
The claim of the lemma is equivalent to showing that for every bounded measurable function $f:\Omega\to \R$, we have 
\begin{equation}\label{eq:lemmarevabsconteqv}\int_{x\in \Omega}\frac{\mrd (\nu\B{Q})}{\mrd \pi}(x) f(x)\pi(\mrd x)=\int_{x\in \Omega}\l(\B{Q} \l(\frac{\mrd \nu}{\mrd \pi}\r)\r)(x) f(x)\pi(\mrd x) .\end{equation}
Since if we add a constant to $f$, both sides increase by this constant, we can assume without loss of generality that $f$ is non-negative. Under this assumption, we have
\begin{align*}
    &\int_{x\in \Omega}\frac{\mrd (\nu\B{Q})}{\mrd \pi}(x) f(x)\pi(\mrd x)=\int_{x\in \Omega}f(x) (\nu \B{Q})(\mrd x)\\
    &=\int_{x,y\in \Omega}f(x) \nu(\mrd y) \B{Q}(y,\mrd x)=\int_{x,y\in \Omega}f(y) \nu(\mrd x) \B{Q}(x,\mrd y)\\
    &=\int_{x,y\in \Omega}f(y) \frac{\mrd \nu}{\mrd \pi}(x) \pi(\mrd x) \B{Q}(x,\mrd y)\\
    \intertext{ by the monotone convergence theorem (using the non-negativity of $f$)}
    &=\lim_{M\to \infty}\int_{x,y\in \Omega}f(y) \min\l(\frac{\mrd \nu}{\mrd \pi}(x),M\r) \pi(\mrd x) \B{Q}(x,\mrd y)\\
    \intertext{ using the reversibility of $\B{Q}$ (in the equivalent bounded measurable test function formulation)}
    &=\lim_{M\to \infty}\int_{x,y\in \Omega}f(y) \min\l(\frac{\mrd \nu}{\mrd \pi}(x),M\r) \pi(\mrd y) \B{Q}(y,\mrd x)\\
    \intertext{ by the monotone convergence theorem (using the non-negativity of $f$)}
    &=\int_{x,y\in \Omega}f(y) \frac{\mrd \nu}{\mrd \pi}(x) \pi(\mrd y) \B{Q}(y,\mrd x)\\
    &=\int_{y\in \Omega}f(y) \l(\B{Q}\l(\frac{\mrd \nu}{\mrd \pi}\r)\r)(y) \pi(\mrd y), 
\end{align*}
hence \eqref{eq:lemmarevabsconteqv} and the claim of our lemma holds.
\end{proof}

\begin{proof}[Proof of Proposition \ref{prop:tvfromspectralgap}]
We define the Hilbert space $L^2(\pi)$ as measurable functions $f$ on $\Omega$ satisfying $\E_{\pi}(f^2)<\infty$, endowed with the scalar product $\inner{f}{g}_{\pi}=\int_{\bz\in \Omega}f(\bz)g(\bz)\pi(\mrd \bz)$. Let us define the linear operator $\B{\Pi}(f)(\bz):=\E_{\pi}(f)$ for any $f\in L^2(\pi)$, $\bz\in \Omega$.

Using Lemma \ref{lemma:revabscont} with $\B{Q}=\B{P}^t$, it follows that
\begin{align}
    \nonumber&\|\nu \B{P}^t-\pi\|_{\mathrm{TV}}=\frac{1}{2}\int_{x\in \Omega} \l|\frac{\mrd \nu \B{P}^t}{\mrd \pi}(x)-1\r| \pi (\mrd x)
    \nonumber\intertext{using Jensen's inequality, we have}
    \label{eq:tvsqrtbnd}&\le  \frac{1}{2} \sqrt{ \int_{x\in \Omega} \l(\frac{\mrd \nu \B{P}^t}{\mrd \pi}(x)-1\r)^2 \pi (\mrd x)}.
\end{align}
Using Lemma \ref{lemma:revabscont} again, the integral inside the square root can be further bounded as
\begin{align*}
    &\int_{x\in \Omega} \l(\frac{\mrd \nu \B{P}^t}{\mrd \pi}(x)-1\r)^2 \pi (\mrd x)=
    \int_{x\in \Omega} \l(\l(\B{P}^t \l(\frac{\mrd \nu}{\mrd \pi}\r)\r)(x)-1\r)^2 \pi (\mrd x)\\
    &=\int_{x\in \Omega} \l(\l((\B{P}^t-\B{\Pi}) \l(\frac{\mrd \nu}{\mrd \pi}\r)\r)(x)\r)^2 \pi(\mrd x)=\int_{x\in \Omega} \l(\l((\B{P}-\B{\Pi})^t \l(\frac{\mrd \nu}{\mrd \pi}\r)\r)(x)\r)^2 \pi(\mrd x)\\
    &=\int_{x\in \Omega} \l(\l((\B{P}-\B{\Pi})^t \l(\frac{\mrd \nu}{\mrd \pi}-1\r)\r)(x)\r)^2 \pi(\mrd x)=\inner{\frac{\mrd \nu}{\mrd \pi}-1}{ (\B{P}-\B{\Pi})^{2t} \l(\frac{\mrd \nu}{\mrd \pi}-1\r)}_{\pi}\\
    &\le\|\B{P}-\B{\Pi}\|_{\pi}^{2t} \l\|\frac{\mrd \nu}{\mrd \pi}-1\r\|_{\pi}^2=(1-\gamma^*)^{2t}\l\|\frac{\mrd \nu}{\mrd \pi}-1\r\|_{\pi}^2,
\end{align*}
and the claim of the proposition follows by substituting this into  \eqref{eq:tvsqrtbnd}.
\end{proof}

Now we are ready to prove our convergence bound in total variation distance.
\begin{proof}[Proof of Theorem \ref{prop:compcomplexityTV}]
    From Corollary \ref{cor:spectralgapSGS}, we know that the absolute spectral gap of SGS satisfies that $\gamma^*\ge K_{\mathrm{SGS}}$ (defined in \eqref{eq:kappasgsdef}), and Proposition \ref{prop:tvfromspectralgap} implies that
    \begin{align*} 
    \nr{\nu \B{P}_{\mathrm{SGS}}^{t} - \pi_{\rho}}_{\mathrm{TV}}
    &\le \sqrt{\E_{\pi_{\rho}}
    \l[\l(\frac{\mathrm{d} \nu}{\mathrm{d} \pi_{\rho}}\r)^2\r]-1}
    \cdot (1-\gamma^*)^{t}\\
    &\le \sqrt{\E_{\pi_{\rho}}
    \l[\l(\frac{\mathrm{d} \nu}{\mathrm{d} \pi_{\rho}}\r)^2\r]-1}
    \cdot (1-K_{\mathrm{SGS}})^{t}.
    \end{align*}
\end{proof}

\subsection{Proof of Theorem \ref{prop:compcomplexityWass}}
\label{subsec:mixing_time_bounds_Wasserstein}

\begin{proof}[Proof of Theorem \ref{prop:compcomplexityWass}]
From Theorem \ref{thm:Wassersteinpipirhosinglesplitting}, it follows that if $\rho$ is chosen as in \eqref{eq:rho2singlesplittingwass}, then \begin{equation}\label{eq:w1pirhopi}W_1(\pi_{\rho},\pi)\le \frac{\epsilon}{2}\cdot \frac{\sqrt{d}}{\sqrt{m_1}}.
\end{equation}
From Proposition 1 part (ii) in \cite{durmus2018high} it follows that for the initial distribution $\delta_{\bthetastar}$ (Dirac measure at $\bthetastar$), we have
\[W_1(\delta_{\bthetastar},\pi)\le W_2(\delta_{\bthetastar},\pi)\le \frac{\sqrt{d}}{\sqrt{m_1}},\]
and hence by combining this with \eqref{eq:w1pirhopi} using the triangle inequality and the assumption $\epsilon\le 1$, it follows that 
\[W_1(\delta_{\bthetastar},\pi_{\rho})\le  \frac{3}{2}\frac{\sqrt{d}}{\sqrt{m_1}}.\]
Now from Theorem \ref{thm:RicciSGS}, it follows that the coarse Ricci curvature of SGS is lower bounded by
\[K_{\mathrm{SGS}}:=\frac{\rho^2 m_1}{1+\rho^2 m_1},\]
and therefore by Corollary 21 of \cite{ollivier2009ricci}, we have
\begin{align*}W_1(P_{\mathrm{SGS}}^{t}(\bthetastar,\cdot),\pi_{\rho})\le W_1(\delta_{\bthetastar},\pi_{\rho})\cdot (1-K_{\mathrm{SGS}})^{t}\le \frac{\epsilon}{2}\cdot \frac{\sqrt{d}}{\sqrt{m_1}}.
\end{align*}
The claim of the theorem now follows by the triangle inequality.
\end{proof}

\subsection{Complexity Bounds for Implementing SGS by Rejection Sampling}
\label{sec:RS}
The following bound is a standard result in rejection sampling; see for instance Section 2.3 of \cite{robert2013monte}.

\begin{lemma}\label{lemma:rejectiongeneral}
Suppose that $\mu(\bz)=\tilde{\mu}(\bz)/\tilde{Z}$ is the target density on $\R^d$, and $\nu(\bz)$ is the proposal density (both absolutely continuous w.r.t. the Lebesgue measure). Here $\tilde{\mu}(\bz)$ is the unnormalized target and $\tilde{Z}$ is the normalising constant (which is typically unknown). Suppose that the condition
\begin{equation}
    \tilde{\mu}(\bz)\le M \nu(\bz)
\end{equation}
holds for some constant $M<\infty$ for every $\bz\in \R^d$ . Under this assumption, if we take samples $\bZ_1,\bZ_2,\ldots$ from $\nu$ and accept $\bZ_i$ with probability $\frac{\tilde{\mu}(\bZ_i)}{M \nu(\bZ_i)}$, then the accepted samples will be distributed according to $\mu$. Moreover, the expected number of samples taken until the first acceptance is equal to $M/\tilde{Z}$.
\end{lemma}
The following lemma gives a complexity bound for rejection sampling for log-concave distributions. We assume that we have access to an approximation of the minimum of the strongly convex and smooth potential $U$, which will be denoted by $\tilde{\bz}$. The quality of this approximation is taken into account in the proposal distribution using the norm of $\grad U(\tilde{\bz})$.
\begin{lemma}[Rejection sampling upper bound for 
log-concave densities] \label{lem:rejectionupperbnd}
Suppose \,\, that $\mu(\bz)\propto \exp(-U(\bz))$ is a distribution on $\R^d$ such that $U$ is twice differentiable and 
\begin{align}
    A \B{I}_d\preceq \grad^2 U(\bz)\preceq B \B{I}_d\label{eq:HessUAB}
\end{align}
for some $0<A\le B$ (strongly convex and smooth). Let $\bz^{*}$ be the unique minimizer of $U$, $\tilde{\bz}$ another point (an approximation of $\bz^{*}$), and $\nu(\bz)=\mathcal{N}(\bz;\tilde{\bz},\tA^{-1}\B{I}_d)$, where
\begin{equation}\label{eq:tAdef}
    \tA=A+\frac{\|\grad U(\tbz)\|^2}{2d}-\sqrt{\frac{\|\grad U(\tbz)\|^4}{4d^2} + \frac{A\|\grad U(\tbz)\|^2}{d}}.
\end{equation}
Suppose that we take samples $\bZ_1,\bZ_2,\ldots$ from $\nu$, and accept them with probability
\[\PP(\bZ_j\text{ is accepted})=\exp\left(-\frac{\|\nabla U(\tbz)\|^2}{2(A-\tilde{A})} - [U(\bz)-U(\tbz))] +\frac{\tilde{A} \|\bz-\tbz\|^2}{2}\right).\]
Then these accepted samples are distributed according to $\mu$. Moreover, the expected number of samples taken until one is accepted is less than or equal to $\l(B/\tilde{A}\r)^{d/2}\cdot \exp\l[\frac{\|\grad U(\tbz)\|^2}{2}\l(\frac{1}{A-\tA}-\frac{1}{B}\r)\r]$.
\end{lemma}
\begin{proof}
    The proposal density equals
    \begin{align*}
        \nu(\bz)&=\mathcal{N}(\bz;\tilde{\bz},\tA^{-1}\B{I}_d)\\
        &=\exp\l(-\frac{\tA \|\bz-\tbz\|^2}{2}\r)\cdot \l(\frac{\tA}{2\pi}\r)^{d/2}.
    \end{align*}
    We define the unnormalized version of $\mu$ as 
    \[\tilde{\mu}(\bz)=\exp(-[U(\bz)-U(\tilde{\bz})])\cdot \l(\frac{\tA}{2\pi}\r)^{d/2}.\]
    Notice that
    \begin{align*}
        U(\bz)-U(\tilde{\bz})&=\inner{\int_{t=0}^{1} \grad U(\tilde{\bz}+t (\bz-\tbz)) \mathrm{d}t}{\bz-\tbz}.\\
        \intertext{By the intermediate value theorem, there is some $\bz(t)$ such that}
        &=\inner{\grad U(\tilde{\bz})}{\bz-\tbz}+\inner{\bz-\tbz}{\l(\int_{t=0}^{1} t\grad^2 U(\bz(t)) \mathrm{d}t\r)^\top(\bz-\tbz)},\\
        \intertext{so using the assumption \eqref{eq:HessUAB} it follows that}
        &\ge -\|\grad U(\tbz)\|\|\bz-\tbz\|+\frac{A}{2} \|\bz-\tbz\|^2.
    \end{align*}
    Based on this, one gets
    \begin{align*}
        \frac{\tilde{\mu}(\bz)}{\nu(\bz)}\le \exp\l(\|\grad U(\tbz)\|\cdot \|\bz-\tbz\|-\frac{A-\tA}{2}\|\bz-\tbz\|^2\r)\le \exp\l(\frac{\|\grad U(\tbz)\|^2}{2 (A-\tA)}\r).
    \end{align*}
    Hence we have $\tilde{\mu}(\bz)\le M \nu(\bz)$ for $M=\exp\l(\frac{\|\grad U(\tbz)\|^2}{2 (A-\tA)}\r)$.
    
    For the normalising constant, we have
    \begin{align*}
        \tilde{Z}&=\int_{\bz\in \R^d} \tilde{\mu}(\bz)\mathrm{d}\bz=\exp(U(\tbz)-U(\bz^*))\cdot \l(\frac{\tA}{2\pi}\r)^{d/2} \cdot \int_{\bz\in \R^d} \exp(-(U(\bz)-U(\bz^*))) \mathrm{d} \bz\\
        \intertext{ using Taylor's expansion with second order remainder term, and assumption \eqref{eq:HessUAB} yields}
        &\ge \exp(U(\tbz)-U(\bz^*))\cdot \l(\frac{\tA}{2\pi}\r)^{d/2} \cdot 
         \int_{\bz\in \R^d} \exp\l(-\frac{B}{2}\|\bz-\bz^*\|^2\r) \mathrm{d} \bz\\
         &= \l(\frac{\tA}{B}\r)^{d/2} \cdot \exp\l(U(\tbz)-U(\bz^*)\r)\ge \l(\frac{\tA}{B}\r)^{d/2} \exp\l(\frac{\|\grad U(\tbz)\|^2}{2B}\r),
    \end{align*}
    where in the last step we have used the fact that
    for $\bz'=\tbz-\frac{\|\grad U(\tbz)\|}{B}$, we have 
    \begin{align*}
    &U(\tbz)-U(\bz^*)\\
    &\ge U(\tbz)-U(\bz')=\inner{\int_{t=0}^{1} \grad U(\tbz+t(\bz'-\tbz)) \mathrm{d}t}{\tbz-\bz'}\\
    \intertext{ using the fact that $\bz^*$ is the minimum of $U$.}
    \intertext{By the intermediate value theorem, there is some $\tbz(t)\in \R^d$ such that}
    &= \inner{\grad U(\tbz)}{\tbz-\bz'}+\inner{\tbz-\bz'}{\l(\int_{t=0}^{1} t \grad^2 U(\tbz(t)) \mathrm{d}t\r)\cdot (\tbz-\bz')}\\
    &\ge \inner{\grad U(\tbz)}{\tbz-\bz'}-\frac{B}{2} \|\tbz-\bz'\|^2=\frac{\|\grad U(\tbz)\|^2}{2B}.
    \end{align*}
    Now it follows by Lemma \ref{lemma:rejectiongeneral} and the above bound on $\tilde{Z}$ that the expected number of samples until the first acceptance is less than or equal to
    \begin{equation*}
    E(\tA):=\exp\l(\|\grad U(\tbz)\|^2 \l(\frac{1}{2 (A-\tA)}-\frac{1}{2B}\r)\r)\l(\frac{B}{\tA}\r)^{d/2}.
    \end{equation*}
    The parameter $\tA$ in \eqref{eq:tAdef} is chosen such that $E(\tA)$ is minimized. Note that the minimizer of $E(\tA)$ is the same as the minimizer of  \begin{align*}
    \log(E(\tA))=\frac{d}{2}\log(B)-\frac{\|\grad U(\tbz)\|^2}{2B}+\frac{\|\grad U(\tbz)\|^2}{2(A-\tA)}-\frac{d}{2}\log(\tA).
    \end{align*}
    It is easy to check that this is a strictly convex function of $\tA$ on the interval $(0,A)$, and hence the unique minimum is taken at a point where the derivative is zero. This point, denoted by $\tA_{\min}$, thus satisfies 
    \begin{align*}
        &\l.\frac{\partial \log(E(\tA))}{\partial \tA}\r|_{\tA=\tA_{\min}}=\frac{\|\grad U(\tbz)\|^2}{2(A-\tA)^2}-\frac{d}{2}\cdot \frac{1}{\tA}=0.
        \intertext{Hence by rearrangement }
        &(\tA-A)^2-(\|\grad U(\tbz)\|^2/d)\tA=0\\
        &\tA^2-(2A+\|\grad U(\tbz)\|^2/d)\tA + A^2 =0\\
        &\tA=\frac{(2A+\|\grad U(\tbz)\|^2/d)\pm \sqrt{(2A+\|\grad U(\tbz)\|^2/d)^2-4A^2}}{2}\\
        &\quad =A +\|\grad U(\tbz)\|^2/(2d)\pm \sqrt{\|\grad U(\tbz)\|^4/(4d^2) + A\|\grad U(\tbz)\|^2/d}.
    \end{align*}
    Only the solution with the $-$ sign falls in the interval $(0,A)$, hence it is the minimizer of $M/\tilde{Z}$.
\end{proof}

\begin{proof}[Proof of Proposition \ref{prop:rejectionsamplingcomplexity}]
The fact that the accepted samples are distributed according to $\joint(\bz_i|\btheta)$ and the formula \eqref{eq:cor2Eidef} about the expected number of samples until acceptance follows from Lemma \ref{lem:rejectionupperbnd}. 

Let $G:=\|\grad V_i(\tbz_i(\btheta))\|$, then
$\tA_i=1/\rho^2+m_i+G^2/(2d_i)-\sqrt{G^4/(4d_i^2) + G^2\l(1/\rho^2+m_i\r)/d_i}$, and we have
\begin{align}
\log(E_i)&=\frac{d_i}{2}\log\l(\frac{1/\rho^2+M_i}{1/\rho^2+m_i+G^2/(2d_i)-\sqrt{G^4/(4d_i^2) + G^2\l(1/\rho^2+m_i\r)/d_i} }\r)\nonumber
\\
&+\frac{G^2}{2}\l(\frac{1}{\sqrt{G^4/(4d_i^2) + G^2\l(1/\rho^2+m_i\r)/d_i}-G^2/(2d_i)}-\frac{1}{1/\rho^2+M_i}\r).\label{eq:Eisecondpart}
\end{align}
For the first part, notice that 
\begin{align*}&\log\l(\frac{1/\rho^2+M_i}{1/\rho^2+m_i+G^2/(2d_i)-\sqrt{G^4/(4d_i^2) + G^2\l(1/\rho^2+m_i\r)/d_i} }\r)\\
&=\log\l(\frac{1/\rho^2+M_i}{1/\rho^2+m_i}\r)+
\log\l(\frac{1/\rho^2+m_i}{1/\rho^2+m_i+G^2/(2d_i)-\sqrt{G^4/(4d_i^2) + G^2\l(1/\rho^2+m_i\r)/d_i}}\r)\\
&=\log\l(1+\frac{\rho^2 (M_i-m_i)}{1+\rho^2 m_i}\r)+\log\l(\frac{1}{1+c-\sqrt{c^2+2c}}\r),
\end{align*}
where $c=\frac{G^2/(2d_i)}{1/\rho^2+m_i}$.
Now using the fact that $\log(1+x)\le x$ for $x>0$, and that $\log\l(\frac{1}{1+c-\sqrt{c^2+2c}}\r)\le \sqrt{2c}$ for $c\ge 0$, it follows that we have
\begin{align*}
    &\frac{d_i}{2}\log\l(\frac{1/\rho^2+M_i}{1/\rho^2+m_i+G^2/(2d_i)-\sqrt{G^4/(4d_i^2) + G^2\l(1/\rho^2+m_i\r)/d_i} }\r)\\
    &\le \frac{d_i}{2}\l(\frac{\rho^2 (M_i-m_i)}{1+\rho^2 m_i}+  \frac{G}{\sqrt{d_i(1/\rho^2+m_i)}}\r).
\end{align*}
For the second part \eqref{eq:Eisecondpart},
\begin{align*}
&\frac{G^2}{2}\l(\frac{1}{\sqrt{G^4/(4d_i^2) + G^2\l(1/\rho^2+m_i\r)/d_i}-G^2/(2d_i)}-\frac{1}{1/\rho^2+M_i}\r)\\
&=\frac{d_i}{\sqrt{1+4\l(1/\rho^2+m_i\r)d_i/G^2}-1}-\frac{G^2}{2}\cdot \frac{1}{1/\rho^2+M_i}
\intertext{using the fact that $\frac{1}{\sqrt{1+x}-1}\le \frac{2}{\sqrt{x}}$ for $x\ge 2$, for $G\le \sqrt{2d_i(1/\rho^2+m_i)}$, we have}
&\le G\cdot \frac{\sqrt{d_i}}{\sqrt{1/\rho^2+m_i}} -\frac{G^2}{2}\cdot \frac{1}{1/\rho^2+M_i}.
\end{align*}
Hence by combining these terms, we obtain that for $G\le \sqrt{2d_i(1/\rho^2+m_i)}$,
\begin{equation*}
\label{eq:logEifinalbnd}
\log(E_i)\le \frac{d_i}{2}\frac{\rho^2 (M_i-m_i)}{1+\rho^2 m_i}+ G\cdot \frac{3}{2}\cdot  \frac{\sqrt{d_i}}{\sqrt{(1/\rho^2+m_i)}}-\frac{G^2}{2}\cdot \frac{1}{1/\rho^2+M_i}
\end{equation*}
Under the first part of assumption \eqref{eq:lessthan2bnd}, $\rho^2(2d_i (M_i-m_i)-m_i)\le 1$, one can check that 
$\frac{d_i}{2}\frac{\rho^2 (M_i-m_i)}{1+\rho^2 m_i}\le \frac{1}{4}$. Using the second part of \eqref{eq:lessthan2bnd}, $G\le \frac{2}{7}\cdot \frac{\sqrt{1/\rho^2+m_i}}{\sqrt{d_i}}$, it follows that
$G\cdot \frac{3}{2}\cdot  \frac{\sqrt{d_i}}{\sqrt{(1/\rho^2+m_i)}}-\frac{G^2}{2}\cdot \frac{1}{1/\rho^2+M_i}\le \log(2)-\frac{1}{4}$, so $\log(E_i)\le \log(2)$ and our claim holds.
\end{proof}

\subsection{Proof of Theorems \ref{PROP:COMPCOMPLEXITYTV_SINGLE} and \ref{prop:compcomplexityTV}}
\label{subsec:mixing_time_bounds_TV}

The next two lemmas will be used for obtaining our total variation distance convergence rates.
\begin{lemma}\label{lem:grad2upperbnd}
Suppose that $U:\R^d\to \R$ is continuously differentiable and $M$-gradient-Lipschitz. Then for every $\bx\in \R^d$, we have
\[\|\grad U(\bx)\|^2\le 2M (U(\bx)-\inf_{\bx\in \R^d} U(\bx)).\]
\end{lemma}
\begin{proof}
    Let $\bx'=\bx-\grad U(\bx)/M$, then we have
    \begin{align*}U(\bx)-U(\bx')&=\int_{t=0}^{1} \inner{\grad U(\bx+t(\bx'-\bx))}{\bx-\bx'}\mathrm{d}t\\
    &=\inner{\grad U(\bx)}{\bx-\bx'}+\int_{t=0}^{1} \inner{\grad U(\bx+t(\bx'-\bx))-\grad U(\bx)}{\bx-\bx'}\mathrm{d}t\\
    \intertext{using the $M$-gradient Lipschitz property}
    &\ge \inner{\grad U(\bx)}{\bx-\bx'}-\int_{t=0}^{1} Mt\|\bx-\bx'\|^2\mathrm{d}t\\
    &\ge \frac{\|\grad U(\bx)\|^2}{2M},
    \end{align*}
    hence the result.
\end{proof}

\begin{lemma}\label{lem:Crho}
Suppose that Assumptions \asszero, \asstwo, \assfive, \assseven\  and $\mathrm{det}\l(\sum_{i=1}^b m_i\B{A}_i^\top\B{A}_i\r)>0$ hold.
Let $\bthetastar$ be the unique minimizer of $U(\btheta)=\sum_{i=1}^{b} U_i(\B{A}_i\btheta)$, and \[\nu(\btheta)=\mathcal{N}\l(\btheta;\bthetastar, \l(\sum_{i=1}^{b}M_i\B{A}_i^\top\B{A}_i\r)^{-1}\r).\] 
If $b=1$, $d=d_1$, and $\B{A}_1$ is of full rank, 
then for any $\rho>0$, we have $\frac{\nu(\btheta)}{\pi_{\rho}(\btheta)}\le C_{\rho}$ for every $\btheta\in \R^d$, where
\begin{align}
C_{\rho}&:=(1+\rho^2 M_1)^{d/2}\cdot \l(\frac{M_1}{m_1}\r)^{\frac{d}{2}}.\label{eq:Crho_single}
\end{align}
More generally, for multiple splitting, for $\rho^2\le \frac{1}{6 \sigma^2_U}$, we have $\frac{\nu(\btheta)}{\pi_{\rho}(\btheta)}\le C_{\rho}$ for every $\btheta\in \R^d$, where
\begin{align}
C_{\rho}&:=\exp\pr{d\sigma^2_U + \rho^4(2+d)\sigma_U^4}
      \cdot\prod_{i=1}^{b}(1+\rho^2 M_i)^{d_i/2}\cdot \frac{\mathrm{det}\l(\sum_{i=1}^b M_i\B{A}_i^\top\B{A}_i\r)^{1/2}}{\mathrm{det}\l(\sum_{i=1}^b m_i\B{A}_i^\top\B{A}_i\r)^{1/2}},\label{eq:Crho}
      \end{align}
with $\sigma_U^2$ defined as in \eqref{eq:sigma2def}.
\end{lemma}
\begin{proof}
Let $U^{\rho}$ be defined as in \eqref{eq:Urhodef}. By \eqref{eq:proof_theorem_2_2_needed_for_Lemma_26} and \eqref{eq:proof_theorem_2_4_needed_for_Lemma_26}, we have
    \begin{align}
    \exp(-U^{\rho}(\btheta))
    &\le \exp\l(-U(\btheta)\r)\cdot 
    \prod_{i=1}^{b}\frac{1}{\l(1+\rho^2 m_i\r)^{d_i/2}} \cdot \exp\l(\sum_{i=1}^b \frac{\rho^2 \nr{\grad U_i(\B{A}_i\btheta)}^2}{2 (1+\rho^2 m_i)}\r),\nonumber\\
    \exp(-U^{\rho}(\btheta))
    &\ge \exp\l(-U(\btheta)\r)\cdot 
    \prod_{i=1}^{b}\frac{1}{\l(1+\rho^2 M_i\r)^{d_i/2}} \cdot \exp\l(\sum_{i=1}^b \frac{\rho^2 \nr{\grad U_i(\B{A}_i\btheta)}^2}{2 (1+\rho^2 M_i)}\r)\label{eq:Urhoulower}.
    \end{align}
    Using \eqref{eq:Urhoulower}, we have
    \begin{align}
    \nonumber
    \pi_{\rho}(\btheta)&=\frac{\exp(-U^{\rho}(\btheta))}{Z_{\rho}}\\
    \nonumber
       &\ge \frac{\exp(-U(\btheta))}{Z_{\rho}} \cdot \frac{1}{\prod_{i=1}^{b}(1+\rho^2 M_i)^{d_i/2}}\\
    \label{eq:pirholowerbnd1}&\ge \frac{\exp\l(-U(\bthetastar)-\frac{1}{2} (\btheta-\bthetastar)^\top(\sum_{i=1}^bM_i\B{A}_i^\top\B{A}_i)(\btheta-\bthetastar)\r)}{Z_{\rho}}\cdot \frac{1}{\prod_{i=1}^{b}(1+\rho^2 M_i)^{d_i/2}}.
\end{align}
To lower bound $\pi_{\rho}(\btheta)$, we need to upper bound $Z_{\rho}$. Using Lemma  \ref{lem:ratioofnormalizingconstants}, we can do this based on an upper bound on $Z$. Using \assfive, we have 
\begin{align}
    Z &= \int_{\mathbb{R}^d} \exp\pr{- \sum_{i=1}^b U_i(\B{A}_i\btheta)}\mathrm{d}\btheta \nonumber \\
    &\leq \exp\pr{- \sum_{i=1}^b U_i(\B{A}_i\bthetastar)}\int_{\mathbb{R}^d} \exp\pr{- \sum_{i=1}^b\frac{m_i}{2} \nr{\B{A}_i\btheta - \B{A}_i\bthetastar}^2}\mathrm{d}\btheta \nonumber \\
    &= \exp\pr{- U(\bthetastar)}(2\pi)^{d/2}\mathrm{det}\l(\sum_{i=1}^b m_i\B{A}_i^\top\B{A}_i\r)^{-1/2}.  \label{eq:Zupperbnd}
\end{align}
Note that the proposal density is of the form
    \begin{align}
    \nu(\btheta)&=\mathcal{N}\pr{\btheta;\bthetastar, \pr{\sum_{i=1}^{b}M_i\B{A}_i^\top\B{A}_i}^{-1}}\nonumber\\
    &=\exp\l(-\frac{1}{2} (\btheta-\bthetastar)^\top(\sum_{i=1}^bM_i\B{A}_i^\top\B{A}_i)(\btheta-\bthetastar)\r)\frac{\mathrm{det}\l(\sum_{i=1}^b M_i\B{A}_i^\top\B{A}_i\r)^{1/2}}{(2\pi)^{d/2}}\label{eq:nubthetadens}.
    \end{align}

Under the assumption that $b=1$, $d_1=d$ and $\B{A}_1$ is of full rank, we have $Z_{\rho}=Z$ by Lemma \ref{lem:ratioofnormalizingconstants}. The claim of the lemma in this single splitting case follows by comparing \eqref{eq:nubthetadens}, \eqref{eq:pirholowerbnd1} and using \eqref{eq:Zupperbnd}.

More generally, from Lemma \ref{lem:ratioofnormalizingconstants}, it follows for $\rho^2 \leq \frac{1}{6 \sigma^2_U}$ that
\begin{align}
    Z_{\rho} &\leq Z \exp\pr{\mathbb{E}_{\pi}(\bar{B}(\btheta)) + \rho^4(2+d)\sigma_U^4} \nonumber \\
    &\leq \exp\pr{- U(\bthetastar)}(2\pi)^{d/2}\mathrm{det}\l(\sum_{i=1}^b m_i\B{A}_i^\top\B{A}_i\r)^{-1/2} \exp\pr{d\sigma^2_U + \rho^4(2+d)\sigma_U^4}, \nonumber
\end{align}
where, in the last line, we used the fact that $\mathbb{E}_{\pi}(\bar{B}(\btheta)) \leq d \sigma^2_U$, see Lemma \ref{lem:momgenbeta2}. The claim of the lemma in this multiple splitting case now follows by comparing \eqref{eq:nubthetadens} and \eqref{eq:pirholowerbnd1}, and using the above upper bound on $Z_{\rho}$.
\end{proof}

Now we are ready to prove our convergence bound in total variation distance.

\begin{proof}[Proof of Theorem \ref{PROP:COMPCOMPLEXITYTV_SINGLE}]
 From Propositions \ref{proposition:2} and \ref{prop:IUUrho}, a sufficient condition to satisfy $\|\pi_{\rho} - \pi\|_{\mathrm{TV}}\le \epsilon/2$ is to have 
    $$
    \rho^2 \leq \frac{\epsilon}{dM_1}.
    $$
    From Corollary \ref{cor:spectralgapSGS}, we know that the absolute spectral gap of SGS satisfies that $\gamma^*\ge K_{\mathrm{SGS}}$ (defined in \eqref{eq:kappasgsdef}), and Proposition \ref{prop:tvfromspectralgap} implies that
\begin{align*} 
\nr{\nu \B{P}_{\mathrm{SGS}}^{t} - \pi_{\rho}}_{\mathrm{TV}}
&\le \sqrt{\E_{\pi_{\rho}}
\l[\l(\frac{\mathrm{d} \nu}{\mathrm{d} \pi_{\rho}}\r)^2\r]-1}
\cdot (1-\gamma^*)^{t}\\
&\le \sqrt{\E_{\nu}
\l(\frac{\mathrm{d} \nu}{\mathrm{d} \pi_{\rho}}\r)}\cdot (1-K_{\mathrm{SGS}})^{t}\\
&\le \sqrt{C_{\rho}}(1-K_{\mathrm{SGS}})^{t},
\end{align*}
where in the last step we have used Lemma \ref{lem:Crho} ($C_{\rho}$ is defined as in Equation \ref{eq:Crho_single}).
By some algebra, using the definition of $t_{\mathrm{mix}}(\epsilon;\nu)$, and the fact that $\frac{1}{\log(1/(1-x))}\le \frac{1}{x}$ for $0<x<1$, the above bound implies that
\begin{equation*}
\nr{\nu \B{P}_{\mathrm{SGS}}^{t(\epsilon)} - \pi_{\rho}}_{\mathrm{TV}}\le \frac{\epsilon}{2},
\end{equation*}
with the choice
\begin{equation}
    t \ge  \frac{\log\l(\frac{2}{\epsilon}\r) + C/2}{K_{\mathrm{SGS}}}. \label{eq:t_epsilon_ini_1}
\end{equation} 
Here
\[C=\frac{5d}{8} + \frac{d}{2}\log\l(\frac{ M_1}{m_1}\r).
\]
With the above choice for $\rho^2$ and the condition \eqref{eq:t_epsilon_ini_1}, the claim of Theorem \ref{PROP:COMPCOMPLEXITYTV_SINGLE} then follows by the triangle inequality.
\end{proof}

\begin{proof}[Proof of Theorem \ref{prop:compcomplexityTV}]
From \eqref{eq:proof_theorem_2}, we have for $\rho^2 \leq \frac{1}{6 \sigma^2_U}$,
\begin{equation*}
\|\pi_{\rho} - \pi\|_{\mathrm{TV}}\le \rho^2 \frac{1}{2}\sum_{i=1}^b d_i M_i + \rho^4 \sigma_{U}^4\pr{2+\frac{3}{2}d}.
\end{equation*}
Then, a sufficient condition to satisfy $\|\pi_{\rho} - \pi\|_{\mathrm{TV}}\le \epsilon/2$ is to have
\begin{align*}
    \rho^2 \frac{1}{2}\sum_{i=1}^b d_i M_i + \rho^4 \sigma_{U}^4\pr{2+\frac{3}{2}d} &\le \frac{\epsilon}{2} \\
    \rho^4 \sigma_{U}^4\pr{2+\frac{3}{2}d} + \frac{1}{2}\rho^2 \sum_{i=1}^b d_i M_i - \frac{\epsilon}{2} &\le 0 \\
    R^2 \sigma_{U}^4\pr{2+\frac{3}{2}d} + R \frac{1}{2}\sum_{i=1}^b d_i M_i - \frac{\epsilon}{2} &\le 0, \ \text{ with $R = \rho^2$.}
\end{align*}
This inequality is satisfied under the condition \begin{equation*}\label{eq:rho2TVquadratic_ini}
    \rho^2  \le \frac{\displaystyle\sum_{i=1}^bd_i M_i\pr{\sqrt{1 + 8\epsilon \sigma_U^4\pr{2+\frac{3}{2}d} \pr{\displaystyle\sum_{i=1}^bd_i M_i}^{-2}}-1}}{4\sigma_U^4\pr{2+\frac{3}{2}d}} \wedge \frac{1}{6 \sigma^2_U}.
\end{equation*}
From Corollary \ref{cor:spectralgapSGS}, we know that the absolute spectral gap of SGS satisfies that $\gamma^*\ge K_{\mathrm{SGS}}$ (defined in \eqref{eq:kappasgsdef}), and Proposition \ref{prop:tvfromspectralgap} implies that
\begin{align*} 
\nr{\nu \B{P}_{\mathrm{SGS}}^{t} - \pi_{\rho}}_{\mathrm{TV}}
&\le \sqrt{\E_{\pi_{\rho}}
\l[\l(\frac{\mathrm{d} \nu}{\mathrm{d} \pi_{\rho}}\r)^2\r]-1}
\cdot (1-\gamma^*)^{t}\\
&\le \sqrt{\E_{\nu}
\l(\frac{\mathrm{d} \nu}{\mathrm{d} \pi_{\rho}}\r)}\cdot (1-K_{\mathrm{SGS}})^{t}\\
&\le \sqrt{C_{\rho}}(1-K_{\mathrm{SGS}})^{t},
\end{align*}
where in the last step we have used Lemma \ref{lem:Crho} ($C_{\rho}$ is defined as in \eqref{eq:Crho}). 
Again, by some algebra, using the definition of $t_{\mathrm{mix}}(\epsilon;\nu)$, and the fact that $\frac{1}{\log(1/(1-x))}\le \frac{1}{x}$ for $0<x<1$, the above bound implies that
\begin{equation*}
\nr{\nu \B{P}_{\mathrm{SGS}}^{t(\epsilon)} - \pi_{\rho}}_{\mathrm{TV}}\le \frac{\epsilon}{2},
\end{equation*}
with the choice
\begin{equation}
    t \ge  \frac{\log\l(\frac{2}{\epsilon}\r) + C/2}{K_{\mathrm{SGS}}}. \label{eq:t_epsilon_ini}
\end{equation} 
Here
\[C=d\sigma^2_U + \rho^4(2+d)\sigma_U^4 + \frac{17}{32} \sum_{i=1}^{b}d_i + \frac{1}{2}\log\l(\frac{\mathrm{det}\l(\sum_{i=1}^b M_i\B{A}_i^\top\B{A}_i\r)}{\mathrm{det}\l(\sum_{i=1}^b m_i\B{A}_i^\top\B{A}_i\r)}\r).
\]
With the above choice for $\rho^2$ and the condition \eqref{eq:t_epsilon_ini}, the claim of Theorem \ref{prop:compcomplexityTV} then follows by the triangle inequality.
\end{proof}

\subsection{Additional Details for the Toy Gaussian Example}
\label{appendix:toy_Gaussian_example_2}

This section gives additional details concerning the results depicted on Figure \ref{fig:toy_gaussian_1}.
For each splitting strategy introduced in Section \ref{subsec:toy_Gaussian_model}, we give explicit formulas for the bounds on both TV and 1-Wasserstein distances.

\subsubsection{Splitting Strategy 1}

Starting from an initial value $\theta^{[0]} \sim \nu$, we now show the explicit form of the Markov transition kernel $\nu P_{\mathrm{SGS}}^t$ after $t$ iterations.
To this purpose, we take advantage that the $\theta$-chain corresponds in this case to an auto-regressive process of order 1.
Indeed, the conditional distributions of $\theta$ and $\bz_{1:b}$ writing
\begin{align*}
    \joint(\bz_i|\theta) &= \mathcal{N}\pr{\bz_i;\dfrac{\mu\rho^2+\theta\sigma^2}{\sigma^2+\rho^2},\dfrac{\rho^2\sigma^2}{\rho^2 + \sigma^2}}, \forall i \in [b]\\
    \joint(\theta|\bz_{1:b}) &= \mathcal{N}\pr{\theta;\bar{z},\dfrac{\rho^2}{b}}, \text{ where } \bar{z} \coloneqq \dfrac{1}{b}\sum_{i=1}^b \bz_i,
\end{align*}
we have 
\begin{align*}
    P_{\mathrm{SGS}} \coloneqq \mathrm{Pr}\pr{\theta^{[t]} | \theta^{[t-1]}} = \mathcal{N}\pr{\theta^{[t]};\dfrac{\sigma^2}{\sigma^2+\rho^2}\theta^{[t-1]} + \dfrac{\rho^2}{\sigma^2+\rho^2}\mu, \dfrac{2\rho^2\sigma^2 + \rho^4}{b(\rho^2 + \sigma^2)}}.
\end{align*}
By a straightforward induction, it follows that the Markov transition kernel $\nu P^t$ after $t$ iterations and with initial distribution $\nu$ has the form
\begin{align*}
    &\nu P_{\mathrm{SGS}}^t \coloneqq \mathrm{Pr}\pr{\theta^{[t]} | \theta^{[0]} \sim \nu} \nonumber \\ &=\mathcal{N}\pr{\theta^{[t]};\pr{\dfrac{\sigma^2}{\sigma^2+\rho^2}}^t \theta^{[0]} + \dfrac{\rho^2\mu}{\sigma^2+\rho^2}\sum_{i=0}^{t-1}\pr{\dfrac{\sigma^2}{\sigma^2+\rho^2}}^i, \dfrac{2\rho^2\sigma^2 + \rho^4}{b(\rho^2 + \sigma^2)}\sum_{i=0}^{t-1}\pr{\dfrac{\sigma^4}{(\sigma^2+\rho^2)^2}}^i}.
\end{align*}

\subsubsection{Splitting Strategy 2}

Similar calculus as in the above section can be undertaken by simply replacing $\rho^2$ by $\rho^2 b$.

\vskip 0.2in
\bibliography{bibliography.bib}

\end{document}